\documentclass[11pt]{article}
\usepackage{appendix}
\usepackage{mathtools}
\usepackage[T1]{fontenc}
\usepackage{amsfonts}
\usepackage{amsmath}
\usepackage{amssymb}
\usepackage{pgfplots}
\usepackage{amsthm}
\usepackage{amsmath,amssymb,bm,booktabs,graphicx,xcolor}
\usepackage{caption}
\usepackage{algorithm}
\usepackage{algpseudocode}
\usepackage{subcaption}
\usepackage{float}
\usepackage{microtype}
\usepackage{enumitem}
\usepackage{thmtools}
\usepackage{bbm}
\usepackage{bm}
\usepackage{mathrsfs}
\usepackage{color}
\usepackage{pdfsync}
\usepackage{enumitem}
\usepackage{mathtools}
\usepackage{xcolor} 
\definecolor{darkgreen}{rgb}{0,0.5,0} 
\definecolor{darkbkue}{rgb}{1,0,0} 
\usepackage{natbib}
\setcitestyle{authoryear,aysep={,}}  
\usepackage[colorlinks=true, linkcolor=darkgreen, citecolor=darkgreen, urlcolor=blue]{hyperref}
\makeatletter
\mathtoolsset{showonlyrefs}
\pgfplotsset{compat=1.18}
\renewcommand\@makefnmark{%
  \hbox{\textsuperscript{\textcolor{darkgreen}{\@thefnmark}}}}
\makeatother
\usepackage{tikz}

\usepackage[scaled]{helvet}
\DeclareMathOperator*{\argmin}{argmin}

\newcommand{\RR}{\mathbb{R}}
\newcommand{\R}{\RR}

\usepackage[margin=31 mm]{geometry}
\newcommand{\mykill}[1]{}

\usepackage[capitalize, nameinlink]{cleveref}

\newlist{hypotheses}{enumerate}{1}
\setlist[hypotheses]{
    label=\textcolor{darkgreen}{(H\arabic*)},
    ref=(H\arabic*)
}

\usepackage{amsfonts,bbm,bm}
\newcommand{\rom}[1]{%
  \textup{\uppercase\expandafter{\romannumeral#1}}%
}

\theoremstyle{plain}
\newtheorem{theorem}{Theorem}[section]
\newtheorem{proposition}[theorem]{Proposition}
\newtheorem{lemma}[theorem]{Lemma}

\theoremstyle{definition}
\newtheorem{definition}[theorem]{Definition}
\newtheorem{remark}[theorem]{Remark}

\theoremstyle{remark}

\newlist{myenum}{enumerate}{3}
\setlist[myenum,1]{label={\rm (H\arabic*)},
                   ref  ={\rm (H\arabic*)}}
\crefname{myenumi}{property}{properties}

\renewenvironment{thebibliography}[1]{%
\begin{oldthebibliography}{#1}%
\setlength{\baselineskip}{.9em}
\linespread{1}
\small
\setlength{\parskip}{.40ex}%
\setlength{\itemsep}{.30em}%
}%
{%
\end{oldthebibliography}%
}

\definecolor{grape}{rgb}{0.43, 0.17, 0.71}

\begin{document}

\title{ Huber–Wasserstein barycenters for robust distribution-valued data}
\date{}
\author{  
  Carlos Cardoso-Perelló\thanks{Department of Statistics, Columbia University, \textcolor{darkgreen}{cc5415@columbia.edu}.}
\and Alberto Gonz{\'a}lez-Sanz%
  \thanks{Department of Statistics, Columbia University, \textcolor{darkgreen}{ag4855@columbia.edu}.}  
  }
  
\maketitle 
\vspace{-1.5em}
\begin{abstract}
We propose a robust barycenter for distribution-valued data by incorporating the Huber loss directly into the optimal transport cost. In contrast to metric-space Huber means, which apply the Huber loss to the Wasserstein distance after optimization, our construction acts on individual transport displacements, preserving quadratic behavior locally while limiting the influence of large displacements. The resulting Huber--Wasserstein barycenters form a natural interpolation between Wasserstein means and $L^1$-type Wasserstein medians.

We establish the analytical and statistical foundations of this construction. For optimal transport with Huber loss, we prove regularity and uniqueness properties of dual potentials, existence of optimal transport maps, and stability as the Huber parameter varies. For the associated barycenter problem, we prove existence and characterization results, consistency of empirical plug-in estimators, and a finite-sample breakdown point essentially equal to $1/2$. In dimension one, we further derive the pointwise influence function and
asymptotic distribution, quantify the associated robustness--efficiency
trade-off, and show that displacement-wise Huberization can retain
first-order information that is lost by distance-based Huberization under
localized shape contamination. Numerical experiments on contaminated distribution-valued data demonstrate the robustness of the proposed barycenters and illustrate their interpolation between mean- and median-like behavior.
\end{abstract}

 \vspace{1em}

{\small

\noindent \emph{MSC2020 subject classifications. Primary: } 62G35, 62G30.  
}
 \vspace{.2em}
 \section{Introduction}
 Modern statistical problems increasingly involve data objects that do not naturally belong to Euclidean spaces. Among the most important examples is distribution-valued data, where each observation is itself a probability measure rather than a scalar or vector. Such data arise in a wide range of applications, including mortality and age-at-death distributions, regional income and house-price distributions \cite{chen2023wasserstein}, neuroimaging intensity distributions \cite{petersen2016functional,petersen2021wasserstein}, particle-size and sediment-size distributions in geoscience \cite{vandenboogaart2014bayes}, and collections of histograms, point clouds, or empirical measures obtained from repeated samples \cite{bigot2019statistical,bachoc2023improved,szabo2016learning}. Distributional representations also appear in modern machine-learning and AI pipelines, where texts, images, and media objects may be encoded or compared through probability distributions \cite{ghorbani2020distributional,chan2022data}. These examples motivate the development of statistical methods that respect the intrinsic geometry of spaces of probability measures.

A central task in distributional data analysis is to define a meaningful notion of center, or mean, for a probability measure ${\bf P}$ on the space of probability distributions. The standard approach is to use a Fr\'echet mean, which requires a suitable metric between probability measures. Let $\mathcal{P}_p(\mathbb{R}^d)$ denote the set of probability measures on $\mathbb{R}^d$ with finite moments of order $p$.  On $\mathcal{P}_p(\mathbb{R}^d)$, the canonical choice is the $p$-Wasserstein distance
\begin{equation}
    \label{eq:Wasserstein-distance}
    \mathcal{W}_p(\mu,\nu)
    =
    \left(
    \inf_{\pi \in \Pi(\mu,\nu)}
    \int \|x-y\|^p\,d\pi(x,y)
    \right)^{\frac{1}{p}},
\end{equation}
where $\Pi(\mu,\nu)$ is the set of probability measures on $\mathbb{R}^d \times \mathbb{R}^d$ with marginals $\mu$ and $\nu$.  The Wasserstein distance is particularly well suited for distributional data because it takes into account the geometry of the underlying space. In particular, it endows $\mathcal{P}_p(\mathbb{R}^d)$ with a rich structure, which is in several aspects analogous to Euclidean geometry; see \cite{Panaretos.Zemel.2019.ARS}.

Let us focus on $p=2$. Under the $   \mathcal{W}_2$ metric, the Fr\'echet mean of a probability measure ${\bf P}$ on $\mathcal{P}_2(\mathbb{R}^d)$ is the Wasserstein barycenter, defined by
\begin{equation}
    \label{eq:Wasserstein-bary}
    m({\bf P})
    :=
    \argmin_{\nu\in \mathcal{P}_2(\mathbb{R}^d)}
    \int \mathcal{W}_2^2(\mu,\nu)\,d{\bf P}(\mu).
\end{equation}
Since their introduction in \cite{Agueh.Carlier.2011.SIMA}, Wasserstein barycenters have been studied extensively. Existence results were established for discrete ${\bf P}$ in \cite{Agueh.Carlier.2011.SIMA}, and for general ${\bf P}$ in \cite{LeGouic.Loubes.2017.PTRF}. Computational methods for Wasserstein barycenters have been developed in
\cite{AlvarezEsteban.delBarrio.2016.JMAA,Cuturi.Doucet.2014.ICML,
Benamou.Carlier.2015.SISC}.
Their statistical behavior and computational complexity have been analyzed in
\cite{Bigot.Klein.2018.ESAIMPS,Kroshnin.Tupitsa.2019.ICML,
Altschuler.BoixAdsera.2021.JMLR,Heinemann.Munk.2022.SIMODS}.
For broader background and applications, we refer to
\cite{Peyre.Cuturi.2019.FTML,Panaretos.Zemel.2019.ARS}.

Despite their theoretical and practical appeal, Wasserstein barycenters can be highly
sensitive to outlying distributions. This issue has motivated several robust approaches in
Wasserstein space. An early line of work is based on impartial trimming:
\cite{AlvarezEsteban.etal.2008.TrimmedComparison} introduced robust comparisons
between probability distributions by optimally trimming mass within the distributions before
computing their Wasserstein discrepancy. Building on this principle at the level of
distribution-valued data, \cite{AlvarezEsteban.etal.2018.WideConsensus} proposed
trimmed Wasserstein barycenters for robust consensus aggregation of probability measures
and established existence and consistency results for the resulting robustified Fr\'echet means.
Similar ideas were developed in \cite{delBarrio.etal.2019.RobustClustering}, which
introduced trimmed \(k\)-barycenters for robust clustering of probability measures. In these latter approaches, robustness is achieved by trimming entire distribution-valued observations that are sufficiently discrepant from the consensus.

Related robust notions of center are provided by Wasserstein medians; see
\cite{Carlier.Chenchene.2024.SIMA} for \(\mathcal{W}_1\)-type Wasserstein medians and
\cite{You.Shung.Giuffre.2025} for Fr\'echet medians with respect to the \(\mathcal{W}_2\) metric. More recently,
\cite{Lee-Jung.2026.JRRSS-b} proposed a Huber-type interpolation between Fr\'echet means and
medians on general metric spaces. Applied to the Wasserstein space, their construction
amounts to placing the Huber loss outside the optimal transport problem, leading to
\begin{equation}
    \label{eq:Wasserstein-bary-huber-BAD}
    m_{H,c}'(\mathbf{P})
    :=
    \argmin_{\nu\in \mathcal{P}_2(\mathbb{R}^d)}
    \int \rho_c\big(\mathcal{W}_2(\mu,\nu)\big)\,d{\bf P}(\mu),
\end{equation}
where $\rho_c$ is the Huber loss with parameter $c>0$, defined as
\begin{equation}
    \label{eq:Huber}
    \rho_c(t)
    :=
    \begin{cases}
        \frac{1}{2}t^2, & \text{if } |t|<c,\\
        c|t|-\frac{c^2}{2}, & \text{if } |t|\geq c.
    \end{cases}
\end{equation}
The Huber loss behaves quadratically near the origin and linearly for large values,
thereby reducing the contribution of extreme observations.  Consequently,
\eqref{eq:Wasserstein-bary-huber-BAD} provides a natural robust alternative to the
classical Wasserstein barycenter. However, in general Wasserstein spaces, the map
$\nu \mapsto \rho_c\!\bigl(\mathcal{W}_2(\mu,\nu)\bigr)$
does not inherit the same convexity structure as the squared-Wasserstein objective.
This creates additional theoretical and computational difficulties, including the possible
presence of multiple local minimizers and a more delicate optimization landscape. The
one-dimensional case is exceptional, since the quantile representation restores convexity. 

We propose a different robustification strategy; instead of applying the Huber loss to the Wasserstein distance itself, we incorporate the Huber loss at the level of the optimal transport cost. Specifically, we define
\begin{equation}
    \label{eq:Wasserstein-bary-huber}
    m_{c}(\mathbf{P})
    :=
    \argmin_{\nu\in \mathcal{P}_1(\mathbb{R}^d)}
    \Gamma_{c,{\mathbf{P}}}(\nu),
    \qquad
    \Gamma_{c,{\mathbf{P}}}(\nu)
    :=
    \int
    \left(
    \mathcal{T}_{\rho_c}(\nu,\mu)
    -
    c\int \|x\|\,d\mu(x)
    \right)
    d{\mathbf{P}}(\mu),
\end{equation}
where
\begin{equation}
    \label{eq:primal-Huber-intro}
    \mathcal{T}_{\rho_c}(\mu,\nu)
    :=
    \inf_{\pi \in \Pi(\mu,\nu)}
    \int \rho_c(\|x-y\|)\,d\pi(x,y)
\end{equation}
is the optimal transport cost associated with the Huber loss. The centering term \(c\int \|x\|\,d\mu(x)\) ensures that
\(\Gamma_{c,P}(\nu)\) is finite for every
\(P\in\mathcal P(\mathcal P_1(\mathbb R^d))\), without requiring
any integrability condition on the first moments of the random measure
\(\mu\sim \mathbf{P}\).

\begin{figure}[H]
\centering
\includegraphics[width=0.98\linewidth]{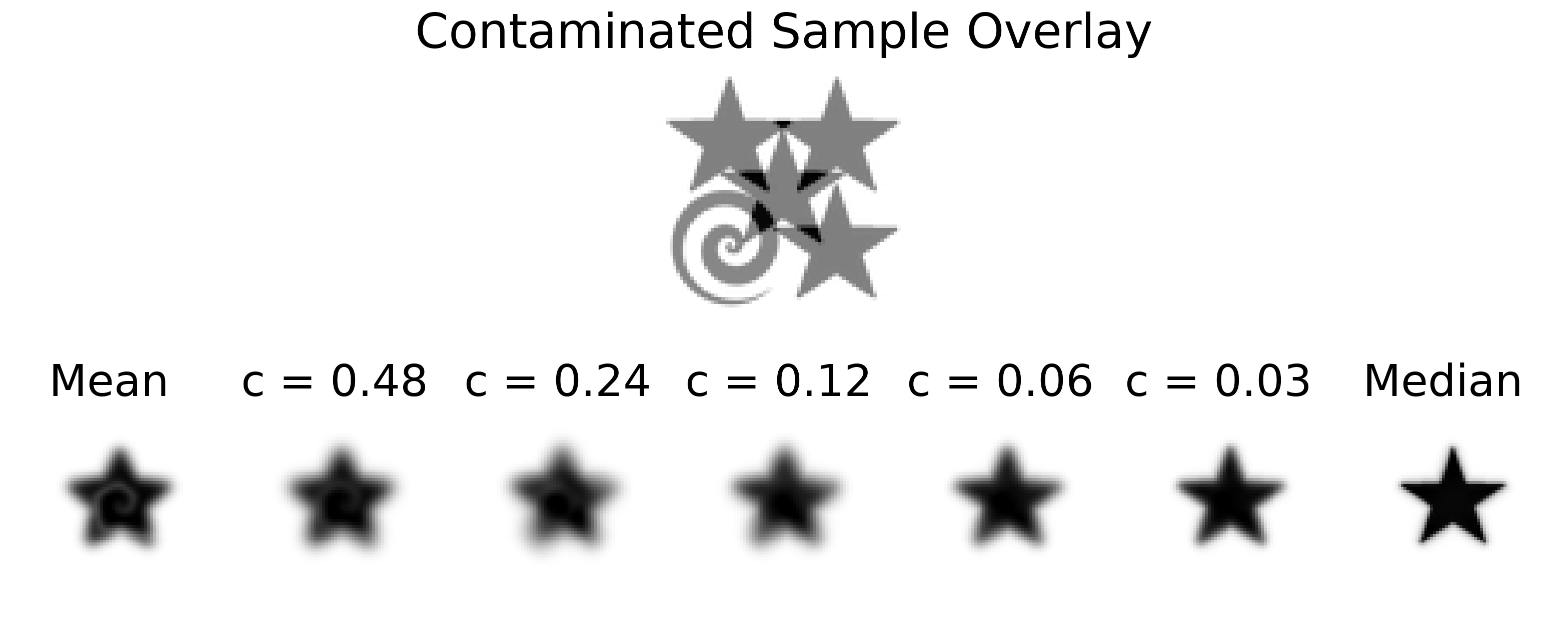}
\caption{Star--spiral contamination experiment. The Huber barycenters of the stars and the spirals interpolate between the Wasserstein mean and median as $c$ decreases to $0$.}
\label{fig:star_intro}
\end{figure}

This construction differs fundamentally from the metric-space Huber mean in
\eqref{eq:Wasserstein-bary-huber-BAD}. By modifying the transportation cost itself,
rather than the Wasserstein distance after optimization, the Huber loss acts locally on
individual transport displacements. The resulting objective $\Gamma_{c,{\bf P}}$ remains
convex in the standard linear geometry of probability measures, while limiting the
contribution of large transport displacements. We further show that the associated empirical
Huber--Wasserstein barycenter has a finite-sample breakdown point essentially equal to
$1/2$. Thus, our construction combines the convex structure of transport-based barycenter
problems with the soft-clipping principle of classical Huber robustification.

The proposed construction also connects naturally with existing notions of centrality in Wasserstein space: as $c\to\infty$, the Huber cost approaches the quadratic cost and the corresponding barycenter recovers the classical Wasserstein barycenter, whereas as $c\to 0$, after the appropriate normalization, it approaches the $L^1$-type Wasserstein median studied in \cite{Carlier.Chenchene.2024.SIMA}. This behavior is illustrated in \Cref{fig:star_intro}. More generally, barycenters associated with nonquadratic transportation costs have also
been investigated. \cite{Brizzi.Friesecke.Ried.NonLin.2025} study
\(p\)-Wasserstein barycenters, while
\cite{Brizzi.Friesecke.Ried.2026.hWasserstein} develop a barycenter theory for general
smooth strictly convex ground costs with nondegenerate Hessian. The Huber cost falls
outside the latter framework: it is only strictly convex in its quadratic region and becomes
affine in the radial variable beyond the cutoff \(c\). Thus, its large-displacement regime is
precisely a degenerate case not covered by the strictly convex \(h\)-Wasserstein theory.

\paragraph{Main contributions}  First, we develop the optimal-transport theory
associated with the Huber ground cost. Despite the loss of strict convexity in its
linear regime, we establish regularity and stability of Kantorovich potentials,
uniqueness under suitable conditions, existence of Monge solutions for absolutely
continuous source measures, and quantitative convergence toward quadratic optimal
transport as \(c\to\infty\). Second, we develop a population and statistical theory for
the resulting Huber--Wasserstein barycenters on
\(\mathcal P_1(\mathbb R^d)\): we prove existence and convexity, establish consistency
under both one-stage and two-stage sampling, and characterize the minimizers through
averaged Kantorovich potentials. Third, we show that the robustification parameter provides a genuine
interpolation between two transport geometries: as $c\downarrow 0$ the
barycenters converge to $W_1$-type Wasserstein medians, whereas as
$c\uparrow\infty$ they converge to classical quadratic Wasserstein
barycenters. Fourth, we establish both global and local robustness
properties. The empirical Huber--Wasserstein barycenter has a finite-sample
breakdown point essentially equal to $1/2$. In dimension one, we derive an
explicit quantile characterization, obtain its pointwise influence function
and asymptotic distribution, and quantify the corresponding
robustness--efficiency trade-off. We further show that, under localized
shape contamination, displacement-wise Huberization can retain first-order
information that is lost when a single Huber weight is assigned to the
distribution as a whole.

\paragraph{Related work}
Robustness of optimal transport discrepancies  has been pursued through several
distinct mechanisms; see \cite{Ronchetti.2023} for an overview from the perspective of
classical robust statistics. One early approach is to modify the ground cost. In computer
vision, \cite{Pele.Werman.2009} introduced Earth Mover's Distances with thresholded
ground distances, showing that truncation improves robustness to outlier noise and
quantization effects. More recently, \cite{Mukherjee.etal.2021} developed the ROBOT
framework for outlier-robust optimal transport. Their formulation penalizes perturbations of a marginal in total variation and admits an equivalent optimal transport formulation with truncated ground cost
$C_\lambda(x,y)=\min\{C(x,y),2\lambda\}.$
Thus, transportation costs above the threshold are completely saturated. Building on this
framework, \cite{Ma.Liu.LaVecchia.Lerasle.2025} develop a statistical theory for robust
optimal transportation, including measure-theoretic properties, concentration inequalities,
minimum-distance estimation, and a data-driven procedure for selecting the truncation
parameter.

Other approaches robustify different components of the transport problem.
\cite{Staerman.etal.2021} exploit the Kantorovich--Rubinstein dual representation of
\(\mathcal{W}_1\) and replace empirical expectations by Median-of-Means estimators.
\cite{Balaji.Chellappa.Feizi.2020} relax the marginal constraints through an unbalanced
optimal transport formulation, thereby allowing observations identified as outliers to
receive less mass. \cite{Nietert.Goldfeld.Cummings.2022} instead introduce a
partial-transport-based Wasserstein discrepancy that allows a prescribed amount of mass
to be removed and establish duality, structural properties, and statistical robustness
guarantees under Huber contamination.

From the viewpoint of classical robust statistics, the truncated ground cost used in
ROBOT corresponds to a hard-rejection mechanism. As observed by
\cite{Ronchetti.2023}, when the original transportation cost is quadratic, this
hard truncation may entail a substantial loss of efficiency at the model.
Ronchetti explicitly suggests a Huber loss as a potentially preferable alternative,
while noting that an analogous robustness theory was not available for this case.
In classical \(M\)-estimation terms, the distinction is that Huberization clips the
score rather than saturating the loss itself.

This soft-clipping principle is precisely
the one pursued in our construction. Rather than bounding the transportation loss, we replace the quadratic ground
cost by the Huber loss. Consequently, large transport displacements are not rejected:
their marginal contribution is clipped,
$
\rho_c'(r)=\min\{r,c\},$
while their transportation cost continues to grow linearly. This distinction is important
both statistically and analytically. Unlike a hard-truncated quadratic cost, the Huber cost
is convex and continuously differentiable, retains the quadratic geometry for small
displacements, and transitions continuously to a linear transport regime for large
displacements.

Related robust barycenter formulations based on robustified optimal transport have also
been considered. \cite{Le.etal.2021} study a fixed-support robust barycenter obtained by
relaxing marginal constraints through a Kullback--Leibler penalty, and
\cite{Wang.etal.2024} develop computational methods for fixed- and free-support variants.
More recently, \cite{Cheng.Liu.2026} construct a robust Wasserstein barycenter from the
hard-truncated ground metric
$d_\lambda(x,y)=\min\{d(x,y),\lambda\}$
and establish existence and consistency. Our proposal differs from these methods in using
an unbounded but linearly growing convex Huber cost rather than trimming entire
distribution-valued observations, relaxing or deleting transport mass, or saturating the
ground cost. This structure leads to a continuous interpolation between quadratic and
linear optimal transport and enables us to study the associated dual potentials, optimal
transport maps, barycenter stability, and finite-sample breakdown behavior.

A separate line of research concerns the robustness of the optimal transport map itself. \cite{Avella-Gonzalez.25,paindaveine:passeggeri2024} study the breakdown point of transport-based quantiles by treating the optimal transport map from a fixed reference distribution to a target distribution as a statistical functional of the latter. Their results relate the breakdown point of the map evaluated at a reference point to the halfspace depth of that point. This connection has subsequently been extended to broad classes of convex transportation costs in \cite{GonzalezSanz.AvellaMedina.2026.GeneralCosts}, showing in particular that the breakdown behavior of these transport maps is largely insensitive to the choice of regular ground cost. Complementarily, \cite{GonzalezSanz.Sheng.Wu.AvellaMedina.2026.Influence} investigate local robustness through the influence function of transport-based quantiles. The statistical object considered here is different. Our observations are themselves probability measures, so contamination occurs at the level of a population of distributions, and the object whose robustness we study is their barycenter rather than the value of a transport map at a fixed reference point. In particular, our finite-sample breakdown result quantifies how many entire distribution-valued observations must be replaced before the set of Huber--Wasserstein barycenters can be driven arbitrarily far away.

\paragraph{Organization of the paper.} The remainder of the paper is organized as follows. In \Cref{sect:stability-properties-cost}, we develop the theory of optimal transport with Huber cost, including Kantorovich duality, the regularity and stability of dual potentials, uniqueness of potentials, the existence of optimal transport maps, and the convergence of the Huber transport cost and maps to their quadratic counterparts as $c\to\infty$. In \Cref{sect:Barycenters}, we introduce the Huber--Wasserstein barycenter and establish existence, stability, and statistical consistency under one- and two-stage sampling. We further analyze its limiting behavior as $c\downarrow 0$ and $c\uparrow\infty$, characterize its minimizers in terms of averaged Kantorovich potentials, and determine its finite-sample breakdown point. \cref{sect:numerics} presents numerical experiments illustrating the interpolation between the Wasserstein barycenter and the Wasserstein median in several contaminated distributional-data settings. The appendix contains the proofs omitted from the main text (\hyperref[app:proofs]{Appendix A}), together with computational details for the Huber--Sinkhorn algorithm (\hyperref[app:huber-sinkhorn-barycenter]{Appendix B}).
\section{Optimal transport with Huber loss: stability and basic properties}\label{sect:stability-properties-cost}
In this section we study the optimal transport problem with Huber loss: 
\begin{equation}
    \label{eq:primal-Huber}
    \mathcal{T}_{\rho_c}(\mu,\nu)
    :=
    \inf_{\pi \in \Pi(\mu,\nu)}
    \int \rho_c(\|x-y\|)\,d\pi(x,y)
\end{equation}
In \Cref{sect:duality} we state the dual problem to \eqref{eq:primal-Huber}. In \Cref{sect:stability-potential} we study the stability and Lipschitz properties of the dual potentials; in particular, \Cref{lemma:Lipschitz-extension} shows that the dual potentials are $c$-Lipschitz, from which we infer in \Cref{theorem:stability-cost} that 
$$  |\mathcal{T}_{ \rho_c}(\mu_1,\nu) - \mathcal{T}_{ \rho_c}(\mu_2,\nu)|\leq c \cdot  \mathcal{W}_1(\mu_1,\mu_2).$$
\Cref{sect:potentials-and-map} shows that the dual potentials are unique under mild assumptions on the measures $\mu$ and $\nu$. \Cref{theorem:Huber-maps} shows that there exists a solution of \eqref{eq:primal-Huber} concentrated on the graph of a function, in the spirit of \cite{Cuesta1989NotesOT, brenier1991polar}. \Cref{propositon:Stable-cost-quantitative} and \cref{theorem:stability_map} show a quantitative result on the stability as $c\to \infty$ of the $ \mathcal{T}_{\rho_c}(\mu,\nu)$ and the OT map. The proofs of all the results of this section can be found in \hyperref[app:proofs]{Appendix A}. 

\subsection{Duality}\label{sect:duality}
A function $f \colon \mathbb{R}^d \to \mathbb{R} \cup \{-\infty\}$ is called \textit{$g$-concave} if it can be expressed as
\begin{align}\label{CLTgeneq:c_concave}
f(x) = \inf_{(y, t) \in \mathcal{T}} \big\{ g(\|x-y\|) - t \big\},
\end{align}
for some set $\mathcal{T} \subset \mathbb{R}^d \times \mathbb{R}$.  For such a function, its \textit{$g$-superdifferential} $\partial^{g} f$ is defined as the set of pairs $(x, y) \in \mathbb{R}^d \times \mathbb{R}^d$ satisfying
\begin{equation}
    \label{eq:def-superdiff}
    \partial^{g} f:=\left\{(x,y)\in \mathbb{R}^d\times \mathbb{R}^d:\forall z \in \mathbb{R}^d,\, f(z)-f(x) \leq  g(\|z-y\|) -  g(\|x-y\|) \right\}.
\end{equation}
We use the notation $
\partial^{g} f(x):=\{ y: (x,y)\in \partial^{g} f\}. $
The \emph{$g$-conjugate of $f$}  is defined as
\begin{align}\label{CLTgeneq:c_conj}
f^{g}({y}):=\inf_{{x}\in {\mathbb{R}^d}}\left\{ g (\|x-y\|)-f({x})\right\},\,\forall y\in \mathbb{R}^d.
\end{align}
If $f:\R^d\rightarrow \R\cup \{-\infty\}$ is $g$-concave, then it follows that
$ f(x) + f^g(y)\leq g(\|x-y\|) $ with equality if and only if $y\in \partial^{ g} f(x)$ (cf.~\citealp[Proposition~3.3.7]{RachevRueschendorf1998}). The domains of $ f$ and $\partial^g f$ are denoted as 
$$ {\rm dom}(f)= \{ x: f(x) \in \R\}\quad {\rm and}\quad {\rm dom}(\partial^{ g} f)= \{ x: \partial^{ g} f(x) \neq \emptyset \} . $$
Choosing $g=\rho_c$, we obtain the following result whose proof can be found in \cite{Villani.09}. 

\begin{theorem}\label{theorem:duality}
    Assume that $\mu,\nu\in \mathcal{P}_1(\R^d)$. Then  
    \begin{enumerate}
    \item (Existence of plans) there exists $\pi_{*} \in \Pi(\mu,\nu)$ solving \eqref{eq:primal-Huber}. 
        \item  (Strong duality) It follows that
        \begin{equation}\label{eq:Huber-dual}
\mathcal{T}_{\rho_c}(\mu,\nu)
=
\sup_{\substack{
f(x)+g(y) \le \rho_c(\|x-y\|)\\
f\in L^1(\mu),\; g\in L^1(\nu)
}}
\left\{
\int f\,d\mu + \int g\,d\nu
\right\}.
\end{equation}
       \item (Dual solutions) There exists a pair $(\psi,\varphi)$ solving  \eqref{eq:Huber-dual}. Furthermore, any pair  $(\psi,\varphi)$ solving \eqref{eq:Huber-dual} consists of $\rho_c $-conjugate, $\rho_c $-concave functions such that 
       \begin{equation}
           \label{eq:Dual-Primal-relation}
           \psi(x)+\varphi(y) = \rho_c(\|x-y\|) , \quad \text{for } \pi_*\text{-a.e.\ } (x,y)\in \R^{d}\times \R^d, 
       \end{equation}
       for some $\pi_{*} \in \Pi(\mu,\nu)$ solving \eqref{eq:primal-Huber}. 
       \item (Cyclically monotone support) A plan $\pi_{*} \in \Pi(\mu,\nu)$ solves \eqref{eq:primal-Huber} if and only if ${\rm supp}(\pi_*)$ is $\rho_c$-cyclically monotone. 
    \end{enumerate}
\end{theorem}
Any function $\psi$ for which there exists another function $\varphi$ such that $(\psi,\varphi)$ solves \eqref{eq:Huber-dual} is called a {\it (Huber) dual potential}. We write ${\rm Sol}^*(\nu,\mu)$ for the set of dual potentials $\psi$ between $\nu$ and $\mu$, and ${\rm Sol}_2^*(\nu,\mu)$ for the set of pairs $(\psi,\varphi)$.
\subsection{Stability and Lipschitz properties of dual solutions}\label{sect:stability-potential}
In the following result we use the fact that 
\begin{equation}
    \label{eq:Huber-loos-is-lip}
     |\rho_c(t)- \rho_c(s)| \leq c |t-s| 
\end{equation}
to guarantee that the dual solutions are Lipschitz with Lipschitz constant $c$. 
\begin{lemma}\label{lemma:Lipschitz-extension}
   Let \(f:\mathbb R^d\to\mathbb R\cup\{-\infty\}\) be a proper
\(\rho_c\)-concave function, i.e.~\(\operatorname{dom}(f)\neq\varnothing\). Then 
\begin{enumerate}
    \item  ${\rm dom}( f )=\R^d$.
    \item $f$ is Lipschitz in $\R^d$ with constant $c$.
\end{enumerate}
As a consequence, for $\mu,\nu\in \mathcal{P}_1(\R^d)$ any pair $(\psi,\varphi)$ solving \eqref{eq:Huber-dual}  is formed by $c$-Lipschitz functions. 
\end{lemma}

Now we state some results regarding the stability of the $\rho_c$-conjugation.  We start with two technical lemmas. We denote by $\mathbb{B}_r(x)$ the ball of radius $r$ centered at $x$.
\begin{lemma}\label{lemma:sup-c-conjugate-reestriced}
    Let $f:\R^d\to \R$ be $c$-Lipschitz. Then 
    $$ f^{\rho_c}(y)= \inf_{x\in \overline{\mathbb{B}}_c(y)}\{ \rho_c(\|x-y\|)-f(x)\} = \min_{x\in \overline{\mathbb{B}}_c(y)}\{ \rho_c(\|x-y\|)-f(x)\} , 
    $$
where the infimum is attained. 
\end{lemma}

\begin{lemma}\label{lemma:cojugate-converges}
    Let $\{f_n\}$ be a sequence of $\rho_c$-concave functions converging uniformly on compact sets to $f$. Then $f_n^{\rho_c}$ converges uniformly on compact sets to $f^{\rho_c}$. 
\end{lemma}

We can now easily show the following stability results. 
\begin{theorem}[Stability of the cost]\label{theorem:stability-cost}
    Let $\mu_1,\mu_2,\nu\in \mathcal{P}_1(\R^d)$
    $$  |\mathcal{T}_{ \rho_c}(\mu_1,\nu) - \mathcal{T}_{ \rho_c}(\mu_2,\nu)|\leq c \cdot  \mathcal{W}_1(\mu_1,\mu_2).$$
\end{theorem}

\begin{theorem}[Stability of the potentials]\label{theorem:stsbilitypot}
    Fix $\{\mu_n\}_n, \{\nu_n\}_n \subset  \mathcal{P}_1(\R^d)$, and assume that there exist $\mu,\nu\in  \mathcal{P}_1(\R^d)$ such that  $\mathcal{W}_1(\mu_n, \mu) \to 0$ and $\mathcal{W}_1(\nu_n, \nu) \to 0$. Let $f_n$ be a sequence of dual potentials for $(\mu_n,\nu_n)$ such that there is a point $x_0\in \mathbb{R}^d$ so that
    \begin{equation}
        \label{eq:bounded_at_a_point}
    \sup_n |f_n(x_0)|\leq C.
    \end{equation}
    Then $f_n$ admits a further subsequence $f_{n_k}$ converging to $f_\infty$ uniformly on the compact sets, where $f_\infty$ is $c$-Lipschitz and is a dual potential for $(\mu,\nu)$
\end{theorem}
\subsection{Uniqueness of the potentials and existence of transport maps}\label{sect:potentials-and-map}
First, we highlight that under the standard assumptions on $\mu$, the solutions of \eqref{eq:Huber-dual} are unique up to additive constants. 
The following result is a direct consequence of \cite[Corollary~2]{Staudt.et.al.2025.SIMA}. Note that the same result does not apply to the Euclidean distance, as it is not differentiable at zero.
\begin{theorem}[Uniqueness of potentials]\label{theorem:uniqueness-potentials}
Set $\mu, \nu\in \mathcal{P}_1(\R^d)$, and assume that $\mu\ll \mathcal{L}_d$ is such that $\mu(\partial {\rm supp}(\mu))=0$ and ${\rm supp}(\mu)$ is connected. Define the equivalence relation $\sim$ on pairs of functions $(f,g)\in L^1(\mu)\times L^1(\nu)$ as
    $$
    (f_1,g_1)\sim (f_2,g_2)\iff \exists a\in \mathbb{R}: f_1\equiv f_2+a, g_1\equiv g_2-a.
    $$
    Then ${\rm Sol}_2^*(\nu,\mu)/\sim$ is a singleton. 
\end{theorem}
For costs of the form $c(x,y)=h(\|x-y\|)$, the existence and structure of optimal plans are highly dependent on the convexity of $h$. When $h$ is strictly convex and has superlinear growth, the optimal transport problem is known to admit a unique solution, and, under mild assumptions on the reference measure $\mu$, this solution is induced by a transport map; see, for instance, \cite{Gangbo.McCann.1996.Acta}.  The situation is more involved when the cost is convex but not strictly convex. In this case, uniqueness of optimal plans may fail even in one dimension, and the existence theory for Monge solutions in arbitrary dimension becomes substantially more subtle.

The singular linear cost $c(x,y)=\|x-y\|$ represents a particularly important degenerate case and has been studied extensively; see, for example, \cite{Caffarelli.Feldman.2002.JAMS,Evans.Gangbo.1999.MemoirsAMS,Ambrosio.Pratelli.2003.LNM,Champion.dePascale.2011.DMJ}. The Huber cost
$    c(x,y)=\rho_c(\|x-y\|)$
lies between the strictly convex and linear regimes: it is quadratic near the origin, but becomes linear at infinity, combining strictly convex behavior at small displacements with degenerate linear behavior at large displacements. To the best of our knowledge, the existence of optimal transport maps for this specific cost has not been addressed in the literature.

In the sequel, we call a plan $\pi\in \Pi(\mu,\nu)$ attaining the infimum in \eqref{eq:primal-Huber} a \textit{Huber plan}. The following result concerns the existence of such optimal transport maps.
\begin{theorem}\label{theorem:Huber-maps}
    Assume that $\mu,\nu\in \mathcal{P}_1(\R^d)$ with $\mu\ll \mathcal{L}_d$. Then there exists  a Huber plan   $\pi_{H_c}\in \Pi(\mu,\nu)$ of the form $\pi_{H_c}=(I\times H_c)_\# \mu$, where $H_c$ is the Huber OT map which pushes $\mu$ forward to $\nu$. 
\end{theorem}
The proof of this result is relegated to the supplementary material. We now outline its main steps. Let \(\pi_c\) be an optimal Huber transport plan and decompose it according to whether \(\|x-y\|\le c\) or \(\|x-y\|>c\). On the first region, the Huber cost is quadratic, so the corresponding part of \(\pi_c\) is quadratically cyclically monotone and is induced by the gradient of a convex function. On the second region, each displacement is truncated at distance \(c\), which yields an \(L^1\)-optimal transport plan calibrated by a \(1\)-Lipschitz potential. A measurable selection along the associated transport rays provides a map with the same marginals and no larger linear cost. The original residual target is then recovered by translating the selected map by \(c\) along each oriented ray. Finally, the quadratic and residual maps are pasted together, and the resulting map is shown to push \(\mu\) forward to \(\nu\) while attaining the optimal Huber cost.

\subsection{Relation with the 2-Wasserstein distance}
We examine the relationship between the optimal transport problem with Huber loss and that with quadratic loss. The following result gives sufficient conditions so that both problems have the same optimal transport map
\begin{proposition}
\label{prop:ot_is_huber_ot}
Fix \(c>0\) and let \(\mu,\nu\in\mathcal P_2(\mathbb R^d)\). The following
statements are equivalent:
\begin{enumerate}
    \item There exists
    \begin{equation}
    \label{eq:OT-plan}
    \pi_*
    \in
    \argmin_{\pi\in\Pi(\mu,\nu)}
    \int \frac12\|x-y\|^2\,d\pi(x,y)
    \end{equation}
    which is also Huber optimal and satisfies
    \(\pi_*(\{(x,y):\|x-y\|\le c\})=1\).

    \item \(\frac12\mathcal W_2^2(\mu,\nu)=\mathcal T_{\rho_c}(\mu,\nu).\)

    \item Every solution of \eqref{eq:OT-plan} is concentrated on
    \(\{(x,y):\|x-y\|\le c\}\) and is also Huber optimal.

    \item There exists a Huber optimal plan concentrated on
    \(\{(x,y):\|x-y\|\le c\}\).
\end{enumerate}
\end{proposition}

\begin{remark}
One might wonder whether it is sufficient to assume only that there exists a
quadratic optimal plan concentrated on
\(\{(x,y):\|x-y\|\le c\}\). This is not true in general.
Indeed, let \(c=1\), \(d=2\), \(x_1=y_2=(0,0)\), \(y_1=(1,0)\), and
$x_2=\left(-\frac1{12},\frac{\sqrt{143}}{12}\right).$ 
Then \(\|x_2\|=\|y_1\|=1\) and
\(\|x_2-y_1\|^2=13/6\). Define $\mu=\frac12(\delta_{x_1}+\delta_{x_2})$ and $
\nu=\frac12(\delta_{y_1}+\delta_{y_2})$. The two extreme points of the transport polytope are  
  $\pi_*=
\frac12\delta_{(x_1,y_1)}
+
\frac12\delta_{(x_2,y_2)}$ and  $\widetilde\pi
=
\frac12\delta_{(x_1,y_2)}
+
\frac12\delta_{(x_2,y_1)}$. 
The plan \(\pi_*\) is concentrated on \(\{\|x-y\|\le1\}\), and its
quadratic cost is
$\int\frac12\|x-y\|^2\,d\pi_*=\frac12.$
Since 
\[
\int\frac12\|x-y\|^2\,d\widetilde\pi
=
\frac14\cdot   \frac{13}{6}
>
\frac12, 
\]
it follows that \(\pi_*\) is the unique quadratic optimal plan. Nevertheless,
$\int\rho_1(\|x-y\|)\,d\pi_*=\frac12,$
whereas
\[
\int\rho_1(\|x-y\|)\,d\widetilde\pi
=
\frac12\left(\sqrt{\frac{13}{6}}-\frac12\right)
<
\frac12.
\]
Hence \(\pi_*\) is not Huber optimal. 
\end{remark}
Now we study the stability of the problem when $c\to \infty$. Here $(t)_+=\max(0,t).$ We start with the stability of the cost. 
\begin{proposition}
\label{propositon:Stable-cost-quantitative}
Assume that \(\mu,\nu\in\mathcal P_2(\mathbb R^d)\). Then
\[
0
\le
\frac12\mathcal W_2^2(\mu,\nu)-\mathcal T_{\rho_c}(\mu,\nu)
\le
\int\left(\|x\|^2-\frac{c^2}{4}\right)_+\,d\mu(x)
+
\int\left(\|y\|^2-\frac{c^2}{4}\right)_+\,d\nu(y).
\]
Consequently,
\(\mathcal T_{\rho_c}(\mu,\nu)\to\frac12\mathcal W_2^2(\mu,\nu)\) as
\(c\to\infty\).
\end{proposition}
The following result uses the quadratic detachment of the $L^2$-cost to derive the stability of the maps. 

\begin{theorem}[Stability of the OT map]\label{theorem:stability_map}
Assume that \(\mu,\nu\in\mathcal P_2(\mathbb R^d)\), with
\(\mu,\nu\ll\mathcal L^d\). Let \(H_c\) be a Huber OT map and let \(H_0\)
denote the quadratic OT map pushing \(\mu\) forward to \(\nu\). If \(H_0\) is
\(\alpha\)-Lipschitz, then
\[
\|H_0-H_c\|_{L^2(\mu)}
\le
\left[
2\alpha
\left\{
\int\left(\|x\|^2-\frac{c^2}{4}\right)_+\,d\mu(x)
+
\int\left(\|y\|^2-\frac{c^2}{4}\right)_+\,d\nu(y)
\right\}
\right]^{\frac{1}{2}}.
\]
\end{theorem}

\section{Wasserstein barycenter with Huber loss}\label{sect:Barycenters}
In this section we study the Wasserstein barycenter with Huber loss. 
Let \(\mathcal P(\mathcal P_1(\mathbb R^d))\) denote the set of Borel
probability measures on \(\mathcal P_1(\mathbb R^d)\), endowed with the topology
of weak convergence. Fix \(\mathbf P\in\mathcal P(\mathcal P_1(\mathbb R^d))\), \(c>0\), and define
\[
\Gamma_{c,\mathbf P}(\nu)
:=
\int
\left(
\mathcal T_{\rho_c}(\nu,\mu)
-
c\int\|x\|\,d\mu(x)
\right)
\,d\mathbf P(\mu).
\]
We subtract \(c\int \|x\|\,d\mu(x)\) in \eqref{eq:Wasserstein-bary-huber} to ensure that the objective is well defined even without imposing any moment condition on \(\mathbf{P}\)
\begin{proposition}\label{propo:finite}
For every \(\mu,\nu\in\mathcal P_1(\mathbb R^d)\),
\begin{equation}
\label{eq:bound-integrand-bary}
\left|
\mathcal T_{\rho_c}(\nu,\mu)
-
c\int\|x\|\,d\mu(x)
\right|
\le
c\int\|y\|\,d\nu(y)+\frac{c^2}{2}.
\end{equation}
In particular, for every
\(\mathbf P\in\mathcal P(\mathcal P_1(\mathbb R^d))\) and
\(\nu\in\mathcal P_1(\mathbb R^d)\), the quantity
\(\Gamma_{c,\mathbf P}(\nu)\) is finite.
\end{proposition}

The result above motivates the definition of Huber-Wasserstein barycenters.
\begin{definition}
  A  Huber-Wasserstein barycenter with  parameter $c$ is defined as
  $$
  m_c({\bf P}):= \argmin_{\nu\in \mathcal{P}_1(\R^d)} \Gamma_{c,{\bf P}}(\nu).
  $$
\end{definition}

\subsection{Stability and existence}\label{sect:Stability-OT}
The goal of this section is to prove the existence and stability of barycenters.
We first present a stability result, crucial to proving the existence of barycenters: 

\begin{proposition}
\label{theorem:Stability}
Fix
\(\{{\bf P}_n\}_{n=1}^{\infty}\cup\{\mathbf P\}
\subset\mathcal P(\mathcal P_1(\mathbb R^d))\), with
\({\bf P}_n\to{\bf P}\) in \(\mathcal P(\mathcal P_1(\mathbb R^d))\) as
\(n\to\infty\). Then the following statements hold:
\begin{enumerate}
\item \textup{(Weak coercivity)}
There exists a constant \(C>0\) such that, for every
{\(\varepsilon\geq0\)} and \(n\in\mathbb N\), if
\(\nu\in\mathcal P_1(\mathbb R^d)\) satisfies
\(\Gamma_{c,{\bf P}_n}(\nu)\leq\varepsilon\), then
\[
\int\|y\|\,d\nu(y)\leq C(1+\varepsilon)
\quad\text{and}\quad
\Gamma_{c,{\bf P}_n}(\nu)\geq-C(1+\varepsilon).
\]

\item \textup{(Stability)}
Let \(\{\nu_n\}_n\) be such that, for some
\(\varepsilon_n\geq0\) with \(\varepsilon_n\to0\),
\[
\Gamma_{c,{\bf P}_n}(\nu_n)
\leq
\Gamma_{c,{\bf P}_n}(\nu)+\varepsilon_n
\]
for all \(n\in\mathbb N\) and \(\nu\in\mathcal P_1(\mathbb R^d)\).
Then \(\{\nu_n\}_n\) is relatively compact in
\(\mathcal P(\mathbb R^d)\), and every limit point belongs to
\(m_c({\bf P})\).
\end{enumerate}
\end{proposition}

From \cref{theorem:Stability}, we can deduce the following theorem:
\begin{theorem}[Existence of barycenters]
\label{theorem:Existence-General}
For any \({\bf P}\in\mathcal P(\mathcal P_1(\mathbb R^d))\), the set
\(m_c({\bf P})\) of Huber--Wasserstein barycenters with parameter
\(c>0\) is non-empty, bounded, convex\footnote{Here convexity is in the standard
flat geometry.}, and closed in \(\mathcal W_1\).
\end{theorem}

\subsubsection{Empirical estimation}

We consider two sampling models. In the \emph{one-stage sampling model}, we observe independent random
probability measures $\mu_1,\ldots,\mu_n
\overset{\mathrm{i.i.d.}}{\sim}{\bf P},$ 
and form the empirical distribution
${\bf P}_n
:=
\frac1n\sum_{i=1}^n\delta_{\mu_i}.$
An empirical Huber--Wasserstein barycenter is any $\widehat\nu_{n,c}\in m_c({\bf P}_n).$ In many applications, the probability measures \(\mu_i\) are not observed
directly. Instead, each \(\mu_i\) is estimated from observations sampled
from it \cite{szabo2016learning,bachoc2023improved,bachoc2026wassersteinspatialdepth}. This gives the following two-stage sampling model: $\mu_1,\ldots,\mu_n
\overset{\mathrm{i.i.d.}}{\sim}{\bf P},$ 
and, conditionally on \(\mu_1,\ldots,\mu_n\), sample 
$X_{i1},\ldots,X_{iN}
\overset{\mathrm{i.i.d.}}{\sim}\mu_i,
$ for $
i=1,\ldots,n.$
Define
$\widehat\mu_{i,N}
:=
\frac1{N}\sum_{j=1}^{N}\delta_{X_{ij}}$
and the empirical distribution of the estimated measures
$\widehat{\bf P}_{n,N}
:=
\frac1n\sum_{i=1}^n\delta_{\widehat\mu_{i,N}},$
A two-stage empirical Huber--Wasserstein barycenter is any $\widehat\nu_{n,N,c}
\in
m_c(\widehat{\bf P}_{n,N}).$ 
\begin{proposition}
    [Consistency under one-stage and two-stage sampling]
\label{corollary:consistency-Huber-barycenters}
Fix \(c>0\) and  ${\bf P}\in\mathcal P\bigl(\mathcal P_1(\mathbb R^d)\bigr)$. 

\begin{enumerate}
\item \textup{(One-stage sampling.)}
Let $\widehat\nu_{n,c}\in m_c({\bf P}_n).$ 
Then, almost surely, \(\{\widehat\nu_{n,c}\}_n\) is relatively compact in
\(\mathcal P(\mathbb R^d)\), and every limit point belongs to
\(m_c({\bf P})\). If $m_c({\bf P})=\{\nu_c\},$ 
then
\[
\widehat\nu_{n,c}\longrightarrow\nu_c
\qquad\text{in }\mathcal P(\mathbb R^d)
\quad\text{almost surely}.
\]

\item \textup{(Two-stage sampling.)}
Assume that $N=N(n)\to \infty,$ 
and let $\widehat\nu_{n,\mathbf N,c}
\in m_c(\widehat{\bf P}_{n,\mathbf N}).$ 
Then, almost surely,
\(\{\widehat\nu_{n,\mathbf N,c}\}_n\) is relatively compact in
\(\mathcal P(\mathbb R^d)\), and every limit point belongs to
\(m_c({\bf P})\). If $m_c({\bf P})=\{\nu_c\},$ 
then
\[
\widehat\nu_{n,\mathbf N,c}\longrightarrow\nu_c
\qquad\text{in }\mathcal P(\mathbb R^d)
\quad\text{almost surely}.
\]
\end{enumerate}
\end{proposition}

\subsubsection{Stability with respect to the robustification parameter}

Define the Wasserstein median functional by
\[
\mathcal F_1(\nu)
:=
\int_{\mathcal P_1(\mathbb R^d)}\left(
\mathcal W_1(\nu,\mu)-\int \|x\|d\mu \right)\,d{\bf P}(\mu),
\]
and let
\[
m_0({\bf P})
:=
\argmin_{\nu\in\mathcal P_1(\mathbb R^d)}
\mathcal F_1(\nu).
\]
We also define the quadratic Wasserstein barycenter functional as
\[
\mathcal F_2(\nu)
:=
\frac12\int_{\mathcal P_2(\mathbb R^d)}
\mathcal W_2^2(\nu,\mu)\,d{\bf P}(\mu),
\]
if $\mathbf{P}$ is supported on $\mathcal{P}_2$, and +$\infty$ otherwise. We write
\[
m_\infty({\bf P})
:=
\argmin_{\nu\in\mathcal P_2(\mathbb R^d)}
\mathcal F_2(\nu).
\]
\begin{theorem}[Stability as \(c\downarrow0\) and \(c\uparrow\infty\)]
\label{theorem:stability-in-c}
The following statements hold.

\begin{enumerate}
\item \textup{(Limit as \(c\downarrow0\).)} Assume that ${\bf P}\in
\mathcal P\bigl(\mathcal P_1(\mathbb R^d)\bigr)$. 
Let \(c_n\downarrow0\), and choose $\nu_n\in m_{c_n}({\bf P}).$ 
Then \(\{\nu_n\}_n\) is weakly relatively compact, and every weak limit point
belongs to \(m_0({\bf P})\). If \(m_0({\bf P})=\{\nu_0\}\), then $\nu_n\to \nu_0$ in $\mathcal{P}(\R^d)$. 
\item \textup{(Limit as \(c\uparrow\infty\).)}
Assume that ${\bf P}\in
\mathcal P\bigl(\mathcal P_2(\mathbb R^d)\bigr),$ and 
$\int_{\mathcal P_2(\mathbb R^d)}
\int_{\mathbb R^d}\|x\|^2\,d\mu(x)\,d{\bf P}(\mu)<\infty.$
Let \(c_n\uparrow\infty\), and choose $\nu_n\in m_{c_n}({\bf P})$
Then \(\{\nu_n\}_n\) is weakly relatively compact, and every weak limit point
belongs to \(m_\infty({\bf P})\). If \(m_\infty({\bf P})=\{\nu_\infty\}\), then $\nu_n\to\nu_\infty$ in $\mathcal{P}(\R^d)$. 
\end{enumerate}
\end{theorem}

\subsection{Characterization}
In the following result, we characterize the set of  Huber--Wasserstein barycenters of \({\bf P}\) with parameter
\(c>0\). 

\begin{theorem}
\label{theorem:charcterization}
Fix \({\bf P}\in\mathcal P(\mathcal P_1(\mathbb R^d))\).
Then \(\nu\) is a Huber--Wasserstein barycenter of \({\bf P}\) with parameter
\(c>0\) if and only if there exists a jointly measurable function $(\mu,x)\longmapsto f_{\nu,\mu}(x)$ 
such that
\begin{enumerate}
\item for \({\bf P}\)-a.e.\ \(\mu\), the function \(f_{\nu,\mu}\) is a
Huber potential for \((\nu,\mu)\), in the sense of \cref{theorem:duality};

\item for \(\nu\)-a.e.\ \(x\), $\int f_{\nu,\mu}(x)\,d{\bf P}(\mu)=0;$

\item for every \(x\in\mathbb R^d\), $\int f_{\nu,\mu}(x)\,d{\bf P}(\mu)\geq0.$ 
\end{enumerate}
\end{theorem}
The proof of \Cref{theorem:charcterization} is postponed to the
supplementary material. It relies on the following auxiliary lemmas,
which we state here because they may be of independent interest. In
particular, combining \Cref{lemma:Gateaux} with
\Cref{theorem:uniqueness-potentials}, we obtain that the functional
$\nu \mapsto \mathcal T_{\rho_c}(\mu,\nu)$
 is directionally differentiable along admissible mixture directions. If 
 $\mu\ll \mathcal{L}_d$ is such that $\mu(\partial {\rm supp}(\mu))=0$ and ${\rm supp}(\mu)$, the directional
derivative is linear in the direction $\gamma-\nu$.
This differentiability property generally fails for the
$1$-Wasserstein distance, since the associated Kantorovich potentials
need not be unique. The same regularity is inherited by the
barycenter objective functional. Thus, replacing the Euclidean
distance by the Huber loss yields an additional degree of smoothness,
in analogy with the role played by the Huber loss in classical robust
statistics on the real line.   

\begin{lemma}[Subgradient of the cost functional]
\label{lemma:Gateaux}
Fix \(\mu,\nu,\gamma\in\mathcal P_1(\mathbb R^d)\), and let
\({\rm Sol}^*(\nu,\mu)\) be the set of Huber potentials for the pair \((\nu,\mu)\).
Then, for every \(f_{\nu,\mu}\in{\rm Sol}^*(\nu,\mu)\),
\begin{equation}
\label{Subgradient-1}
\mathcal T_{\rho_c}(\mu,\nu)
\leq
\mathcal T_{\rho_c}(\mu,\gamma)
+
\int f_{\nu,\mu}\,d(\nu-\gamma).
\end{equation}
Furthermore,
\[
\lim_{t\downarrow0}
\frac{
\mathcal T_{\rho_c}\bigl(\mu,\nu+t(\gamma-\nu)\bigr)
-
\mathcal T_{\rho_c}(\mu,\nu)
}{t}
=
\sup_{f\in{\rm Sol}^*(\nu,\mu)}
\int f\,d(\gamma-\nu).
\]
\end{lemma}

\begin{lemma}[Subgradient of the barycenter objective]
\label{lemma:Gateaux-Ganmma}
Fix \(\nu,\gamma\in\mathcal P_1(\mathbb R^d)\) and
\({\bf P}\in\mathcal P(\mathcal P_1(\mathbb R^d))\). Let
$(\mu,x)\mapsto f_{\nu,\mu}(x)$
be jointly measurable and such that, for \({\bf P}\)-a.e.\ \(\mu\),
\(f_{\nu,\mu}\) is a potential for \((\nu,\mu)\). Then
\begin{equation}
\label{Subgradient-1-Gamma}
\Gamma_{c,{\bf P}}(\nu)
\leq
\Gamma_{c,{\bf P}}(\gamma)
+
\int\!\!\int
f_{\nu,\mu}(x)\,d(\nu-\gamma)(x)\,d{\bf P}(\mu).
\end{equation}
Furthermore,
\begin{align}
\label{eq:Gateaux-Gamma}
&\lim_{t\downarrow0}
\frac{
\Gamma_{c,{\bf P}}\bigl(\nu+t(\gamma-\nu)\bigr)
-
\Gamma_{c,{\bf P}}(\nu)
}{t}
\nonumber\\
&\qquad=
\int
\left(
\sup_{f\in{\rm Sol}^*(\nu,\mu)}
\int f(x)\,d(\gamma-\nu)(x)
\right)d{\bf P}(\mu).
\end{align}
\end{lemma}

\subsection{Finite-sample breakdown point}
In this section, we analyze the robustness of the empirical
Huber--Wasserstein barycenter through its finite-sample breakdown
point. This is a widely used measure of robustness that quantifies the
smallest proportion of observations that must be replaced by arbitrary
contaminants in order to make an estimator arbitrarily unstable \cite{Huber.Ronchetti.2009}. It
therefore provides a natural way to assess the resistance of the
Huber--Wasserstein barycenter to severe contamination.

The Hausdorff distance between two bounded closed sets
\(A,B\subset\mathcal P_1(\mathbb R^d)\) is defined as
\[
{\rm dist}(A,B)
:=
\max\left\{
\sup_{\alpha\in A}\inf_{\beta\in B}\mathcal W_1(\alpha,\beta),
\sup_{\beta\in B}\inf_{\alpha\in A}\mathcal W_1(\alpha,\beta)
\right\}.
\]

\begin{definition}
Let
$\boldsymbol{\mu}
=
\frac1n\sum_{i=1}^n\delta_{\mu_i}$
be an empirical probability measure. The breakdown point of
\(m_c(\boldsymbol{\mu})\) is
\[
{\rm BP}(m_c(\boldsymbol{\mu}))
:=
\frac1n
\inf
\left\{
s\in\{1,\ldots,n\}:
\sup_{\boldsymbol{\nu}\in\mathcal Q_s(\boldsymbol{\mu})}
{\rm dist}
\bigl(m_c(\boldsymbol{\nu}),m_c(\boldsymbol{\mu})\bigr)
=\infty
\right\},
\]
where \(\mathcal Q_s(\boldsymbol{\mu})\) is the set of
empirical probability measures with \(n\) atoms sharing at
least \(n-s\) atoms with \(\boldsymbol{\mu}\).
\end{definition}

\begin{theorem}[Breakdown point]
\label{Theorem:BP}
We have
\[
{\rm BP}(m_c(\boldsymbol{\mu}))
\in
\left[
\frac1n\left\lceil\frac n2\right\rceil,
\frac1n\left(\left\lfloor\frac n2\right\rfloor+1\right)
\right].
\]
\end{theorem}
The theorem shows that the finite-sample breakdown point is essentially
$1/2$. In particular, replacing fewer than half of the observations
cannot drive the Huber--Wasserstein barycenter arbitrarily far away,
whereas contaminating slightly more than half is sufficient to do so.
For odd $n$, the two bounds coincide and
${\rm BP}(m_c(\boldsymbol{\mu}))
=
\frac{n+1}{2n},$
while for even $n$ the breakdown point lies between $1/2$ and
$1/2+1/n$.
\subsection{The one-dimensional case}
\label{sec:one-dimensional}

For $\mu\in\mathcal P_1(\mathbb R)$, let $Q_\mu(u)
:=
\inf\bigl\{x\in\mathbb R:F_\mu(x)\geq u\bigr\}$ for $u\in(0,1),$ 
denote its quantile function.
Then \cite[Theorem~5.1]{Ambrosio.Pratelli.2003.LNM} yields that, for every $\mu,\nu\in\mathcal P_1(\mathbb R)$,
\begin{equation}
\label{eq:huber-quantile-cost}
T_{\rho_c}(\mu,\nu)
=
\int_0^1
\rho_c\bigl(\lvert Q_\mu(u)-Q_\nu(u)\rvert\bigr)\,du.
\end{equation}
For $u\in(0,1)$, ${\bf P}\in\mathcal P\bigl(\mathcal P_1(\mathbb R)\bigr)$
 and $c>0$, define
\begin{equation}
\label{eq:renormalized-pointwise-functional}
\Phi_{c,u}(z)
:=
\int_{\mathcal P_1(\mathbb R)}
\left[
\rho_c\bigl(\lvert z-Q_\mu(u)\rvert\bigr)
-c\lvert Q_\mu(u)\rvert
\right]\,d{\bf P}(\mu),
\qquad z\in\mathbb R.
\end{equation}
This quantity is finite for every $(u,z)\in(0,1)\times\mathbb R$. Note that by \eqref{eq:huber-quantile-cost}, for every
$\nu\in\mathcal P_1(\mathbb R)$,
\begin{equation}
\label{eq:objective-quantile-decomposition}
\Gamma_{c,{\bf P}}(\nu)
=
\int_0^1
\Phi_{c,u}\bigl(Q_\nu(u)\bigr)\,du.
\end{equation}
Using this relation we derive the following result. 
\begin{proposition}[One-dimensional Huber--Wasserstein barycenters]
\label{prop:huber-barycenter-1d}
Fix $u\in(0,1)$, ${\bf P}\in\mathcal P\bigl(\mathcal P_1(\mathbb R)\bigr)$
 and $c>0$. The set
\begin{equation}
\label{eq:pointwise-minimizer-set}
M_c(u)
:=
\operatorname{argmin}_{z\in\mathbb R}\Phi_{c,u}(z)
=
\bigl\{z\in\mathbb R:S_{c,u}(z)=0\bigr\}
\end{equation}
is a nonempty compact interval. Furthermore, 
$\nu\in\mathcal P_1(\mathbb R)$  is a Huber--Wasserstein barycenter if and only if 
\begin{equation}
\label{eq:population-huber-score}
\int_{\mathcal P_1(\mathbb R)}
\psi_c\bigl(Q_\nu(u)-Q_\mu(u)\bigr)\,d{\bf P}(\mu)
=0
\qquad\text{for a.e. }u\in(0,1), 
\end{equation}
where  $\psi_c(t):=
\max\{-c,\min\{t,c\}\}.$  
\end{proposition}

In dimension one, the distance-based Huber problem remains convex thanks to
the quantile representation
\[
\mathcal{W}_2(\nu,\mu)
=
\left\|Q_\nu-Q_\mu\right\|_{L^2(0,1)}.
\]
Indeed, the set of quantile functions is a convex subset of
$L^2(0,1)$, and the map $q\mapsto \rho_c\bigl(\|q-Q_\mu\|_{L^2(0,1)}\bigr)$ 
is convex, since $\rho_c$ is convex and nondecreasing on
$[0,\infty)$. Consequently, the functional
\[
\nu\longmapsto
\int_{\mathcal P_2(\mathbb R)}
\rho_c\bigl(\mathcal{W}_2(\nu,\mu)\bigr)\,d{\bf P}(\mu)
\]
is convex in  the quantile representation. Its minimizers are characterized by
\begin{equation}
\label{eq:outside-huber-score}
\int_{\mathcal P_2(\mathbb R)}
\omega^{\mathrm{out}}_{\nu}(\mu)
\bigl(Q_\nu(u)-Q_\mu(u)\bigr)\,d{\bf P}(\mu)
=0
\qquad\text{for a.e. }u\in(0,1),
\end{equation}
where
\begin{equation}
\label{eq:outside-huber-weight}
\omega^{\mathrm{out}}_{\nu}(\mu)
:=
\min\left\{
1,\frac{c}{\mathcal{W}_2(\nu,\mu)}
\right\},
\end{equation}
with the convention $\omega^{\mathrm{out}}_{\nu}(\mu)=1$ when
$\mathcal{W}_2(\nu,\mu)=0$. Comparing \eqref{eq:outside-huber-score} with the characterization
obtained when the Huber loss is placed inside the transport cost,
\begin{equation}
\label{eq:inside-huber-score}
\int_{\mathcal P_1(\mathbb R)}
\omega^{\mathrm{in}}_{\nu}(u,\mu)
\bigl(Q_\nu(u)-Q_\mu(u)\bigr)\,d{\bf P}(\mu)
=0
\qquad\text{for a.e. }u\in(0,1),
\end{equation}
where
\begin{equation}
\label{eq:inside-huber-weight}
\omega^{\mathrm{in}}_{\nu}(u,\mu)
:=
\min\left\{
1,
\frac{c}{|Q_\nu(u)-Q_\mu(u)|}
\right\},
\end{equation}
with the value $1$ when $Q_\nu(u)=Q_\mu(u)$, we see that both
constructions use the same Huber-type clipping mechanism, but at
different levels.

For the construction placing the Huber loss outside the transport cost, the weight $\omega^{\mathrm{out}}_{\nu}(\mu)$ depends on the full Wasserstein
distance between $\nu$ and $\mu$. Hence, each measure $\mu$ receives
a single scalar weight, which is applied simultaneously to all its
quantile levels. By contrast, for the construction placing the Huber loss inside the transport cost, the weight
$\omega^{\mathrm{in}}_{\nu}(u,\mu)$ depends on the individual
quantile displacement $|Q_\nu(u)-Q_\mu(u)|$ and may therefore vary
with $u$. Thus, both procedures bound the contribution of large
discrepancies, but the outside formulation downweights an observation
as a whole, whereas the inside formulation downweights each of its
quantile displacements separately.

\subsubsection{Local robustness and asymptotic efficiency in dimension one}
\label{sec:local-robustness-efficiency}

The breakdown point in \Cref{Theorem:BP} measures global resistance to
contamination, which we complement with a local robustness analysis based on
the influence function and first-order efficiency calculations. This is
the classical perspective of robust $M$-estimation developed in
\cite{Huber.1964.AMS,Hampel.1974.JASA}; see also
\cite{Hampel.etal.1986.Wiley,Huber.Ronchetti.2009}. In dimension one,
\Cref{prop:huber-barycenter-1d} shows that the proposed barycenter is obtained
by solving a scalar Huber location problem at each quantile level. This makes
the usual robustness--efficiency trade-off particularly transparent.

The quantile representation also clarifies the difference between placing the
Huber loss inside and outside the transport problem. By
\eqref{eq:inside-huber-score}, the proposed estimator clips each quantile
displacement separately. On the other hand, \eqref{eq:outside-huber-score} assigns
one weight to the whole probability measure through its $\mathcal W_2$
distance from the center. Thus, the two procedures coincide in ordinary
location models, but they may behave differently when only part of a
distribution is atypical. This distinction is closely related to the
difference between componentwise and casewise robustification in classical
robust statistics; see \cite{AlqallafEtAl2009,RousseeuwVanDenBossche2018}. For $\eta\in\mathcal P_2(\mathbb R)$, consider the family of contaminated measures
\begin{equation}
\label{eq:distribution-valued-contamination-path}
\{{\bf P}_{\varepsilon,\eta}\}_{\varepsilon\in [0,1)]}
:=
\{(1-\varepsilon){\bf P}+\varepsilon\delta_\eta\}_{\varepsilon\in [0,1)}.
\end{equation}
Notice that the contaminating observation in
\eqref{eq:distribution-valued-contamination-path} is itself a probability
measure. This differs from the influence-function analysis of transport-based
quantiles in
\cite{GonzalezSanz.Sheng.Wu.AvellaMedina.2026.Influence}, where the perturbed functional is a transport map evaluated at a fixed reference
point.

\begin{proposition}[Influence function]
\label{prop:inside-influence-function}
Let ${\bf P}\in\mathcal P(\mathcal P_2(\mathbb R))$ and assume that
$m_c({\bf P})=\{\nu_c\}$. Fix $\eta\in\mathcal P_2(\mathbb R)$ and assume that,
for all sufficiently small $\varepsilon>0$,
$m_c({\bf P}_{\varepsilon,\eta})=\{\nu_{c,\varepsilon,\eta}\}$.
For a fixed $u\in(0,1)$ such that \eqref{eq:population-huber-score} holds, set
\[
a_c(u)
:=
{\bf P}\!\left(
|Q_\mu(u)-Q_{\nu_c}(u)|<c
\right).
\]
If $a_c(u)>0$ and
${\bf P}(|Q_\mu(u)-Q_{\nu_c}(u)|=c)=0$, then
\begin{equation}
\label{eq:inside-IF-formula}
\operatorname{IF}_c(\eta;{\bf P})(u)
:=
\lim_{\varepsilon\downarrow0}
\frac{
Q_{\nu_{c,\varepsilon,\eta}}(u)-Q_{\nu_c}(u)
}{\varepsilon}
=
\frac{
\psi_c\!\left(Q_\eta(u)-Q_{\nu_c}(u)\right)
}{a_c(u)}.
\end{equation}
In particular,
$\sup_{\eta\in\mathcal P_2(\mathbb R)}
|\operatorname{IF}_c(\eta;{\bf P})(u)|
\leq c/a_c(u)$.
\end{proposition}

The formula in \eqref{eq:inside-IF-formula} has the same structure as the
influence function of the classical scalar Huber location estimator: the
numerator is clipped at $c$, whereas the denominator is the local slope of the
population score. Hence the influence function is bounded at every quantile
level for which $a_c(u)>0$. If
$\operatorname*{ess\,inf}_{u\in(0,1)}a_c(u)>0$, this bound is uniform in $u$
and in the contaminating distribution.

On the one hand, as $c\uparrow \infty $, the pointwise influence function tends to 
$Q_\eta(u)-\int Q_\mu(u)\,d{\bf P}(\mu)$ and is therefore unbounded over
$\eta$. On the other hand, suppose that $Q_\mu(u)$ has a unique median $m(u)$
and a density $f_u$ that is continuous and positive at $m(u)$. Then, whenever
$Q_\eta(u)\neq m(u)$,
\begin{equation}
\label{eq:median-IF-limit}
\operatorname{IF}_c(\eta;{\bf P})(u)
\longrightarrow
\frac{
\operatorname{sgn}(Q_\eta(u)-m(u))
}{
2f_u(m(u))
}
\qquad\text{as }c\downarrow0.
\end{equation}

The scalar representation also gives a simple asymptotic distribution.
Recall the one-stage sampling model
$\mu_1,\ldots,\mu_n\overset{\mathrm{i.i.d.}}{\sim}{\bf P}$ and
${\bf P}_n=n^{-1}\sum_{i=1}^n\delta_{\mu_i}$. Among the empirical barycenters,
take the one obtained from the smallest pointwise minimizer in the proof of
\Cref{prop:huber-barycenter-1d}; denote it by $\widehat\nu_{n,c}$.

\begin{proposition}[Pointwise asymptotic normality]
\label{prop:inside-pointwise-clt}
Fix $u\in(0,1)$ and assume that $Q_{\nu_c}(u)$ is the  unique minimizer of $\Phi_{c,u}$.  Suppose that
$a_c(u)>0$ and
${\bf P}(|Q_\mu(u)-Q_{\nu_c}(u)|=c)=0$. Then
\begin{equation}
\label{eq:inside-pointwise-clt}
\sqrt n\left(
Q_{\widehat\nu_{n,c}}(u)-Q_{\nu_c}(u)
\right)
\Longrightarrow
N\left(
0,
\frac{1}{
a_c(u)^2
}{
\displaystyle
\int
\psi_c\!\left(Q_\mu(u)-Q_{\nu_c}(u)\right)^2
\,d{\bf P}(\mu)
}
\right).
\end{equation}
The analogous joint convergence at any finite collection of quantile levels
follows from the multivariate central limit theorem.
\end{proposition}

Suppose, for example, that $Q_\mu(u)$ is symmetric about its mean
$q_0(u)$ and has variance $\sigma^2(u)<\infty$. Then
$Q_{\nu_c}(u)=q_0(u)$ and the asymptotic relative efficiency of the proposed
estimator with respect to the Wasserstein mean, at quantile level $u$, is
\begin{equation}
\label{eq:inside-pointwise-ARE}
\operatorname{ARE}_{c,\infty}(u)
=
\frac{
\sigma^2(u)\,
{\bf P}(|Q_\mu(u)-q_0(u)|<c)^2
}{
\displaystyle
\int
\psi_c(Q_\mu(u)-q_0(u))^2\,d{\bf P}(\mu)
}.
\end{equation}
This is similar to the scalar Huber efficiency formula, applied to the
random quantile $Q_\mu(u)$.

The distance-based Huber mean in
\eqref{eq:Wasserstein-bary-huber-BAD} is the Huber $M$-location of the random
quantile function under the $L^2(0,1)$ norm. Related robust location
functionals for functional data are studied in
\cite{Sinova.GonzalezRodriguez.VanAelst.2018.Bernoulli}, while asymptotic
results for metric- and manifold-valued Huber means are developed in
\cite{Lee-Jung.2026.JRRSS-b}. Rather than introducing a separate functional
asymptotic theory here, we compare the inside and outside constructions in two
models that isolate the statistical consequences of the two clipping
mechanisms.

\begin{proposition}[Coincidence under random translations]
\label{prop:inside-outside-translations}
Let $\mu_0\in\mathcal P_2(\mathbb R)$ and let $Z$ be a real random variable.
For each $z\in\mathbb R$, let
$\mu_z=(x\mapsto x+z)_\#\mu_0$, and let ${\bf P}$ be the law of $\mu_Z$.
If the scalar Huber location
$\theta_c\in\argmin_{\theta\in\mathbb R}
\mathbb E[\rho_c(|Z-\theta|)]$
is unique, then
\begin{equation}
\label{eq:translation-center-equality}
m_c({\bf P})
=
m'_{H,c}({\bf P})
=
\{(x\mapsto x+\theta_c)_\#\mu_0\}.
\end{equation}
The same identity holds for the corresponding empirical centers whenever the
empirical scalar Huber location is unique. In particular, under the
random-translation model above, the inside- and outside-Huber procedures
have identical influence functions and first-order asymptotic efficiencies.
\end{proposition}

As a toy example, if
$Z\sim N(\theta,\sigma^2)$ and $c=k\sigma$, the common estimator in
\Cref{prop:inside-outside-translations} has asymptotic relative efficiency
with respect to the Wasserstein mean
\begin{equation}
\label{eq:gaussian-huber-ARE}
\operatorname{ARE}_{k,\infty}
=
\frac{
(2\Phi(k)-1)^2
}{
1-2k\phi(k)+2(k^2-1)\{1-\Phi(k)\}
},
\qquad
\gamma_k^*
=
\frac{k\sigma}{2\Phi(k)-1},
\end{equation}
where $\phi$ and $\Phi$ denote the standard normal density and distribution
function and $\gamma_k^*$ is the gross-error sensitivity. In particular,
$k\simeq1.345$ gives approximately $95\%$ Gaussian asymptotic efficiency,
which is the classical Huber calibration
\cite{Huber.1964.AMS,Huber.Ronchetti.2009}. Thus the difference between the
inside and outside constructions is not a location effect; it arises from
shape variation across the distribution.

\begin{proposition}[Localized contamination]
\label{prop:localized-tail-influence}
Let $\mu_0:=\delta_0$ and let ${\bf P}_0:=\delta_{\mu_0}$ be the
degenerate distribution-valued population concentrated at $\mu_0$. For
$M>0$ and $\alpha\in(0,1)$, define
$\eta_{M,\alpha}:=(1-\alpha)\delta_0+\alpha\delta_M.$
Consider contamination of ${\bf P}_0$ in the direction
$\eta_{M,\alpha}$, as in
\eqref{eq:distribution-valued-contamination-path}, and define analogously
the influence function of the outside-Huber barycenter $m'_{H,c}$.
If $M>c$ and  $M\sqrt{\alpha}>c,$ 
then, for $u\in(0,1)$,
\begin{equation}
\label{eq:localized-inside-outside-IF}
\operatorname{IF}^{\rm in}_c(\eta_{M,\alpha};{\bf P}_0)(u)
=
c\,\mathbbm{1}_{(1-\alpha,1)}(u),
\qquad
\operatorname{IF}^{\rm out}_c(\eta_{M,\alpha};{\bf P}_0)(u)
=
\frac{c}{\sqrt{\alpha}}\,
\mathbbm{1}_{(1-\alpha,1)}(u).
\end{equation}
\end{proposition}

The difference in \Cref{prop:localized-tail-influence} is analogous to
cellwise versus casewise robustness. The outside formulation recognizes the
whole distribution as one observation and therefore rescales all of its
quantile displacements by the same factor. The proposed estimator clips only
the affected quantile displacements. Hence a distribution that is unusual in
one region may still contribute essentially full information in another
region. This is the distribution-valued analogue of retaining clean cells
inside an otherwise partially contaminated observation
\cite{AlqallafEtAl2009,RousseeuwVanDenBossche2018}.

\begin{proposition}[Efficiency under localized shape variation]
\label{prop:two-block-efficiency}
Fix $\varepsilon\in[0,1)$ and $b>0$. Let $U$ be symmetric with
$|U|\leq b$ almost surely and
$\mathbb E[U^2]=\sigma_U^2>0$. Let $U$, $B$, and $S$ be mutually independent, with
$B\sim\operatorname{Bernoulli}(\varepsilon)$ and $S$ a Rademacher random
variable. For $\tau>0$, choose $L_\tau>\tau+b$ and let
${\bf P}_\tau$ be the law of
\begin{equation}
\label{eq:two-block-random-measure}
\mu_\tau
=
\frac12\delta_U
+
\frac12\delta_{L_\tau+\tau BS}.
\end{equation}
Assume $c>2b$ and $\tau>\sqrt2c$. Let $\widehat\nu_n^{\rm in}$ and $\widehat\nu_n^{\rm out}$ denote the
empirical inside- and outside-Huber centers. Since all quantile functions in
this model are constant on $(0,1/2]$ and on $(1/2,1)$, write
\[
Q_{\widehat\nu_n^{\rm in}}(u)=\widehat\theta_n^{\rm in},
\qquad
Q_{\widehat\nu_n^{\rm out}}(u)=\widehat\theta_n^{\rm out},
\qquad
u\in(0,1/2].
\]
Then $\sqrt n\,\widehat\theta^{\,\rm in}_{n}
\Longrightarrow N(0,\sigma_U^2)$ and 
$\sqrt n\,\widehat\theta^{\,\rm out}_{n}
\Longrightarrow N(0,V_\tau),$ 
where
\begin{equation}
\label{eq:two-block-outside-variance}
V_\tau
=
\frac{
(1-\varepsilon)\sigma_U^2
+
2\varepsilon c^2
\mathbb E\!\left[\dfrac{U^2}{U^2+\tau^2}\right]
}{
\left[
(1-\varepsilon)
+
\varepsilon
\mathbb E\!\left\{
\dfrac{\sqrt2c\,\tau^2}{(U^2+\tau^2)^{3/2}}
\right\}
\right]^2
}.
\end{equation}
In particular,
$V_\tau\to\sigma_U^2/(1-\varepsilon)$ as $\tau\to\infty$, and therefore the
asymptotic relative efficiency of the outside estimator with respect to the
inside estimator for the unaffected lower-half location converges to
$1-\varepsilon$.
\end{proposition}
We now explain the idea behind \Cref{prop:two-block-efficiency}.
The lower half of every distribution contains the clean signal $U$, while
only the upper half is perturbed for an $\varepsilon$-fraction of the
observations. For the inside-Huber estimator, the two halves are treated
separately. Since $c>2b$, the lower-half displacements stay in the quadratic
part of the Huber loss, so all observations contribute to the estimation of
the lower-half location. This gives asymptotic variance $\sigma_U^2$. For the outside-Huber estimator, the weight depends on the Wasserstein
distance of the whole distribution. When $B=1$ and $\tau$ is large, the
perturbation in the upper half makes this distance large, so the whole
observation is downweighted, including its clean lower half. As
$\tau\to\infty$, these observations receive almost no weight. The outside
estimator then effectively uses only the $1-\varepsilon$ fraction of
unperturbed observations, which explains the limiting variance
$\sigma_U^2/(1-\varepsilon)$ and the relative efficiency $1-\varepsilon$.

\Cref{prop:two-block-efficiency} is not a uniform dominance statement. If a
distribution is globally aberrant, assigning one robust weight to the whole
observation may be appropriate. The proposition isolates the complementary
regime in which one part of a distribution is anomalous while another part
remains informative. In that setting, clipping individual transport
displacements avoids the first-order information loss caused by downweighting
the whole distribution.

\paragraph{Insight in higher dimension.}
The same distinction persists beyond the quantile representation. Suppose
that the relevant optimal maps from a candidate barycenter $\nu$ to $\mu$ are
unique and sufficiently regular. If $H_{c,\mu}$ denotes a Huber optimal map
and $T_\mu$ the quadratic optimal map, the contribution of $\mu$ to the two
first-order conditions has, formally, the form
\begin{equation}
\label{eq:higher-dimensional-score-comparison}
\min\left\{
1,\frac{c}{\|x-H_{c,\mu}(x)\|}
\right\}
\bigl(x-H_{c,\mu}(x)\bigr),
\qquad
\min\left\{
1,\frac{c}{\mathcal W_2(\nu,\mu)}
\right\}
\bigl(x-T_\mu(x)\bigr),
\end{equation}
for the inside and outside constructions, respectively. Thus the proposed
method clips each transport displacement separately, whereas the
distance-based Huber mean applies one scalar weight to the complete
displacement field generated by $\mu$. A full influence-function theory in
$d>1$ would additionally require differentiating the optimal maps with
respect to the barycenter measure, which involves a linearization of the
Monge problem and is beyond the scope of the present paper. Nevertheless,
\eqref{eq:higher-dimensional-score-comparison} identifies the same mechanism
as in dimension one: localized transport anomalies need not suppress the
contribution of unaffected regions of the same distribution under inside
Huberization.

\section{Numerical experiments}
\label{sect:numerics}
We report three experiments in which each observation is a probability distribution, and
the statistical task is to estimate a representative distribution in the presence of structured
contamination. The first experiment is a two-dimensional image example, included primarily
to illustrate the interpolation between the Wasserstein mean and median provided by the
Huber--Wasserstein barycenters. The second and third experiments are one-dimensional
distributional examples in which the cutoff parameter \(c>0\) is selected in a data-driven
manner. In the second experiment, \(c\) is selected using held-out clean validation data,
whereas in the third experiment it is selected by a robust cross-validation procedure applied
directly to the contaminated training sample. The code used to generate the figures in this paper can be found at \href{https://github.com/carlosaccp/huber_ot}{https://github.com/carlosaccp/huber\_ot}. 

\subsection{MNIST digit contamination}
\label{subsec:mnist-contamination}

The first experiment uses nine handwritten digit images from the MNIST data set (\cite{LeCun.Cortes.Burges.MNIST}).  Each image is converted into a probability measure on the common \(28\times 28\) pixel grid by normalizing its nonnegative pixel intensity to have total mass one.  We use six images of the digit \(1\) as the target population and three images of the digit \(0\) as contaminants.  Let
\[
a_{1,1},\ldots,a_{1,6},a_{0,1},a_{0,2},a_{0,3}\in\Delta_m,
\qquad m=28^2,
\]
denote the resulting normalized images, where $\Delta_m:=\{x\in \mathbb{R}^m_{\geq0}: \sum_{i=1}^mx_i=1\}$ is the $m$-dimensional unit simplex.  We regard the digit-\(1\) images as the target population and the digit-\(0\) images as contamination, so the empirical distribution of distributions is
\begin{equation}
\widehat{\mathbf P}_{\rm MNIST}
=
\frac19\sum_{j=1}^6\delta_{a_{1,j}}
+
\frac19\sum_{j=1}^3\delta_{a_{0,j}}.
\end{equation}

We compute the Wasserstein mean, the Wasserstein median, and Huber barycenters with parameters
$c\in\{0.48,0.24,0.12,0.06,0.03\},$
using the Sinkhorn-type methods outlined in \hyperref[app:huber-sinkhorn-barycenter]{Appendix B}. The grid coordinates are rescaled to \([0,1]^2\), so the cutoff values are interpreted on the same normalized spatial scale across the image experiments.  To avoid artificially changing the sharpness of the Huber path as \(c\) varies, the Sinkhorn regularization is not fixed in absolute units.  Instead, for each \(c\) we use $\varepsilon_c
=
\eta\,s\!\left(C^{(c)}\right)$,
where \(s\!\left(C^{(c)}\right)\) is the \(0.90\)-quantile of the positive entries of the Huber cost matrix and \(\eta=0.006\).  Thus, any observed change along the path is attributable to the Huber loss itself rather than to a different effective amount of entropic smoothing.

\begin{figure}[h!]
\centering
\includegraphics[width=0.98\linewidth]{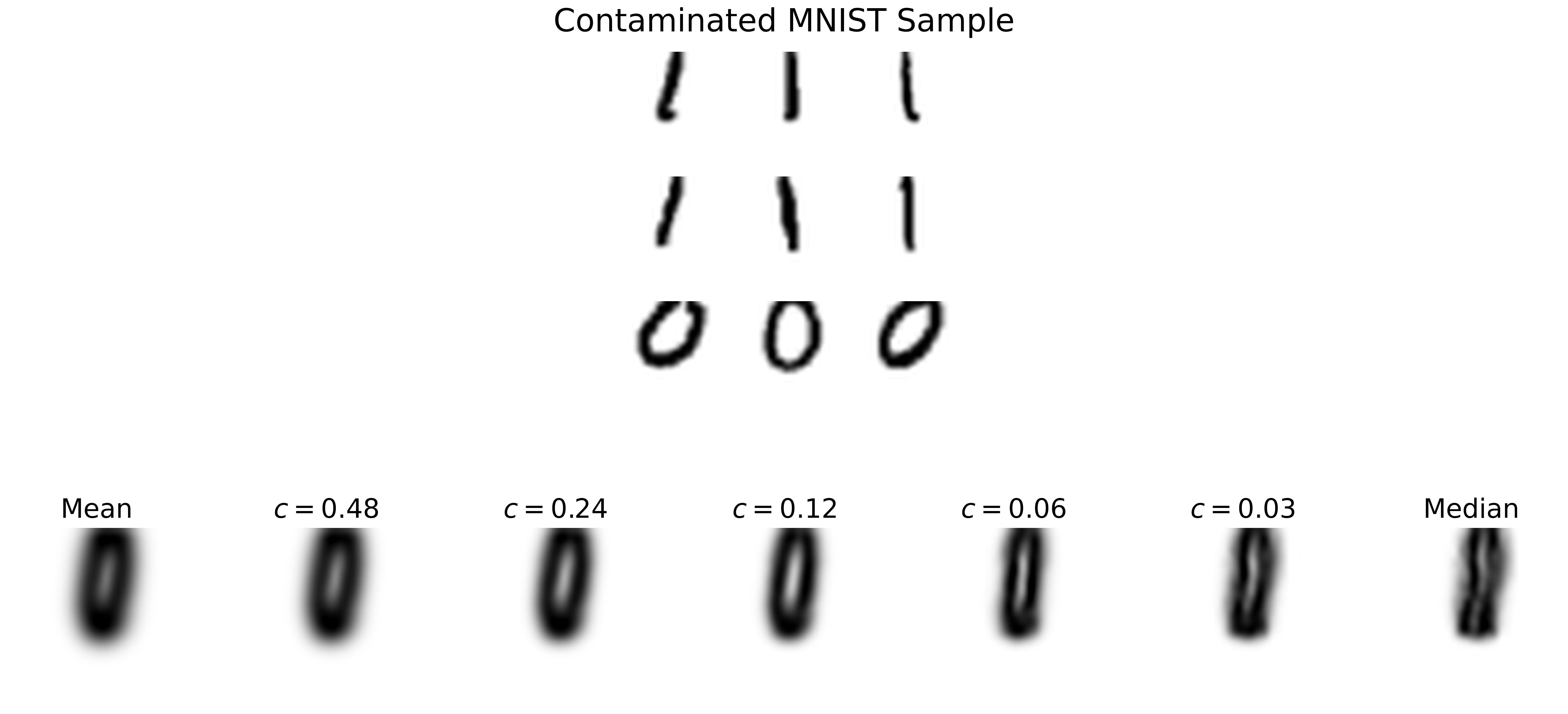}
\caption{MNIST digit contamination experiment.  The upper panel displays the nine input probability images in a \(3\times 3\) grid, consisting of six digit-\(1\) measures and three digit-\(0\) pollutants.  The lower row shows the Wasserstein mean, the Huber path ordered from larger to smaller cutoffs, and the Wasserstein median, all displayed under a common grayscale convention.}
\label{fig:mnist-contamination}
\end{figure}

Figure~\ref{fig:mnist-contamination} illustrates the robustness interpolation in a handwritten-digit example.  The Wasserstein mean is visibly affected by the digit-\(0\) pollutants, because the quadratic transport objective averages the geometry of all nine images and therefore retains mass from the closed loops. The Wasserstein median is more robust and remains closer to the vertical-stroke geometry shared by the majority digit-\(1\) images. Varying \(c\) interpolates between these two behaviors through the ground cost itself: when \(c\) is large, long displacements remain relatively quadratic and the center is closer to the barycenter, while when \(c\) is small, long displacements are effectively linear and the digit-\(0\) pollutants contribute less to the minimizer. The resulting interpolation shows how the proposed estimator moves from a mean-like representation of the contaminated sample toward a median-like representative of the majority digit class.
\subsection{London Underground pollution}
\label{subsec:london}

The second experiment uses passenger-flow profiles for London Underground stations\footnote{Tfl open data \href{https://tfl.gov.uk/info-for/open-data-users}{https://tfl.gov.uk/info-for/open-data-users}, accessed July 15th 2026.}, as in \cite{Carlier.Chenchene.2024.SIMA}. For each station $s$, let $y_s(k)$ be the recorded passenger count in time bin $k$, and define the normalized station profile
\begin{equation}
 p_s(k)=\frac{y_s(k)}{\sum_\ell y_s(\ell)},
 \qquad
 \mu_s=\sum_k p_s(k)\delta_{t_k},
\end{equation}
where $t_k$ is the rescaled time of day.  Stations whose profile peaks during the morning commute are used as the clean target population, while stations whose profile peaks during the evening commute are used as pollutants.  With the fixed random seed used in the figure, the contaminated training sample contains twelve morning-peaked profiles and five evening-peaked profiles, and the validation and test sets contain held-out morning-peaked profiles only.  The target of inference is therefore the typical morning-station distribution rather than the average profile of the contaminated training set.

Because this example is one-dimensional, the barycenters are computed exactly through quantiles (see \Cref{prop:huber-barycenter-1d}).  If $Q_i$ is the quantile function of $\mu_i$, then
\begin{equation}
\label{eq:empiical_mean_median}
Q_{\rm mean}(u)=\frac1n\sum_{i=1}^n Q_i(u),
\qquad
Q_{\rm median}(u)=\operatorname{median}\{Q_1(u),\ldots,Q_n(u)\},
\end{equation}
and for a fixed cutoff $c$,
\begin{equation}
\label{eq:empirical_huber}
Q_c(u)=\argmin_{z\in\mathbb R}\frac1n\sum_{i=1}^n\rho_c(|z-Q_i(u)|).
\end{equation}
The scalar minimization is solved independently for each quantile level by bisection search on the Huber score equation $\sum_i\min\{\max\{z-Q_i(u),-c\},c\}=0$.  The cutoff is selected by minimizing the average held-out Wasserstein distance
\begin{equation}
\mathcal L_{\rm val}(c)
=\frac1{|\mathcal V|}\sum_{\mu_j\in\mathcal V}\mathcal{W}_2(\nu_c,\mu_j)
=\frac1{|\mathcal V|}\sum_{\mu_j\in\mathcal V}
\left(\int_0^1\{Q_c(u)-Q_j(u)\}^2\,du\right)^{1/2}
\end{equation}
over the grid $c\in\{0.01,0.02,0.035,0.05,0.075,0.10,0.15\}$, and final performance is reported on an independent clean test set.

\begin{figure}[H]
\centering
\includegraphics[width=0.98\linewidth]{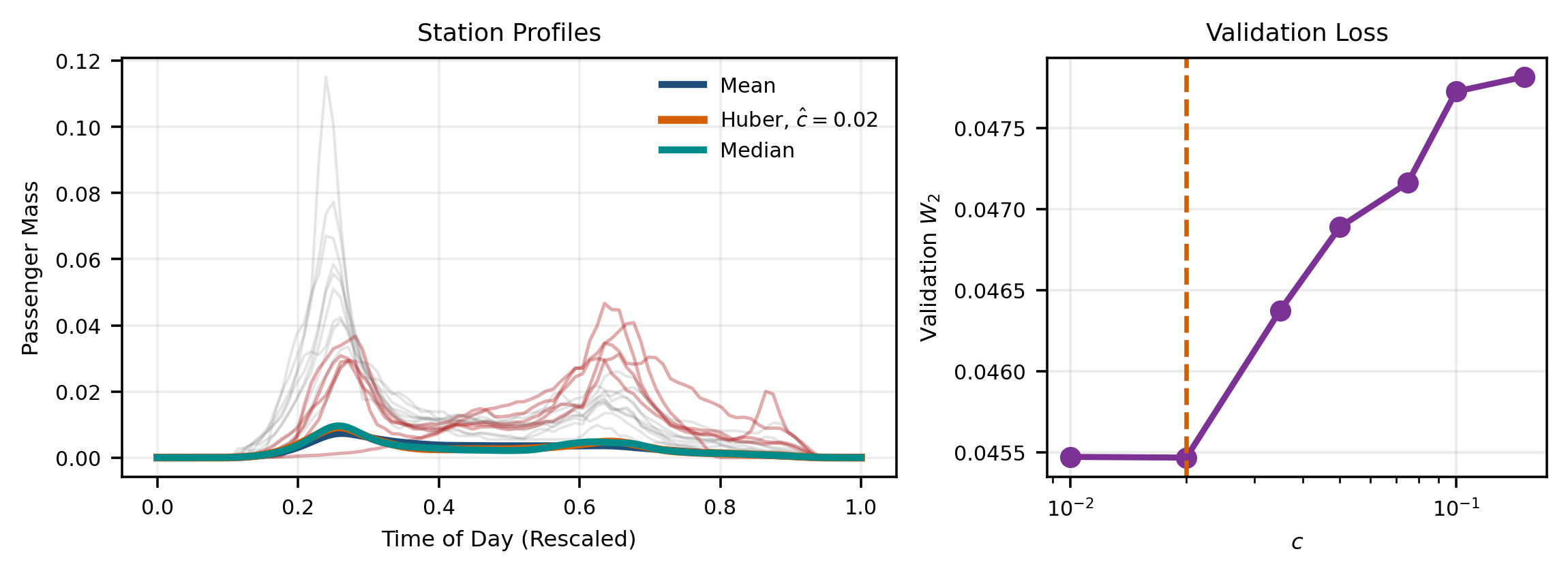}
\caption{London Underground pollution experiment.  Thin gray curves are clean morning-peaked station profiles, thin red curves are evening-peaked pollutants, and the highlighted curves are the Wasserstein mean, the validation-selected Huber center, and the Wasserstein median.  The right panel shows the validation loss over the cutoff grid.}
\label{fig:london}
\end{figure}

The left panel of Figure~\ref{fig:london} shows that the Wasserstein mean inherits a late-day shoulder from the evening pollutants, while the median and the Huber center suppress this shoulder and concentrate around the morning-commute pattern.  The selected value is $\widehat c=0.02$, placing the Huber center close to the median; this is consistent with the fact that the target and pollutant profiles are well separated on the time axis.  The figure also shows that the validation curve is quite flat near the selected value, which explains why the Huber center and the median have nearly identical losses.

\begin{table}[H]
\centering
\begin{tabular}{lcc}
\toprule
Estimator & Validation $\mathcal{W}_2$ & Test $\mathcal{W}_2$\\
\midrule
Wasserstein Mean & 0.049448 & 0.052597\\
\textbf{Huber--Wasserstein mean, $\mathbf{\widehat c=0.02}$} & \textbf{0.045468} & \textbf{0.051821}\\
Wasserstein Median & 0.045483 & 0.052161\\
\bottomrule
\end{tabular}
\caption{Held-out losses in the London Underground experiment.  The best-performing estimator and its loss entries are bolded.}
\label{tab:london}
\end{table}

Table~\ref{tab:london} confirms the visual reading.  Both robust centers improve on the Wasserstein mean, and the selected Huber center has the smallest validation and held-out test losses, although the margin over the median is negligible.  We therefore view this example primarily as a robustness check: the mean places more emphasis on the outliers, while the Huber and median centers estimate the majority morning population.

\subsection{Algorithmic-bias score distributions}
\label{subsec:algorithmic-bias}

The third experiment is motivated by Wasserstein regularization for algorithmic-bias mitigation, where a classifier produces a score \(f_\theta(X)\in[0,1]\) and the distributions of these scores are compared across sensitive groups; see \cite{Risser.GonzalezSanz.Vincenot.Loubes.2022.JMIV,gordaliza2019obtaining,Taskesen.Fairness}.  We do not retrain a classifier in this experiment; instead, we isolate the distributional problem that appears after a classifier has been fitted. We treat each observation as a subgroup-level score distribution on \([0,1]\), and the goal is to estimate a representative score distribution in the presence of biased or corrupted subgroups.

The clean subgroup score distributions are heterogeneous two-component beta mixtures,
\begin{equation}
\mu_{\rm clean}
=
(1-\pi)\,\operatorname{Beta}(2,18)
+
\pi\,\operatorname{Beta}(a,3),
\end{equation}
where $\pi=0.30+0.24B,$ $B\sim\operatorname{Beta}(1.2,3.5),$ $a=10+5B'$ and $B'\sim\operatorname{Beta}(2,1).$ 
These distributions represent subgroups whose classifier scores have a common low-score component but exhibit heterogeneous high-score behavior.  The biased subgroup distributions have a similar low-score component, but their positive-score component is shifted toward the middle of the unit interval:
\begin{equation}
\mu_{\rm bias}
=
(1-\pi_b)\,\operatorname{Beta}(2,21)
+
\pi_b\,\operatorname{Beta}(2.6,9.5),
\quad
\pi_b=\max\{\min\{0.32+\xi,1\},\varepsilon_\pi\}
\end{equation}
where $\xi\sim N(0, 0.015^2)$, and $\varepsilon_\pi = 10^{-6}$ to ensure $\pi_b$ is positive. Thus the biased distributions model subgroups whose scores are systematically depressed relative to the clean population.

The training sample contains both clean and biased score distributions.  With the fixed random seed used in the figure, the training collection consists of ten clean subgroup distributions and four biased subgroup distributions.  The test collection is generated independently and is also mixed, containing both clean and biased subgroup distributions.  Hence, the cutoff \(c\) is not selected using oracle-clean validation data. Instead, it is selected
directly from the contaminated training sample using a robust cross-validation criterion
designed to favor centers that remain stable in the presence of outlying score distributions. As in the previous example, the barycenter can be computed exactly using the one-dimensional quantile representation; see \cref{prop:huber-barycenter-1d}.  If \(Q_i\) denotes the quantile function of \(\mu_i\), then the Wasserstein mean and Wasserstein median are computed pointwise using \eqref{eq:empiical_mean_median} and for a fixed cutoff \(c\), the Huber center has quantile function given by \eqref{eq:empirical_huber},
where the scalar minimization is solved independently for each quantile level \(u\). 

The cutoff is selected over the grid
\begin{equation}
\label{eq:c_gird_bias}
 c\in
 \{0.005,0.01,0.02,0.03,0.04,0.06,0.08,0.10,0.14,0.20\}.
\end{equation}
For each candidate \(c\), we perform four-fold cross-validation on the contaminated training sample.  In each fold, the Huber center is fitted on the remaining folds and scored on the held-out fold.  Because the held-out fold may itself contain biased score distributions, the fold loss is not the ordinary average Wasserstein loss; instead, we use a trimmed Wasserstein criterion.  If \(\mathcal V\) is the held-out fold and \(\nu_c\) is the fitted center, let \(\mathcal V_{0.85}(c)\) denote the subset of the \(85\%\) smallest values\footnote{We round up, i.e.~ we select the $\lceil 0.85 N\rceil$ smallest values.} among \(\{\mathcal{W}_2(\nu_c,\mu_j):\mu_j\in\mathcal V\}\).  The fold loss is
\begin{equation}
\label{eq:fold_loss}
\mathcal L_{\rm fold}(c)
=
\frac{1}{|\mathcal V_{0.85}(c)|}
\sum_{\mu_j\in\mathcal V_{0.85}(c)}
\mathcal{W}_2(\nu_c,\mu_j).
\end{equation}
The selected cutoff is the minimizer of the average trimmed fold loss,
\begin{equation}
\widehat c
=
\argmin_c
\frac14
\sum_{r=1}^4
\mathcal L_{\rm fold}^{(r)}(c).
\end{equation}
This selection rule uses the same mixed clean-and-biased data available to the estimator, while preventing a small number of biased held-out distributions from dominating the tuning criterion.  The selected value is $ \widehat c=0.06$; after its selection, we refit the Huber center on the full contaminated training sample and evaluate the Wasserstein mean, the selected Huber center, and the Wasserstein median on the independent mixed test sample using \eqref{eq:fold_loss}.

\begin{figure}[t]
\centering
\includegraphics[width=0.98\linewidth]{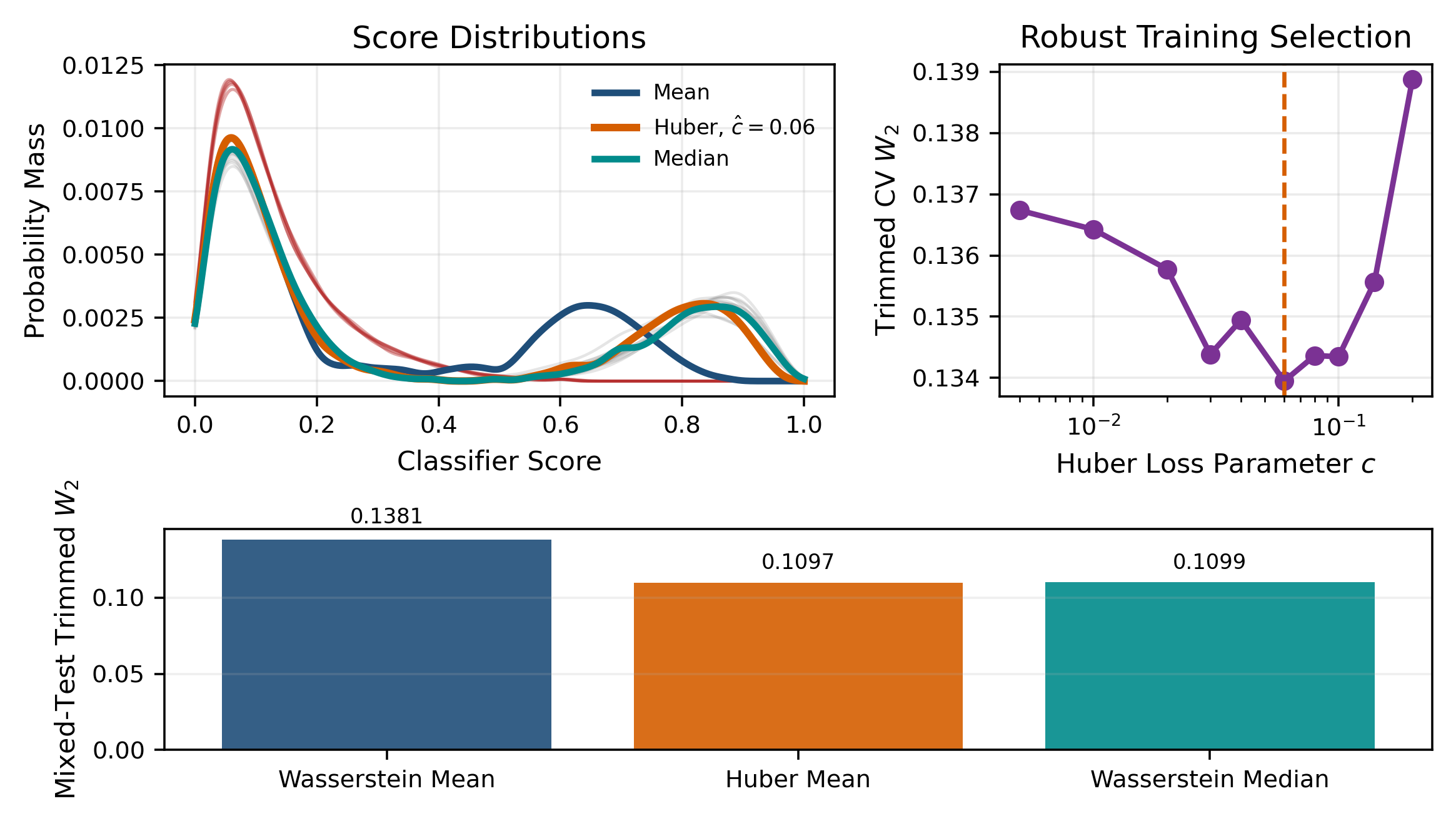}
\caption{Algorithmic-bias score-distribution experiment.  The left panel shows the mixed training sample, with clean subgroup score distributions in translucent gray and biased subgroup score distributions in translucent red, together with the Wasserstein mean, the cross-validation-selected Huber center, and the Wasserstein median.  The upper-right panel shows the robust cross-validation loss over the grid of Huber cutoffs.  The lower panel reports the mixed-test trimmed \(\mathcal{W}_2\) losses.}
\label{fig:bias}
\end{figure}

Figure~\ref{fig:bias} shows the effect of structured score-distribution contamination.  The Wasserstein mean is pulled toward the biased distributions, because it averages the shifted low-score mass with the clean subgroup distributions.  The Wasserstein median is substantially more stable, but it is also the most aggressively robust statistic, and therefore sacrifices some of the clean between-subgroup heterogeneity.  

The selected Huber center lies between these two behaviors.  It discounts the biased score distributions strongly enough to avoid the displacement seen in the mean, while retaining more averaging than the median.  This is the intended role of the cutoff \(c\): it controls how far the estimator moves from the efficient but nonrobust Wasserstein mean toward the highly robust Wasserstein median.

\begin{table}[h!]
\centering
\begin{tabular}{lcc}
\toprule
Estimator & Training-CV Trimmed \(\mathcal{W}_2\) & Mixed-Test Trimmed \(\mathcal{W}_2\)\\
\midrule
Wasserstein Mean & 0.155057 & 0.138132\\
\textbf{Huber--Wasserstein mean, \(\widehat c=0.06\)} & \textbf{0.133948} & \textbf{0.109739}\\
Wasserstein Median & 0.135911 & 0.109907\\
\bottomrule
\end{tabular}
\caption{Robust cross-validation and mixed-test losses in the algorithmic-bias score-distribution experiment.  The cutoff \(c\) is selected on contaminated training folds using an \(85\%\)-trimmed \(\mathcal{W}_2\) criterion.  The best-performing estimator and its loss entries are bolded.}
\label{tab:bias}
\end{table}

Table~\ref{tab:bias} is the strongest quantitative example in this section.  The Huber center selected from the contaminated training sample has the smallest robust cross-validation loss and the smallest mixed-test loss.  Relative to the Wasserstein mean, it substantially reduces the effect of biased low-score distributions.  Relative to the Wasserstein median, the selected Huber--Wasserstein
barycenter achieves a slightly smaller loss in this realization, indicating that the median is very robust but slightly inefficient for a heterogeneous clean target.  This is precisely the regime in which selecting the Huber cutoff from contaminated data is useful: the selected center remains robust to structured outliers without forcing the estimator all the way to the median endpoint.

\bibliographystyle{plainnat}
\bibliography{ref}
\appendix

\section{Proofs}
\label{app:proofs}
\subsection{Proofs in \Cref{sect:stability-properties-cost}}
\begin{proof}[Proof of \cref{lemma:Lipschitz-extension}]
Fix $\varepsilon>0$,  $y,y'\in {\rm dom}({\psi}^{\rho_c})$, with $y\neq y'$,   and $x_\varepsilon,x_\varepsilon'$ such that 
\begin{align}
 \label{eq:psi_y}
{\psi}^{\rho_c}(y)&=\inf_{{x}\in {\mathbb{R}^d}}\{ \rho_c(\|x-y\|)-{\psi}(x)\} \geq \rho_c(\|x_\varepsilon-y\|)-{\psi}(x_\varepsilon) -\varepsilon\|y-y'\| \\  
\label{eq:psi_y_prime}{\psi}^{\rho_c}(y')&=\inf_{{x}\in {\mathbb{R}^d}}\{ \rho_c(\|x-y'\|)-{\psi}(x)\} \geq \rho_c(\|x_\varepsilon'-y'\|)-{\psi}(x_\varepsilon') -\varepsilon\|y-y'\| . 
\end{align}
Using the equality in \eqref{eq:psi_y_prime}, \eqref{eq:Huber-loos-is-lip}, and the inequality \eqref{eq:psi_y}, we derive that
\begin{align*}
   {\psi}^{\rho_c}(y') &\leq   \rho_c(\|x_\varepsilon-y'\|)-{\psi}(x_\varepsilon) \\
   &\leq \rho_c(\|x_\varepsilon-y\|)-{\psi}(x_\varepsilon) +c\|y-y'\| \leq  {\psi}^{\rho_c}(y)+(c+\varepsilon)\|y-y'\| 
\end{align*}
and, by symmetry, 
$  {\psi}^{\rho_c}(y)  \leq  {\psi}^{\rho_c}(y')+(c+\varepsilon)\|y-y'\|$.
Since $\varepsilon$ is arbitrary, we conclude that ${\psi}^{\rho_c}$ is $c$-Lipschitz in its domain, showing (ii). Hence, it only remains to show (i); we will show that, for every $x\in \R^d$,  
$$ s=\inf_{y\in {\mathbb{R}^d}}\{ \rho_c(\|x-y\|)-\psi^{\rho_c}(y)\}>-\infty.$$
Fix $y_0\in {\rm dom}(\psi^{\rho_c})$. We use (ii) together with the triangle inequality to get
\begin{align*}
  s&=\inf_{y\in {\mathbb{R}^d}}\{ \rho_c(\|x-y\|)-\psi^{\rho_c}(y)\}\\
  &= \inf_{y\in {\rm dom}(\psi^{\rho_c})}\{ \rho_c(\|x-y\|)-\psi^{\rho_c}(y)\}\\
  &\geq  -\psi^{\rho_c}(y_0)+ \inf_{y\in {\mathbb{R}^d}}\{ \rho_c(\|x-y\|)-c\|y-y_0\|\}\\
  &\geq  -\psi^{\rho_c}(y_0)-c\|y_0-x\|+ \inf_{y\in {\mathbb{R}^d}}\{ \rho_c(\|x-y\|)-c\|y-x\|\}\\
  &=  -\psi^{\rho_c}(y_0)-c\|y_0-x\|+ \inf_{t\in \R}\{ \rho_c(t)-c |t| \}\\
  &\geq -\psi^{\rho_c}(y_0)-c\|y_0-x\|-\frac{c^2}{2}>-\infty,
\end{align*}
where we have used that
$  \rho_c(t) \geq c |t|-\frac{c^2}{2} $. This shows (i). The proof of the last claim follows directly from (i), (ii) and \cref{theorem:duality}.   
\end{proof}
\begin{proof}[Proof of \cref{lemma:sup-c-conjugate-reestriced}]
Fix $y\in \R^d$. Then 
$$ f^{\rho_c}(y)=\inf_{x\in \R^d}\{ \rho_c(\|x-y\|)-f(x)\}= \inf_{x\in \R^d}\{ \rho_c(\|x\|)-f(x+y)\}= \inf_{x\in \R^d}\{ \rho_c(\|x\|)-h(x)\} , $$
where $h(x):=f(x+y)$ is also $c$-Lipschitz.  
Fix $x$ with $ \|x\|>c$, and let $x_c:=\frac{c}{\|x\|} x $ be the projection of $x$ onto the ball $\overline{\mathbb{B}_c(0)}$.  Then
   \begin{align*}
       \rho_c(\|x\|)-h(x) &= c\|x\|-\frac{c^2}{2}-h(x) \\
       &\geq c (\|x\|-\|x_c-x\|)-\frac{c^2}{2}-h(x_c) = \frac{c^2}{2}-h(x_c)= \rho_c(\|x_c\|)-h(x_c).
   \end{align*}
  The fact that the infimum is attained follows from the compactness of $\overline{\mathbb{B}}_c(y)$. 
  $$f^{\rho_c}(y)=   \inf_{x\in \overline{\mathbb{B}_c(0)}}\{ \rho_c(\|x\|)-f(x+y)\}= \inf_{x\in \overline{\mathbb{B}_c(y)}}\{ \rho_c(\|x-y\|)-f(x)\}. $$ 
\end{proof}
\begin{proof}[Proof of \cref{lemma:cojugate-converges}]
Fix $R>0$ and $y\in \overline{\mathbb{B}_R(0)} $.  By \cref{lemma:sup-c-conjugate-reestriced} we can find $x_n(y), x(y) \in \overline{\mathbb{B}_{R+c}(0)}$  such that 
$ f^{\rho_c}_n(y)=  \rho_c(\|x_n(y)-y\|)-f_n(x_n(y))  $
and 
$  f^{\rho_c}(y)=  \rho_c(\|x(y)-y\|)-f(x(y)) . $
Then, for every $y\in \overline{\mathbb{B}_R(0)} $, 
$$ f^{\rho_c}_n(y) -f^{\rho_c}(y)\leq f(x(y))-f_n(x(y)) \leq \sup_{x\in \overline{\mathbb{B}}_{R+c}(0)}|f_n(x)-f(x)|  $$
and 
$$ f^{\rho_c}(y) -f^{\rho_c}_n(y)\leq f(x_n(y))-f_n(x_n(y)) \leq \sup_{x\in \overline{\mathbb{B}_{R+c}(0)}}|f_n(x)-f(x)| .  $$
Hence, we have shown that 
$$ \sup_{y\in \overline{\mathbb{B}_{R}(0)}}|f_n^{\rho_c}(y)-f^{\rho_c}(y)|  \leq \sup_{x\in \overline{\mathbb{B}_{R+c}(0)}}|f_n(x)-f(x)|, $$
which concludes the proof. 
\end{proof}
\begin{proof}[Proof of \cref{theorem:stability-cost}] 
 Let    $(\psi_i,\varphi_i)$ be $c$-Lipschitz potentials for $(\mu_i,\nu)$, $i=1,2$. Then it follows that  
 $$ \mathcal{T}_{ \rho_c}(\mu_1,\nu)-\mathcal{T}_{ \rho_c}(\mu_2,\nu) \leq  \int \psi_1 d(\mu_1-\mu_2) \leq c\cdot\mathcal{W}_1(\mu_1,\mu_2) ,$$
 where we used the fact that $\psi_1$ is $c$-Lipschitz as well as the dual formula for $\mathcal{W}_1$ \cite{Villani.09}. The inequality 
 $$ \mathcal{T}_{ \rho_c}(\mu_2,\nu)-\mathcal{T}_{ \rho_c}(\mu_1,\nu) \leq c\cdot \mathcal{W}_1(\mu_1,\mu_2)$$
 follows by symmetry of the argument above. 
\end{proof}
\begin{proof}[Proof of \cref{theorem:stsbilitypot}]

By the Arzelà-Ascoli theorem and \cref{lemma:Lipschitz-extension}, any sequence $f_n$ of dual potentials from $\mu_n$ to $\nu_n$ satisfying \eqref{eq:bounded_at_a_point} admits a further subsequence $f_{n_k}$ converging uniformly on compact sets to some $c$-Lipschitz function $f_\infty$. (Note that w.l.o.g.~we can assume that $f_n(0)=0$ for all $n$.)   Using  \cref{lemma:cojugate-converges}, we derive that  $f_{n_k}^{\rho_c}$ converges in $C_{loc}(\mathbb{R}^d)$ to $f_\infty^{\rho_c}$. By definition of $(f_{n_k}, f_{n_k}^{\rho_c})$, for any pair $c$-Lipschitz functions $(f,g)$ with $f(x)+g(y)\leq \rho_c(\|x-y\|)$,  
$$
A_n= \int f_{n_k} d\mu_{n_k}+ \int f_{n_k}^{\rho_c} d\nu_{n_k} \geq \int f d\mu_{n_k}+ \int g d\nu_{n_k}. 
$$
Since, $\mathcal{W}_1(\mu_{n_k}, \mu) \to 0$ and $\mathcal{W}_1(\nu_{n_k}, \nu) \to 0$, we obtain that 
$$  \int f d\mu_{n_k}+ \int g d\nu_{n_k}  \to  \int f d\mu+ \int g d\nu. $$
Note that it is enough to show that 
\begin{equation}
    \label{eq:stability-potentials}
    \int f_{n_k} d\mu_{n_k}+ \int f_{n_k}^{\rho_c} d\nu_{n_k} \to  \int f_{\infty} d\mu+ \int f_{\infty}^{\rho_c} d\nu, 
\end{equation}
as this implies that for all $c$-Lipschitz functions $(f,g)$ with $f(x)+g(y)\leq \rho_c(\|x-y\|)$,
$$ \int f_{\infty} d\mu+ \int f_{\infty}^{\rho_c} d\nu \geq  \int f d\mu+ \int g d\nu,  $$
which in turn implies that $(f_\infty, f_{\infty}^{\rho_c}) $ is a solution to \eqref{eq:Huber-dual} by \cref{theorem:duality} and \cref{lemma:Lipschitz-extension}. We show \eqref{eq:stability-potentials}. By \cite[Theorem~7.12]{Villani-topics}, for every $\varepsilon>0$, there exists  $R>0$ so that 
$$ c \int_{\|x\|\geq R} \|x\| d\mu_{n_k} + c \int_{\|y\|\geq R} \|y\| d\nu_{n_k}  \leq \varepsilon/2. $$
Furthermore, we can select $R$ in such a way that the set $\{x:\|x\|=R\}$ is negligible for $\mu$ and $\nu$.   
Hence, using \cref{lemma:Lipschitz-extension} we get 
$$ \left| \int f_{n_k} d\mu_{n_k}+ \int f_{n_k}^{\rho_c} d\nu_{n_k}- \int_{\overline{\mathbb{B}}_R(0)} f_{n_k} d\mu_{n_k} -\int_{\overline{\mathbb{B}}_R(0)} f_{n_k}^{\rho_c} d\nu_{n_k} \right| \leq \varepsilon/2 . $$
Since $(f_{n_k} ,f_{n_k}^{\rho_c})$ converges uniformly on the compact set $\overline{\mathbb{B}}_R(0)\times \overline{\mathbb{B}}_R(0)$ and $(\mu_{n_k} \otimes \nu_{n_k})$ converges in distribution to $(\mu \otimes \nu)$ and the boundary of $\overline{\mathbb{B}}_R(0)$ is negligible we derive that 
$$ \int_{\overline{\mathbb{B}}_R(0)} f_{n_k} d\mu_{n_k}+ \int_{\overline{\mathbb{B}}_R(0)} f_{n_k}^{\rho_c} d\nu_{n_k} \to \int_{\overline{\mathbb{B}}_R(0)} f_{\infty} d\mu+ \int_{\overline{\mathbb{B}}_R(0)} f_{\infty}^{\rho_c} d\nu. $$
With the same  choice of $R$ we have 
$$ \left|\int_{\overline{\mathbb{B}}_R(0)} f_{\infty} d\mu+ \int_{\overline{\mathbb{B}}_R(0)} f_{\infty}^{\rho_c} d\nu - \int f_{\infty} d\mu+ \int f_{\infty}^{\rho_c} d\nu \right| \leq \varepsilon/2.$$
Therefore, we have shown that for every $\varepsilon>0$, 
$$ \limsup_{k}  \left|\int f_{n_k} d\mu_{n_k}+ \int f_{n_k}^{\rho_c} d\nu_{n_k} -  \int f_{\infty} d\mu - \int f_{\infty}^{\rho_c} d\nu \right| \leq \varepsilon,  $$
so that \eqref{eq:stability-potentials} follows. 
\end{proof}
\begin{proof}[Proof of \Cref{theorem:Huber-maps}]
We define ${\rm Proj}_1(x,y)=x$ and  ${\rm Proj}_2(x,y)=y.$ 
Let \(\pi_c\) be a Huber plan for \((\mu,\nu)\), and define
\[
E_{\le c}:=\{(x,y):\|x-y\|\le c\},
\qquad
E_{>c}:=\{(x,y):\|x-y\|>c\}.
\]
Disintegrate \(\pi_c\) with respect to its first marginal as $\pi_c=\int \pi_c^x\,d\mu(x),$ 
and set
\[
a(x)
:=
\int \mathbf 1_{E_{>c}}(x,y)\,d\pi_c^x(y),
\qquad
A_{\le c}
:=
\{x\in\mathbb R^d:a(x)=0\}.
\]
The function \(a\) has a Borel representative, which we use in the
definition of \(A_{\le c}\). Set $\pi_{\le c}
:=
\pi_c|_{A_{\le c}\times\mathbb R^d}$ and 
$\pi_{>c}
:=
\pi_c|_{A_{\le c}^c\times\mathbb R^d}.$
Thus $({\rm Proj}_1)_\#\pi_{\le c}
=
\mu|_{A_{\le c}}$ and $({\rm Proj}_1)_\#\pi_{>c}
=
\mu|_{A_{\le c}^c}.$  
Here and below, the subscript \(>c\) labels the residual source part;
the plan \(\pi_{>c}\) may also charge pairs with \(\|x-y\|=c\).
By the definition of \(A_{\le c}\), the plan \(\pi_{\le c}\) is
concentrated on \(E_{\le c}\). We decompose the Huber cost as
\[
\int \rho_c(\|x-y\|)\,d\pi_c(x,y)
=
\underbrace{
\int \rho_c(\|x-y\|)\,d\pi_{\le c}(x,y)
}_{=:I_1}
+
\underbrace{
\int \rho_c(\|x-y\|)\,d\pi_{>c}(x,y)
}_{=:I_2}.
\]
\medskip

\noindent
Before giving the details, we summarize the main steps of the proof.
The basic idea is to separate the part of the optimal plan that remains
in the quadratic region of the Huber loss from the residual part that
enters its linear region, and to construct a transport map on each part.

\begin{itemize}
    \item \textit{Step 1: Quadratic part.}
    We show that the portion of the Huber plan corresponding to source
    points whose displacements remain within distance \(c\) is cyclically
    monotone for the quadratic cost. Since \(\mu\) is absolutely continuous,
    this part is therefore induced by the gradient of a convex function.

    \item \textit{Step 2: Truncation of the residual part.}
    Using an optimal Huber potential, we show that, for every source point
    in the residual region, all transported mass lies on a single oriented
    half-line. We then shorten each such displacement by \(c\). The resulting
    truncated plan is shown to be optimal for the Euclidean-distance cost.

    \item \textit{Step 3: A map for the truncated residual problem.}
    We apply the ray decomposition of \cite{Ambrosio.Pratelli.2003.LNM}
    to the truncated Euclidean-distance optimal plan. On each transport ray
    we replace the conditional plan by a monotone transport map, and then
    glue these ray-wise maps into a global Borel map \(R\) with the same
    marginals and the same Euclidean transport cost.

    \item \textit{Step 4: Recovery of the original residual target.}
    The truncation in Step 2 removed exactly a displacement of length \(c\)
    along each transport ray. We identify the orientation of these rays and
    add this displacement back to the map \(R\), thereby obtaining a map
    \(G\) that transports the residual source measure to the original
    residual target measure.

    \item \textit{Step 5: Residual Huber-cost comparison.}
    Since both \(R(x)\) and \(G(x)\) lie on the same oriented ray from \(x\),
    the residual Huber cost can be written in terms of the Euclidean
    distance traveled by \(R\). The optimality of the truncated plan then
    shows that the map \(G\) has no larger Huber cost than the original
    residual plan.

    \item \textit{Step 6: Gluing the two maps.}
    Finally, we combine the quadratic map from Step 1 with the residual map
    \(G\). The resulting Borel map \(H_c\) pushes \(\mu\) forward to \(\nu\)
    and has Huber cost no larger than the original optimal plan. Optimality
    of the latter forces equality, proving that \(H_c\) is a Huber optimal
    transport map.
\end{itemize}

\medskip

\noindent\textit{Step 1. The quadratic portion.}
We first note that
\(E_{\le c}\cap{\rm supp}(\pi_c)\) is cyclically monotone for the
quadratic cost. Indeed, since \(\pi_c\) is optimal for the Huber cost,
\({\rm supp}(\pi_c)\) is \(\rho_c\)-cyclically monotone. Hence, for every
finite family
\[
(x_i,y_i)_{i=1}^N
\subset
E_{\le c}\cap{\rm supp}(\pi_c)
\]
and every permutation \(\sigma\),
\[
\sum_{i=1}^N
\rho_c(\|x_i-y_i\|)
\le
\sum_{i=1}^N
\rho_c(\|x_i-y_{\sigma(i)}\|).
\]
Since \((x_i,y_i)\in E_{\le c}\), the left-hand side equals $\sum_{i=1}^N\frac12\|x_i-y_i\|^2.$ 
Moreover, $\rho_c(t)\le\frac12t^2,$ for $t\ge0.$ 
Therefore
\[
\sum_{i=1}^N\frac12\|x_i-y_i\|^2
\le
\sum_{i=1}^N
\rho_c(\|x_i-y_{\sigma(i)}\|)
\le
\sum_{i=1}^N
\frac12\|x_i-y_{\sigma(i)}\|^2.
\]
Thus \(E_{\le c}\cap{\rm supp}(\pi_c)\) is
\(\|\cdot\|^2\)-cyclically monotone. Hence there exists a convex
function \(f\) such that
\[
E_{\le c}\cap{\rm supp}(\pi_c)
\subset
\partial f.
\]
Since \(\mu\ll\mathcal L^d\), the convex function \(f\) is
differentiable \(\mu\)-a.e. As \(\pi_{\le c}\) is concentrated on
\(E_{\le c}\cap{\rm supp}(\pi_c)\) and has first marginal
\(\mu|_{A_{\le c}}\), it follows that
$\pi_{\le c}
=
(I\times\nabla f)_\#\mu|_{A_{\le c}}.$
Consequently,
\[
I_1
=
\int_{A_{\le c}}
\rho_c(\|x-\nabla f(x)\|)\,d\mu(x).
\]
Set
\[
\nu_{<c}
:=
({\rm Proj}_2)_\#\pi_{\le c}
=
(\nabla f)_\#\mu|_{A_{\le c}}.
\]
Then \(\nu_{<c}\le\nu\). Setting $\alpha_{>c}
:=
\mu|_{A_{\le c}^c}$ and $\nu_{>c}
:=
({\rm Proj}_2)_\#\pi_{>c},$ 
we have $({\rm Proj}_1)_\#\pi_{>c}
=
\alpha_{>c}$ and  $\nu_{>c}
=
\nu-\nu_{<c}.$ 

\medskip

\noindent\textit{Step 2. The residual affine part and the truncated plan.}
Let \((\varphi,\psi)\) be a pair of optimal Kantorovich potentials for
 the Huber cost. Since \(\varphi\) is \(c\)-Lipschitz, it is differentiable
 \(\mu\)-a.e. We choose a Borel representative of \(\nabla\varphi\) on its
 set of differentiability and extend it arbitrarily outside this set. Moreover,
 for \(\pi_c\)-a.e. \((x,y)\),
\[
\varphi(x)+\psi(y)=\rho_c(\|x-y\|).
\]
At every such pair for which \(\varphi\) is differentiable at \(x\), the
function
\[
w\longmapsto \rho_c(\|w-y\|)-\varphi(w)
\]
has a minimum at \(x\). Therefore, writing \(r:=\|x-y\|\),
\begin{equation}
\label{eq:gradient-fiber-dichotomy}
\nabla\varphi(x)
=
\begin{cases}
x-y, & r<c,\\[1mm]
c\,\dfrac{x-y}{r}, & r\ge c,
\end{cases}
\qquad
\pi_c\text{-a.e.}
\end{equation}
Since \eqref{eq:gradient-fiber-dichotomy} holds \(\pi_c\)-a.e., by
 disintegration there exists, for \(\mu\)-a.e. \(x\), a Borel set
 \(N_x\subset\mathbb R^d\) such that \(\pi_c^x(N_x)=0\) and
 \eqref{eq:gradient-fiber-dichotomy} holds for every
 \(y\in\mathbb R^d\setminus N_x\).

Fix \(x\in A_{\le c}^c\) with this property and such that
\(\varphi\) is differentiable at \(x\). Since
\[
a(x)=\pi_c^x\bigl(\{y:\|x-y\|>c\}\bigr)>0,
\]
we have $\pi_c^x\bigl(\{y:\|x-y\|>c\}\setminus N_x\bigr)>0.$ 
For every
\(y\in\{y':\|x-y'\|>c\}\setminus N_x\), the second alternative in
\eqref{eq:gradient-fiber-dichotomy} gives
$\nabla\varphi(x)=c\frac{x-y}{\|x-y\|},$
and hence \(\|\nabla\varphi(x)\|=c\).

Now let \(y\in\mathbb R^d\setminus N_x\). If \(\|x-y\|<c\), then the
first alternative in \eqref{eq:gradient-fiber-dichotomy} gives
\(\nabla\varphi(x)=x-y\), and therefore
\[
c=\|\nabla\varphi(x)\|=\|x-y\|<c,
\]
a contradiction. Thus $\|x-y\|\ge c$ for  $\pi_c^x\text{-a.e. }y.$ 
Consequently, only the second alternative in
\eqref{eq:gradient-fiber-dichotomy} can occur, and
\[
\nabla\varphi(x)
=
c\frac{x-y}{\|x-y\|}
\qquad
\text{for }\pi_c^x\text{-a.e. }y.
\]
Thus, for \(\mu\)-a.e. \(x\in A_{\le c}^c\), the conditional measure
\(\pi_c^x\) is concentrated on the oriented half-line
\begin{equation}
\label{eq:residual-half-line}
\{x+t\tau(x):t\ge c\},
\qquad
\tau(x):=-\frac{\nabla\varphi(x)}{c}.
\end{equation}
Define, on \(\{(x,y):\|x-y\|\ge c\}\),
\[
T(x,y):=y+c\,\frac{x-y}{\|x-y\|}.
\]
For \(\pi_{>c}\)-a.e. \((x,y)\), writing \(r:=\|x-y\|\), one has
$T(x,y)=x+(r-c)\tau(x).$
Furthermore, for every \(z\in\mathbb R^d\),
\[
\varphi(z)-\varphi(x)
\le
\rho_c(\|z-y\|)-\rho_c(\|x-y\|).
\]
Setting \(z:=T(x,y)\), and using \(\|z-y\|=c\) and \(r\ge c\), gives
\[
\varphi(T(x,y))-\varphi(x)
\le
\frac{c^2}{2}-\left(cr-\frac{c^2}{2}\right)
=-c(r-c).
\]
On the other hand, since \(\varphi\) is \(c\)-Lipschitz,
\[
\varphi(T(x,y))-\varphi(x)
\ge
-c\|T(x,y)-x\|
=-c(r-c).
\]
Therefore, \(\pi_{>c}\)-a.e.,
\begin{equation}
\label{case:eq-different}
\varphi(x)-\varphi(T(x,y))
=
c(r-c)
=
c\|T(x,y)-x\|.
\end{equation}

Let \(\widehat\pi\) be the push-forward of \(\pi_{>c}\) by
\((x,y)\mapsto(x,T(x,y))\), set \(u:=\varphi/c\), and let
$\beta:=({\rm Proj}_2)_\#\widehat\pi.$ 
Since \(T(x,y)\) lies on the segment joining \(x\) and \(y\), the measure
\(\beta\) has finite first moment. Since \(u\) is \(1\)-Lipschitz,
\[
u(x)-u(z)\le \|x-z\|
\qquad
\text{for all }x,z,
\]
and, by \eqref{case:eq-different},
\begin{equation}
\label{eq:hatpi-calibration}
u(x)-u(z)=\|x-z\|
\qquad
\widehat\pi\text{-a.e.}
\end{equation}
It follows that \(\widehat\pi\) is an optimal finite coupling of
\(\alpha_{>c}\) and \(\beta\) for the Euclidean distance: for every finite
coupling \(\gamma\) of these two measures,
\[
\int\|x-z\|\,d\gamma(x,z)
\ge
\int u\,d\alpha_{>c}-\int u\,d\beta
=
\int\|x-z\|\,d\widehat\pi(x,z).
\]

Since the truncation leaves the first coordinate unchanged, the
disintegration of \(\widehat\pi\) with respect to \(\alpha_{>c}\) may be
chosen so that
\[
\widehat\pi^x
=
\bigl(T(x,\cdot)\bigr)_\#\pi_c^x
\qquad
\text{for }\alpha_{>c}\text{-a.e. }x.
\]
Together with \eqref{eq:residual-half-line}, this gives
\[
\widehat\pi^x
\bigl(\{x+s\tau(x):s\ge0\}\bigr)=1
\qquad
\text{for }\alpha_{>c}\text{-a.e. }x,
\]
and
\begin{equation}
\label{eq:nontrivial-residual-fibers}
\widehat\pi^x
\bigl(\{x+s\tau(x):s>0\}\bigr)
=
\widehat\pi^x(\mathbb R^d\setminus\{x\})
=
a(x)>0
\qquad
\text{for }\alpha_{>c}\text{-a.e. }x.
\end{equation}
If \(\alpha_{>c}(\mathbb R^d)=0\), then \(\nu_{<c}=\nu\) and
\(H_{<c}:=\nabla f\) is already the desired map. Hence, in the remainder of
 the proof, assume \(\alpha_{>c}(\mathbb R^d)>0\). Set
\[
m:=\alpha_{>c}(\mathbb R^d)
=\beta(\mathbb R^d)
=\widehat\pi(\mathbb R^d\times\mathbb R^d)>0
\]
and normalize $\widetilde\alpha:=\frac{\alpha_{>c}}{m},$ $\widetilde\beta:=\frac{\beta}{m},$ and  $\widetilde\pi:=\frac{\widehat\pi}{m}.$ 
Then \(\widetilde\pi\in\Pi(\widetilde\alpha,\widetilde\beta)\) and is
optimal for the Euclidean distance.

\medskip

\noindent\textit{Step 3. Ray-wise selection for the truncated residual plan.}
Define
\[
\Gamma_u
:=
\left\{
(x,z)\in\mathbb R^d\times\mathbb R^d:
 u(x)-u(z)=\|x-z\|
\right\}.
\]
The set \(\Gamma_u\) is closed, hence \(\sigma\)-compact, and
\(\widetilde\pi\) is concentrated on \(\Gamma_u\). We use the notation
\(T_{\Gamma_u}\) and \(F_{\Gamma_u}\) for the transport set and the
fixed-point set associated with \(\Gamma_u\), as in
\cite[Section~6]{Ambrosio.Pratelli.2003.LNM}. By
\cite[Theorem~6.2]{Ambrosio.Pratelli.2003.LNM}, for the Euclidean norm the
no-crossing condition and assumptions~(ii)--(iii) of
\cite[Theorem~6.1]{Ambrosio.Pratelli.2003.LNM} hold for \(\Gamma_u\).
Assumption~(i) also holds because
\(\widetilde\alpha\ll\mathcal L^d\).

We also note that the isolated fixed-point part of \(\Gamma_u\) carries no
source mass. Indeed, since \(\widehat\pi\) is concentrated on
\(\Gamma_u\), for \(\alpha_{>c}\)-a.e. \(x\in F_{\Gamma_u}\) the
conditional measure \(\widehat\pi^x\) is concentrated on
$\{z:(x,z)\in\Gamma_u\}=\{x\}.$
Thus
\[
\widehat\pi^x(\mathbb R^d\setminus\{x\})=0
\qquad
\text{for }\alpha_{>c}\text{-a.e. }x\in F_{\Gamma_u}.
\]
Comparing with \eqref{eq:nontrivial-residual-fibers}, which gives the
strictly positive value \(a(x)\) for \(\alpha_{>c}\)-a.e. \(x\), yields
\begin{equation}
\label{eq:no-fixed-source-mass}
\alpha_{>c}(F_{\Gamma_u})=0.
\end{equation}

We apply the construction in the proof of
\cite[Theorem~6.1]{Ambrosio.Pratelli.2003.LNM} to
\(\widetilde\pi\). In Step~5 of that proof, the diagonal pairs associated
with the isolated fixed points
\(F_{\Gamma_u}\setminus T_{\Gamma_u}\) are removed before Steps~1--4 are
applied. By \eqref{eq:no-fixed-source-mass}, this removal does not change
\(\widetilde\pi\). Let
$r:\mathbb R^d\times\mathbb R^d\longrightarrow S_c(\mathbb R^d)$
be a Borel extension of the closed-ray label appearing in Step~1 of that
proof. Its measurability follows from the measurable ray decomposition of
\cite[Lemma~6.1]{Ambrosio.Pratelli.2003.LNM} together with the construction
in Step~1 of the proof of
\cite[Theorem~6.1]{Ambrosio.Pratelli.2003.LNM}. Set
$\sigma:=r_\#\widetilde\pi.$
Step~1 gives a disintegration
\begin{equation}
\label{eq:AP-conditional-plan}
\widetilde\pi
=
\int_{S_c(\mathbb R^d)}\widehat\pi_C\,d\sigma(C),
\end{equation}
where \(\widehat\pi_C\) is concentrated on the fiber \(r^{-1}(C)\)
for \(\sigma\)-a.e.~\(C\). Since \(\widetilde\pi\) is concentrated
on the domain on which \(r\) is the Ambrosio--Pratelli closed-ray label,
\(\widehat\pi_C\) is in addition concentrated on \(C\times C\) for
\(\sigma\)-a.e. \(C\). Set $\alpha_C:=({\rm Proj}_1)_\#\widehat\pi_C$ and 
$\beta_C:=({\rm Proj}_2)_\#\widehat\pi_C.$ 
Then
\begin{equation}
\label{eq:AP-conditional-marginals}
\widetilde\alpha
=
\int_{S_c(\mathbb R^d)}\alpha_C\,d\sigma(C),
\qquad
\widetilde\beta
=
\int_{S_c(\mathbb R^d)}\beta_C\,d\sigma(C), 
\end{equation}
which are the conditional source and target measures used in
Steps~2--4 of the proof of
\cite[Theorem~6.1]{Ambrosio.Pratelli.2003.LNM}. Since
\(\widetilde\alpha\) and \(\widetilde\beta\) have finite first moments,
the same is true of \(\alpha_C\) and \(\beta_C\) for
\(\sigma\)-a.e.~\(C\). By Step~2 of that proof, \(\alpha_C\) is
non-atomic and concentrated on the relative interior of \(C\) for
\(\sigma\)-a.e.~\(C\). Hence Step~3, using
\cite[Theorem~5.1]{Ambrosio.Pratelli.2003.LNM}, produces a nondecreasing
map
$R_C:\operatorname{relint}(C)\to C$
such that
\begin{equation}
\label{eq:AP-conditional-push}
(R_C)_\#\alpha_C=\beta_C
\qquad
\text{for }\sigma\text{-a.e. }C.
\end{equation}
Finally, Step~4, using \cite[Theorem~9.3]{Ambrosio.Pratelli.2003.LNM},
glues these conditional maps into a Borel map on the transport set. Since
\(T_{\Gamma_u}\) is Borel, we extend this map arbitrarily outside
\(T_{\Gamma_u}\) and denote the resulting Borel map by
\(R:\mathbb R^d\to\mathbb R^d\). It satisfies
\begin{equation}
\label{eq:AP-glued-map}
R(x)=R_C(x)
\qquad
\text{for }\alpha_C\text{-a.e. }x,
\quad
\text{for }\sigma\text{-a.e. }C.
\end{equation}
Consequently,
\begin{equation}
\label{eq:conditional-ray-transport}
R_\#\alpha_C=\beta_C
\qquad
\text{for }\sigma\text{-a.e. }C,
\end{equation}
and, after integration in \(C\),
$R_\#\alpha_{>c}=\beta.$ Moreover, \cite[Theorem~6.1(c)]{Ambrosio.Pratelli.2003.LNM}, applied to
\(\widetilde\pi\) with the convex function \(t\mapsto t\), and then
multiplied by \(m\), gives
\begin{equation}
\label{eq:AP-cost-comparison}
\int\|x-R(x)\|\,d\alpha_{>c}(x)
\le
\int\|x-z\|\,d\widehat\pi(x,z).
\end{equation}
Since \((I\times R)_\#\alpha_{>c}\) is a finite coupling of
\(\alpha_{>c}\) and \(\beta\), the optimality of \(\widehat\pi\) gives
the reverse inequality. Therefore equality holds in
\eqref{eq:AP-cost-comparison}. Using \eqref{eq:hatpi-calibration},
we obtain
\[
\int
\Bigl[\|x-R(x)\|-u(x)+u(R(x))\Bigr]
\,d\alpha_{>c}(x)=0.
\]
The integrand is nonnegative because \(u\) is \(1\)-Lipschitz. Hence
\begin{equation}
\label{eq:R-calibrated}
u(x)-u(R(x))=\|x-R(x)\|
\qquad
\text{for }\alpha_{>c}\text{-a.e. }x.
\end{equation}

\medskip

\noindent\textit{Step 4. Recovering the original residual target.}
Define $F(x,y):=(x,T(x,y)),$ 
so that \(F_\#\pi_{>c}=\widehat\pi\). We disintegrate the normalized
residual plan with respect to the same closed-ray label as in Step~3. Since
\[
(r\circ F)_\#\left(\frac{\pi_{>c}}m\right)
=
r_\#\widetilde\pi
=
\sigma,
\]
there exist probability measures \(\{(\pi_{>c})_C\}_C\) such that
\begin{equation}
\label{eq:original-residual-ray-disintegration}
\frac{\pi_{>c}}m
=
\int_{S_c(\mathbb R^d)}(\pi_{>c})_C\,d\sigma(C),
\end{equation}
where \((\pi_{>c})_C\) is concentrated on
\((r\circ F)^{-1}(C)\) for \(\sigma\)-a.e. \(C\). Pushing
\eqref{eq:original-residual-ray-disintegration} forward by \(F\) gives
\[
\widetilde\pi
=
\int_{S_c(\mathbb R^d)}F_\#(\pi_{>c})_C\,d\sigma(C).
\]
For \(\sigma\)-a.e. \(C\), \(F_\#(\pi_{>c})_C\) is concentrated on
\(r^{-1}(C)\). Hence this representation and
\eqref{eq:AP-conditional-plan} are two disintegrations of the same measure
with respect to the same Borel map \(r\). By uniqueness of disintegration
\cite[Theorem~9.2]{Ambrosio.Pratelli.2003.LNM},
\begin{equation}
\label{eq:conditional-truncation-pushforward}
F_\#(\pi_{>c})_C=\widehat\pi_C
\qquad
\text{for }\sigma\text{-a.e. }C.
\end{equation}
In particular, since \({\rm Proj}_1\circ F={\rm Proj}_1\),
\begin{equation}
\label{eq:conditional-original-source}
({\rm Proj}_1)_\#(\pi_{>c})_C=\alpha_C
\qquad
\text{for }\sigma\text{-a.e. }C.
\end{equation}

Let \(\tau_C\) denote the unit orientation of the ray \(C\), chosen in the
direction in which \(u\) decreases with unit slope; this orientation is the
one in \cite[Theorem~6.2(i)]{Ambrosio.Pratelli.2003.LNM}. Because
\(x\mapsto\widehat\pi^x\) is a Borel probability kernel and the subsets of
\(\mathbb R^d\times\mathbb R^d\) appearing below are Borel, the corresponding
section probabilities are Borel functions of \(x\). Hence we may choose a
Borel set \(X_0\subset A_{\le c}^c\) of full
\(\alpha_{>c}\)-measure such that, for every \(x\in X_0\),
\eqref{eq:nontrivial-residual-fibers} holds and
\(\widehat\pi^x\) is concentrated both on $\{z:(x,z)\in\Gamma_u\}$ and 
$\{x+s\tau(x):s\ge0\}.$  Since
\eqref{eq:AP-conditional-marginals} gives
$\widetilde\alpha
=
\int_{S_c(\mathbb R^d)}\alpha_C\,d\sigma(C),$
we have $\alpha_C(X_0)=1$ for $\sigma$-a.e.~$C.$ 
Fix such a \(C\). By Step~2 of the Ambrosio--Pratelli construction,
\(\alpha_C\) is concentrated on \(\operatorname{relint}(C)\). Thus, for
\(\alpha_C\)-a.e. \(x\), we have \(x\in X_0\cap\operatorname{relint}(C)\).
By \eqref{eq:nontrivial-residual-fibers}, there exists a nontrivial truncated
pair starting from \(x\), say \(z\ne x\), with \((x,z)\in\Gamma_u\). By
\eqref{eq:residual-half-line}, $z=x+s\tau(x)$ for some $s>0$. 
Since \(z\neq x\) and \((x,z)\in\Gamma_u\), let \(S'\) be a maximal
transport ray containing the open segment \(]]x,z[[\), and set
$C':=\overline{S'}.$
Since \(x\in\operatorname{relint}(C)\cap C'\), the no-crossing property of
\cite[Theorem~6.2]{Ambrosio.Pratelli.2003.LNM} excludes \(S'\) and the
transport ray underlying \(C\) from having different orientations. If their
orientations agree, the two rays overlap, and maximality forces them to
coincide. Consequently,
$C'=C.$
Hence the orientation of the pair \((x,z)\) is \(\tau_C\). Since $z=x+s\tau(x)$ for some $s>0$, 
we also have $z=x+s\tau_C,$ 
and therefore
\begin{equation}
\label{eq:tau-ray-identification}
\tau(x)=\tau_C
\qquad
\text{for }\alpha_C\text{-a.e. }x,
\quad
\text{for }\sigma\text{-a.e. }C.
\end{equation}

For \((\pi_{>c})_C\)-a.e. \((x,y)\), write \(z:=T(x,y)\). By the definition
of \(T\) and \eqref{eq:residual-half-line},
$y=z+c\tau(x).$
Using \eqref{eq:conditional-original-source} and
\eqref{eq:tau-ray-identification}, it follows that
\[
y=z+c\tau_C
\qquad
\text{for }(\pi_{>c})_C\text{-a.e. }(x,y).
\]
Define $\Theta_C(w):=w+c\tau_C$ and  $(\nu_{>c})_C:=({\rm Proj}_2)_\#(\pi_{>c})_C.$ 
Then
\[
\frac{\nu_{>c}}m
=
\int_{S_c(\mathbb R^d)}(\nu_{>c})_C\,d\sigma(C).
\]
Moreover, \eqref{eq:conditional-truncation-pushforward} shows that the
\(z\)-marginal of \((\pi_{>c})_C\) is \(\beta_C\). Consequently,
\begin{equation}
\label{eq:conditional-original-target}
(\nu_{>c})_C=(\Theta_C)_\#\beta_C
\qquad
\text{for }\sigma\text{-a.e. }C.
\end{equation}

Define $G(x):=R(x)+c\tau(x)$ 
on \(A_{\le c}^c\), after an arbitrary modification on an
\(\alpha_{>c}\)-negligible set. The map \(G\) is Borel. By
\eqref{eq:AP-glued-map} and \eqref{eq:tau-ray-identification}, for
\(\alpha_C\)-a.e. \(x\),
\[
G(x)=R_C(x)+c\tau_C=\Theta_C(R_C(x)).
\]
Hence, by \eqref{eq:conditional-ray-transport} and
\eqref{eq:conditional-original-target},
\[
G_\#\alpha_C
=
(\Theta_C)_\#(R_\#\alpha_C)
=
(\Theta_C)_\#\beta_C
=
(\nu_{>c})_C
\]
for \(\sigma\)-a.e. \(C\). Integrating over \(C\) and multiplying by \(m\)
gives
\begin{equation}
\label{marginal-identify-2}
G_\#\alpha_{>c}=\nu_{>c}.
\end{equation}

\medskip

\noindent\textit{Step 5. Cost comparison on the residual part.}
By \eqref{eq:AP-conditional-marginals} and Step~2 of the
Ambrosio--Pratelli construction, for \(\alpha_{>c}\)-a.e. \(x\) there
is a unique closed maximal ray \(C\) such that \(x\in\operatorname{relint}(C)\).
For such an \(x\), \eqref{eq:AP-glued-map} gives \(R(x)\in C\). If
\(R(x)=x\), set \(s(x):=0\). Otherwise,
\eqref{eq:R-calibrated} implies \((x,R(x))\in\Gamma_u\). Since \(x\) is
in the relative interior of \(C\), the same no-crossing and maximality
argument used in Step~4 shows that the maximal ray associated with
\((x,R(x))\) is \(C\). By the orientation convention for \(\Gamma_u\),
\(R(x)\) lies forward from \(x\) along \(C\). Hence, in either case, there exists
\(s(x)\ge0\) such that
$R(x)=x+s(x)\tau_C.$
Using \eqref{eq:tau-ray-identification}, $R(x)=x+s(x)\tau(x).$ 
Therefore
\[
G(x)
=
R(x)+c\tau(x)
=
x+(s(x)+c)\tau(x),
\]
and hence
\[
\|x-G(x)\|
=
\|x-R(x)\|+c.
\]
Since
\(\|x-G(x)\|\ge c\),
\[
\rho_c(\|x-G(x)\|)
=
c\|x-R(x)\|
+
\frac{c^2}{2}.
\]
Thus
\[
\int
\rho_c(\|x-G(x)\|)
\,d\alpha_{>c}(x)
=
c
\int
\|x-R(x)\|
\,d\alpha_{>c}(x)
+
\frac{c^2}{2}
\alpha_{>c}(\mathbb R^d).
\]

On the other hand, for
\(\pi_{>c}\)-a.e. \((x,y)\), with $z:=T(x,y),$ 
we have $\|x-y\|
=
\|x-z\|+c.$ 
Therefore $\rho_c(\|x-y\|)
=
c\|x-z\|
+
\frac{c^2}{2},$ 
and hence
\[
\int
\rho_c(\|x-y\|)
\,d\pi_{>c}(x,y)
=
c
\int
\|x-z\|
\,d\widehat\pi(x,z)
+
\frac{c^2}{2}
\pi_{>c}
(\mathbb R^d\times\mathbb R^d).
\]
Since $\alpha_{>c}(\mathbb R^d)
=
\pi_{>c}
(\mathbb R^d\times\mathbb R^d)$ 
and, by
\eqref{eq:AP-cost-comparison},
\[
\int
\|x-R(x)\|
\,d\alpha_{>c}(x)
\le
\int
\|x-z\|
\,d\widehat\pi(x,z),
\]
we obtain
\[
\int
\rho_c(\|x-G(x)\|)
\,d\alpha_{>c}(x)
\le
\int
\rho_c(\|x-y\|)
\,d\pi_{>c}(x,y).
\]

\medskip

\noindent\textit{Step 6. Construction of the global map.}
We have constructed
\[
H_{<c}:A_{\le c}\longrightarrow\mathbb R^d,
\qquad
H_{<c}(x):=\nabla f(x),
\]
such that $(H_{<c})_\#
\mu|_{A_{\le c}}
=
\nu_{<c},$ 
and
$G:A_{\le c}^c\to\mathbb R^d$
such that $G_\#\alpha_{>c}
=
G_\#
\mu|_{A_{\le c}^c}
=
\nu_{>c}.$ 
Since \(H_{<c}=\nabla f\) is defined \(\mu|_{A_{\le c}}\)-a.e. and
\(G\) is defined \(\alpha_{>c}=\mu|_{A_{\le c}^c}\)-a.e., we choose
Borel representatives of these maps on \(A_{\le c}\) and
\(A_{\le c}^c\), respectively, modifying them only on sets of
\(\mu\)-measure zero. Since \(A_{\le c}\) is Borel, the piecewise map
\[
H_c(x)
:=
\begin{cases}
H_{<c}(x),
&
x\in A_{\le c},
\\[1mm]
G(x),
&
x\in A_{\le c}^c
\end{cases}
\]
is Borel measurable.
For every Borel set
\(B\subset\mathbb R^d\),
\begin{align*}
(H_c)_\#\mu(B)
&=
\mu
\bigl(
A_{\le c}\cap H_{<c}^{-1}(B)
\bigr)
+
\mu
\bigl(
A_{\le c}^c\cap G^{-1}(B)
\bigr)
\\
&=
(H_{<c})_\#
\mu|_{A_{\le c}}(B)
+
G_\#
\mu|_{A_{\le c}^c}(B)
\\
&=
\nu_{<c}(B)
+
\nu_{>c}(B)
=
\nu(B).
\end{align*}
Thus $(I\times H_c)_\#\mu
\in
\Pi(\mu,\nu).$ 
Finally,
\begin{align*}
\int
\rho_c(\|x-H_c(x)\|)
\,d\mu(x)
&=
\int_{A_{\le c}}
\rho_c(\|x-H_{<c}(x)\|)
\,d\mu(x)
\\
&\quad+
\int_{A_{\le c}^c}
\rho_c(\|x-G(x)\|)
\,d\mu(x)
\\
&\le
\int
\rho_c(\|x-y\|)
\,d\pi_{\le c}(x,y)
\\
&\quad+
\int
\rho_c(\|x-y\|)
\,d\pi_{>c}(x,y)
\\
&=
\int_{\mathbb R^d\times\mathbb R^d}
\rho_c(\|x-y\|)
\,d\pi_c(x,y).
\end{align*}
Since $(I\times H_c)_\#\mu
\in
\Pi(\mu,\nu)$ 
and \(\pi_c\) is optimal, the reverse inequality follows from
optimality. Hence equality holds and $(I\times H_c)_\#\mu$ 
is a Huber optimal plan. Therefore \(H_c\) is a Huber optimal
transport map from \(\mu\) to \(\nu\).
\end{proof}

\begin{proof}[Proof of \cref{prop:ot_is_huber_ot}]
Clearly, {\rm (i)} implies {\rm (iv)}. Assume {\rm (iv)}, and let
\(\pi_c\) be a Huber optimal plan concentrated on
\(\{(x,y):\|x-y\|\le c\}\). Since
\(\rho_c(t)=t^2/2\) for \(t\le c\), we have
\[
\frac12\mathcal W_2^2(\mu,\nu)
\le
\int \frac12\|x-y\|^2\,d\pi_c(x,y)
=
\int \rho_c(\|x-y\|)\,d\pi_c(x,y)
=
\mathcal T_{\rho_c}(\mu,\nu).
\]
On the other hand, since \(\rho_c(t)\le t^2/2\) for every \(t\ge0\),
\(\mathcal T_{\rho_c}(\mu,\nu)\le\frac12\mathcal W_2^2(\mu,\nu)\). Thus {\rm (ii)}
holds.

Assume now {\rm (ii)}, and let \(\pi_*\) be any solution of
\eqref{eq:OT-plan}. Then
\[
\mathcal T_{\rho_c}(\mu,\nu)
\le
\int \rho_c(\|x-y\|)\,d\pi_*(x,y)
\le
\int \frac12\|x-y\|^2\,d\pi_*(x,y)
=
\frac12\mathcal W_2^2(\mu,\nu).
\]
By {\rm (ii)}, equality holds throughout. Hence \(\pi_*\) is Huber optimal and
\[
0
=
\int
\left(
\frac12\|x-y\|^2-\rho_c(\|x-y\|)
\right)
\,d\pi_*(x,y).
\]
Since
\[
\frac12t^2-\rho_c(t)
=
\frac12(t-c)^2\mathbf 1_{\{t>c\}},
\]
we conclude that
\(\pi_*(\{(x,y):\|x-y\|>c\})=0\). Since \(\pi_*\) was arbitrary,
{\rm (iii)} follows.

Finally, the quadratic transport problem admits a solution for
\(\mu,\nu\in\mathcal P_2(\mathbb R^d)\); see
\cite{Cuesta1989NotesOT,brenier1991polar}. Therefore, {\rm (iii)} implies
{\rm (i)}.
\end{proof}
\begin{proof}[Proof of \Cref{propositon:Stable-cost-quantitative}]
The inequality \(\rho_c(t)\le t^2/2\) for every \(t\ge0\) yields
\(\mathcal T_{\rho_c}(\mu,\nu)\le\frac12\mathcal W_2^2(\mu,\nu)\). Moreover,
\begin{align*}
\frac{\|x-y\|^2}{2}-\rho_c(\|x-y\|)
&=
\begin{cases}
0,
& \|x-y\|<c,\\[1mm]
\displaystyle
\frac{\|x-y\|^2}{2}
-c\|x-y\|
+\frac{c^2}{2},
& \|x-y\|\ge c,
\end{cases}\\
&\le
\left(\frac{\|x-y\|^2}{2}-\frac{c^2}{2}\right)_+\\
&\le
\left(\|x\|^2+\|y\|^2-\frac{c^2}{2}\right)_+\\
&\le
\left(\|x\|^2-\frac{c^2}{4}\right)_+
+
\left(\|y\|^2-\frac{c^2}{4}\right)_+.
\end{align*}
Here, the last inequality follows from
\[
(a+b-r)_+
\le
\left(a-\frac r2\right)_+
+
\left(b-\frac r2\right)_+,
\qquad a,b,r\ge0.
\]

Let \(\pi_c\) be a Huber optimal plan. Since \(\pi_c\in\Pi(\mu,\nu)\),
\begin{align*}
\frac12\mathcal W_2^2(\mu,\nu)-\mathcal T_{\rho_c}(\mu,\nu)
&\le
\int
\left(
\frac{\|x-y\|^2}{2}
-\rho_c(\|x-y\|)
\right)
\,d\pi_c(x,y)\\
&\le
\int\left(\|x\|^2-\frac{c^2}{4}\right)_+\,d\mu(x)
+
\int\left(\|y\|^2-\frac{c^2}{4}\right)_+\,d\nu(y).
\end{align*}
This proves the claimed estimate. Finally, both integrands converge pointwise
to zero as \(c\to\infty\) and are dominated respectively by
\(\|x\|^2\) and \(\|y\|^2\). The conclusion follows from the dominated
convergence theorem.
\end{proof}
\begin{proof}[Proof of \cref{theorem:stability_map}]
There exists a convex function \(f_0\) such that
\(H_0=\nabla f_0\); see
\cite{Cuesta1989NotesOT,brenier1991polar}. Since \(H_0\) is
\(\alpha\)-Lipschitz, the convex conjugate \(f_0^*\) is
\(1/\alpha\)-strongly convex. Therefore, for \(p,q\in\operatorname{dom}(f_0^*)\)
and \(v\in\partial f_0^*(q)\),
\[
f_0^*(p)-f_0^*(q)
\ge
\langle v,p-q\rangle+\frac{1}{2\alpha}\|p-q\|^2.
\]
Since \(x\in\partial f_0^*(H_0(x))\) for \(\mu\)-a.e. \(x\), we obtain
\[
f_0^*(H_c(x))-f_0^*(H_0(x))
\ge
\langle x,H_c(x)-H_0(x)\rangle
+
\frac{1}{2\alpha}\|H_c(x)-H_0(x)\|^2.
\]
Moreover, \((H_c)_\#\mu=(H_0)_\#\mu=\nu\), and hence
\[
\int f_0^*(H_c)\,d\mu=\int f_0^*(H_0)\,d\mu.
\]
Consequently,
\begin{equation}
\label{eq:stability-map-first}
0
\ge
\int\langle x,H_c-H_0\rangle\,d\mu
+
\frac{1}{2\alpha}\|H_c-H_0\|_{L^2(\mu)}^2.
\end{equation}
Since \(H_c\) and \(H_0\) have the same push-forward,
\[
\int\|H_c(x)\|^2\,d\mu(x)
=
\int\|H_0(x)\|^2\,d\mu(x),
\]
and therefore
\[
\int\langle x,H_c-H_0\rangle\,d\mu
=
\frac12\int\|x-H_0(x)\|^2\,d\mu
-
\frac12\int\|x-H_c(x)\|^2\,d\mu.
\]
Let \(\pi_c=(I\times H_c)_\#\mu\). Then
\[
\frac12\int\|x-H_c(x)\|^2\,d\mu
=
\mathcal T_{\rho_c}(\mu,\nu)
+
\int
\left(
\frac12\|x-y\|^2-\rho_c(\|x-y\|)
\right)
\,d\pi_c(x,y).
\]
Thus \eqref{eq:stability-map-first} yields
\begin{align*}
0
&\ge
\frac12\mathcal W_2^2(\mu,\nu)-\mathcal T_{\rho_c}(\mu,\nu)\\
&\quad
-
\int
\left(
\frac12\|x-y\|^2-\rho_c(\|x-y\|)
\right)
\,d\pi_c(x,y)
+
\frac{1}{2\alpha}\|H_c-H_0\|_{L^2(\mu)}^2.
\end{align*}
Since
\(\frac12\mathcal W_2^2(\mu,\nu)\ge\mathcal T_{\rho_c}(\mu,\nu)\), it follows that
\[
\frac{1}{2\alpha}\|H_c-H_0\|_{L^2(\mu)}^2
\le
\int
\left(
\frac12\|x-y\|^2-\rho_c(\|x-y\|)
\right)
\,d\pi_c(x,y).
\]
Arguing as in the proof of \Cref{propositon:Stable-cost-quantitative}, we have
\[
\frac12\|x-y\|^2-\rho_c(\|x-y\|)
\le
\left(\|x\|^2-\frac{c^2}{4}\right)_+
+
\left(\|y\|^2-\frac{c^2}{4}\right)_+.
\]
Using the marginals of \(\pi_c\), we conclude that
\begin{align*}
\|H_c-H_0\|_{L^2(\mu)}^2
\le
2\alpha\bigg[
&\int\left(\|x\|^2-\frac{c^2}{4}\right)_+\,d\mu(x)\\
&+
\int\left(\|y\|^2-\frac{c^2}{4}\right)_+\,d\nu(y)
\bigg],
\end{align*}
which proves the result.
\end{proof}
\subsection{Proofs in \cref{sect:Barycenters}}
\begin{proof}[Proof of \cref{propo:finite}]
Since \(\rho_c(t)\le ct\) for every \(t\ge0\), for any
\(\pi\in\Pi(\mu,\nu)\),
\[
\int\rho_c(\|x-y\|)\,d\pi(x,y)
\le
c\int\|x-y\|\,d\pi(x,y)
\le
c\int\|x\|\,d\mu(x)
+
c\int\|y\|\,d\nu(y).
\]
Taking the infimum over \(\pi\in\Pi(\mu,\nu)\) gives
\[
\mathcal T_{\rho_c}(\nu,\mu)
-
c\int\|x\|\,d\mu(x)
\le
c\int\|y\|\,d\nu(y).
\]

On the other hand, since
\(\rho_c(t)\ge ct-c^2/2\), the reverse triangle inequality yields, for every
\(\pi\in\Pi(\mu,\nu)\),
\begin{align*}
\int\rho_c(\|x-y\|)\,d\pi(x,y)
&\ge
c\int\|x-y\|\,d\pi(x,y)-\frac{c^2}{2}\\
&\ge
c\int\|x\|\,d\mu(x)
-
c\int\|y\|\,d\nu(y)
-
\frac{c^2}{2}.
\end{align*}
Taking the infimum over \(\pi\in\Pi(\mu,\nu)\), we obtain
\[
\mathcal T_{\rho_c}(\nu,\mu)
-
c\int\|x\|\,d\mu(x)
\ge
-
c\int\|y\|\,d\nu(y)
-
\frac{c^2}{2}.
\]
Combining the two inequalities proves \eqref{eq:bound-integrand-bary}.

Finally,
\begin{align*}
|\Gamma_{c,\mathbf P}(\nu)|
&\le
\int
\left|
\mathcal T_{\rho_c}(\nu,\mu)
-
c\int\|x\|\,d\mu(x)
\right|
\,d\mathbf P(\mu)\\
&\le
c\int\|y\|\,d\nu(y)+\frac{c^2}{2}<\infty,
\end{align*}
which proves the second claim.
\end{proof}

\begin{proof}[Proof of \cref{theorem:Stability}]
We first prove (i). Fix \(\nu\) such that
\(\Gamma_{c,{\bf P}_n}(\nu)\leq\varepsilon\).
Since \({\bf P}_n\to{\bf P}\), the family
\(\{{\bf P}_n:n\in\mathbb N\}\cup\{{\bf P}\}\) is tight. Hence, there exists a
compact set \(\mathcal K\subset\mathcal P_1(\mathbb R^d)\), independent of
\(n\), such that $\inf_{n\in\mathbb N}{\bf P}_n(\mathcal K)\geq\frac34.$ 
Hence,
\begin{multline}
\label{eq:bound-in-K}
\int_{\mathcal P_1(\mathbb R^d)\setminus\mathcal K}
\left(
c\int\|x\|\,d\mu(x)-\mathcal T_{\rho_c}(\nu,\mu)
\right)d{\bf P}_n(\mu)
+\varepsilon
\\
\geq
\int_{\mathcal K}
\left(
\mathcal T_{\rho_c}(\nu,\mu)
-c\int\|x\|\,d\mu(x)
\right)d{\bf P}_n(\mu).
\end{multline}
By \eqref{eq:bound-integrand-bary},
\begin{align}
\label{bound-ink-left}
&\int_{\mathcal P_1(\mathbb R^d)\setminus\mathcal K}
\left(
c\int\|x\|\,d\mu(x)-\mathcal T_{\rho_c}(\nu,\mu)
\right)d{\bf P}_n(\mu)
\nonumber\\
&\qquad\leq
\bigl(1-{\bf P}_n(\mathcal K)\bigr)
\left(
c\int\|y\|\,d\nu(y)+\frac{c^2}{2}
\right)
\nonumber\\
&\qquad\leq
\frac c4\int\|y\|\,d\nu(y)+\frac{c^2}{8}.
\end{align}
Furthermore, since \(\rho_c(t)\geq ct-c^2/2\), we obtain
\[
\mathcal T_{\rho_c}(\nu,\mu)
\geq
c\mathcal W_1(\nu,\mu)-\frac{c^2}{2}
\geq
c\int\|y\|\,d\nu(y)
-c\int\|x\|\,d\mu(x)
-\frac{c^2}{2},
\]
where the last inequality follows from the fact that
\(x\mapsto\|x\|\) is \(1\)-Lipschitz. Consequently,
\[
\mathcal T_{\rho_c}(\nu,\mu)
-c\int\|x\|\,d\mu(x)
\geq
c\int\|y\|\,d\nu(y)
-2c\int\|x\|\,d\mu(x)
-\frac{c^2}{2}.
\]
Since \(\mathcal K\) is compact in \(\mathcal P_1(\mathbb R^d)\),
\[
K_1:=\sup_{\mu\in\mathcal K}\int\|x\|\,d\mu(x)<\infty.
\]
Setting \(C_{\mathcal K}:=2cK_1+c^2/2\), we obtain, for every
\(\mu\in\mathcal K\),
\[
\mathcal T_{\rho_c}(\nu,\mu)
-c\int\|x\|\,d\mu(x)
\geq
c\int\|y\|\,d\nu(y)-C_{\mathcal K}.
\]
Thus,
\begin{align}
\label{bound-ink-right}
&\int_{\mathcal K}
\left(
\mathcal T_{\rho_c}(\nu,\mu)
-c\int\|x\|\,d\mu(x)
\right)d{\bf P}_n(\mu)
\nonumber\\
&\geq
{\bf P}_n(\mathcal K)
\left(
c\int\|y\|\,d\nu(y)-C_{\mathcal K}
\right)
\nonumber\\
&\qquad\geq
\frac{3c}{4}\int\|y\|\,d\nu(y)-C_{\mathcal K}.
\end{align}
Combining \eqref{eq:bound-in-K}, \eqref{bound-ink-left}, and
\eqref{bound-ink-right}, we derive
\begin{equation}
\label{eq:bound-integral-nu}
\int\|y\|\,d\nu(y)
\leq
M(\varepsilon)
:=
\frac{2}{c}
\left(
C_{\mathcal K}+\frac{c^2}{8}+\varepsilon
\right).
\end{equation}
Moreover, \eqref{eq:bound-integrand-bary} gives
\[
\Gamma_{c,{\bf P}_n}(\nu)
\geq
-c\int\|y\|\,d\nu(y)-\frac{c^2}{2}.
\]
Together with \eqref{eq:bound-integral-nu}, this proves (i), after enlarging
the constant \(C\) if necessary.

We now prove (ii). Notice that
\begin{align}
\label{eq:value-diract-at-0}
\Gamma_{c,{\bf P}_n}(\nu_n)
&\leq
\Gamma_{c,{\bf P}_n}(\delta_0)+\varepsilon_n
\nonumber\\
&=
\int_{\mathcal P_1(\mathbb R^d)}
\left[
\int_{\mathbb R^d}
\bigl(\rho_c(\|x\|)-c\|x\|\bigr)\,d\mu(x)
\right]d{\bf P}_n(\mu)
+\varepsilon_n
\nonumber\\
&\leq\varepsilon_n.
\end{align}
Since \(\varepsilon_n\to0\), the sequence \(\{\varepsilon_n\}_n\) is bounded.
Thus, \eqref{eq:bound-integral-nu} yields $\sup_{n\in\mathbb N}\int\|y\|\,d\nu_n(y)\leq M$ 
for some \(M<\infty\). Define
\[
\mathcal M_M
:=
\left\{
\nu\in\mathcal P(\mathbb R^d):
\int\|y\|\,d\nu(y)\leq M
\right\},
\]
endowed with the topology of weak convergence.
The set \(\mathcal M_M\) is weakly closed: indeed, if
\(\gamma_n\rightharpoonup\gamma\) and \(\gamma_n\in\mathcal M_M\), then the
Portmanteau theorem gives
\[
\int\|y\|\,d\gamma(y)
\leq
\liminf_{n\to\infty}\int\|y\|\,d\gamma_n(y)
\leq M.
\]
Moreover, by Markov's inequality, for every \(r>0\),
\[
\gamma(\mathbb R^d\setminus B_r(0))
\leq
\frac1r\int\|y\|\,d\gamma(y)
\leq
\frac Mr,
\qquad \gamma\in\mathcal M_M.
\]
Thus \(\mathcal M_M\) is tight and weakly closed, and is therefore compact by
Prokhorov's theorem. 
Since \(\{\nu_n\}_n\subset\mathcal M_M\), the sequence is relatively compact
in \(\mathcal P(\mathbb R^d)\).

We now show that every limit point belongs to \(m_c({\bf P})\). Passing to a
subsequence, assume that \(\nu_n\rightharpoonup\nu_*\). By the preceding
argument, \(\nu_*\in\mathcal M_M\subset\mathcal P_1(\mathbb R^d)\).

For fixed \(\nu\in\mathcal M_M\), the function
\[
\mathcal P_1(\mathbb R^d)\ni\mu
\longmapsto
\mathcal T_{\rho_c}(\nu,\mu)
-c\int\|x\|\,d\mu(x)
\]
is bounded by \eqref{eq:bound-integrand-bary} and continuous.
Consequently, \(\Gamma_{c,{\bf P}_n}(\nu)\to
\Gamma_{c,{\bf P}}(\nu)\). By the approximate minimality of \(\nu_n\), for every
\(\nu\in\mathcal P_1(\mathbb R^d)\),
\[
\Gamma_{c,{\bf P}_n}(\nu_n)
\leq
\Gamma_{c,{\bf P}_n}(\nu)+\varepsilon_n,
\]
and therefore
\begin{equation}
\label{eq:limsup-before-fatou}
\limsup_{n\to\infty}\Gamma_{c,{\bf P}_n}(\nu_n)
\leq
\Gamma_{c,{\bf P}}(\nu).
\end{equation}

Define
\[
\beta(\mu,\nu)
:=
\mathcal T_{\rho_c}(\nu,\mu)
-c\int\|x\|\,d\mu(x),
\qquad
(\mu,\nu)\in
\mathcal P_1(\mathbb R^d)\times\mathcal M_M.
\]
The function \(\beta\) is lower semicontinuous and, by
\eqref{eq:bound-integrand-bary}, satisfies $\beta(\mu,\nu)\geq-cM-\frac{c^2}{2}.$ 
Since \({\bf P}_n\to{\bf P}\) and \(\nu_n\rightharpoonup\nu_*\), we have
\[
{\bf P}_n\otimes\delta_{\nu_n}
\longrightarrow
{\bf P}\otimes\delta_{\nu_*}
\quad\text{in}\quad
\mathcal P\bigl(
\mathcal P_1(\mathbb R^d)\times\mathcal M_M
\bigr).
\]
Therefore, the Portmanteau theorem yields
\begin{align*}
\liminf_{n\to\infty}\Gamma_{c,{\bf P}_n}(\nu_n)
&=
\liminf_{n\to\infty}
\int\beta(\mu,\nu)\,
d({\bf P}_n\otimes\delta_{\nu_n})(\mu,\nu)\\
&\geq
\int\beta(\mu,\nu)\,
d({\bf P}\otimes\delta_{\nu_*})(\mu,\nu)=
\Gamma_{c,{\bf P}}(\nu_*).
\end{align*}
Combining this inequality with \eqref{eq:limsup-before-fatou}, we obtain $\Gamma_{c,{\bf P}}(\nu_*)
\leq
\Gamma_{c,{\bf P}}(\nu)$ for every $\nu\in\mathcal P_1(\mathbb R^d).$ 
Hence \(\nu_*\in m_c({\bf P})\), which concludes the proof.
\end{proof}
\begin{proof}[Proof of \cref{theorem:Existence-General}]
For every \(\mu\in\mathcal P_1(\mathbb R^d)\), the map
\(\nu\mapsto\mathcal T_{\rho_c}(\nu,\mu)\) is convex. Indeed, if
\(\pi_i\in\Pi(\nu_i,\mu)\), \(i=0,1\), and \(t\in[0,1]\), then
\((1-t)\pi_0+t\pi_1\in\Pi((1-t)\nu_0+t\nu_1,\mu)\). Taking infima shows that
\[
\mathcal T_{\rho_c}((1-t)\nu_0+t\nu_1,\mu)
\leq
(1-t)\mathcal T_{\rho_c}(\nu_0,\mu)
+t\mathcal T_{\rho_c}(\nu_1,\mu).
\]
Thus \(\Gamma_{c,{\bf P}}\) is convex, and consequently
\(m_c({\bf P})\) is convex.

To prove existence, first notice that \(\Gamma_{c,{\bf P}}\) is bounded from
below. Indeed, if \(\Gamma_{c,{\bf P}}(\nu)>0\), this is immediate, while if
\(\Gamma_{c,{\bf P}}(\nu)\leq0\), \cref{theorem:Stability}~(i), applied with
\(\varepsilon=0\), gives
\(\Gamma_{c,{\bf P}}(\nu)\geq-C\). Hence
\[
-\infty<
\inf_{\nu\in\mathcal P_1(\mathbb R^d)}
\Gamma_{c,{\bf P}}(\nu)
<\infty.
\]
Choose a minimizing sequence \(\{\nu_n\}_n\) such that
\[
\Gamma_{c,{\bf P}}(\nu_n)
\leq
\Gamma_{c,{\bf P}}(\nu)+\frac1n
\qquad
\text{for every }\nu\in\mathcal P_1(\mathbb R^d).
\]
Applying \cref{theorem:Stability}~(ii) with
\({\bf P}_n={\bf P}\) and \(\varepsilon_n=1/n\), we obtain a subsequence
converging weakly to some \(\nu_*\in m_c({\bf P})\). Thus
\(m_c({\bf P})\neq\varnothing\). Furthermore, \cref{theorem:Stability}~(i), applied with
\(\varepsilon=0\), gives a uniform bound on
\(\int\|x\|\,d\nu(x)=\mathcal W_1(\nu,\delta_0)\) for
\(\nu\in m_c({\bf P})\). Hence \(m_c({\bf P})\) is bounded in
\(\mathcal W_1\). Finally, let \(\nu_n\in m_c({\bf P})\) and assume that
\(\mathcal W_1(\nu_n,\nu)\to0\). Applying
\cref{theorem:Stability}~(ii) with \({\bf P}_n={\bf P}\) and
\(\varepsilon_n=0\), we obtain \(\nu\in m_c({\bf P})\). Therefore
\(m_c({\bf P})\) is closed in \(\mathcal W_1\).
\end{proof}
\begin{proof}[Proof of \cref{corollary:consistency-Huber-barycenters}]
Since \(\mathcal P_1(\mathbb R^d)\), endowed with \(\mathcal W_1\), is a
Polish space, \cite[Theorem~3]{Varadarajan} gives
\[
{\bf P}_n\longrightarrow{\bf P}
\qquad\text{in }
\mathcal P\bigl(\mathcal P_1(\mathbb R^d)\bigr)
\quad\text{almost surely}.
\]
Fix an outcome in the probability-one event on which this convergence
holds. By \cref{theorem:Stability},
\(\{\widehat\nu_{n,c}\}_n\) is relatively compact in
\(\mathcal P(\mathbb R^d)\), and every limit point belongs to
\(m_c({\bf P})\). If \(m_c({\bf P})=\{\nu_c\}\), every subsequence of
\(\{\widehat\nu_{n,c}\}_n\) has a further subsequence converging in
\(\mathcal P(\mathbb R^d)\) to an element of \(m_c({\bf P})\), which must
be \(\nu_c\). Therefore,
\[
\widehat\nu_{n,c}\longrightarrow\nu_c
\qquad\text{in }\mathcal P(\mathbb R^d)
\quad\text{almost surely}.
\]
We now consider the two-stage model. We first prove that
\begin{equation}
\label{eq:two-stage-empirical-convergence}
\widehat{\bf P}_{n,\mathbf N}
\longrightarrow{\bf P}
\qquad\text{in }
\mathcal P\bigl(\mathcal P_1(\mathbb R^d)\bigr)
\quad\text{almost surely}.
\end{equation}
Let $\mathrm{BL}_1(X)$ be the set of continuous, bounded (by $1$) functions   and $1$-Lipschitz (with respect to $\mathcal{W}_1$) functions $f:X\to [-1,1]$, where $X\subseteq\mathcal{P}_1(\mathbb{R}^d)$. Let $f\in  \mathrm{BL}_1(\mathcal P_1(\mathbb R^d)\bigr)$; for each fixed \(n\), the random variables
\(f\bigl(\widehat\mu_{1,N}\bigr),\ldots,f\bigl(\widehat\mu_{n,N}\bigr) \) are iid and satisfy $|f\bigl(\widehat\mu_{i,N}\bigr)|\leq\|f\|_\infty$ for all $1\leq i\leq n$. Hence, for every \(\varepsilon>0\), Hoeffding's inequality gives
\[
\mathbb P\left(
\left|
\frac1n\sum_{i=1}^n
\biggl(f\bigl(\widehat\mu_{i,N}\bigr)-\mathbb E \biggl[f\bigl(\widehat\mu_{i,N}\bigr)\biggr]\biggr)
\right|>\varepsilon
\right)
\leq
2\exp\left(
-\frac{n\varepsilon^2}{2}
\right).
\]
The right-hand side is summable in \(n\). Therefore, by the
Borel--Cantelli lemma,
\begin{equation}
\label{eq:two-stage-centered}
\frac1n\sum_{i=1}^n
\bigl(f\bigl(\widehat\mu_{i,N}\bigr)-\mathbb E \bigl[f\bigl(\widehat\mu_{1,N}\bigr)\bigr]\bigr)=\frac1n\sum_{i=1}^n\bigl(f\bigl(\widehat\mu_{i,N}\bigr)-\mathbb E \bigl[f\bigl(\widehat\mu_{i,N}\bigr)\bigr]\bigr)
\longrightarrow0
\qquad\text{almost surely}.
\end{equation}
Furthermore,  by 
$$ |\mathbb E[ f\bigl(\widehat\mu_{i,N}\bigr)] - \mathbb E[ f\bigl(\mu_{i}\bigr)] |  \leq \mathbb E[ \min(2, \mathcal{W}_1( \widehat\mu_{i,N}, \mu_{i})) ] . $$
We conclude that $\mathbb E[ \min(2, \mathcal{W}_1( \widehat\mu_{i,N}, \mu_{i})) ]\to 0$ for all $1\leq i\leq n$ by arguing as in the proof of \cite[Lemma~6.3]{bachoc2026wassersteinspatialdepth}. Hence, \eqref{eq:two-stage-centered} yields
\[
\frac{1}{n}\sum_{i=1}^n
\left(
f(\widehat{\mu}_{i,N})-\mathbb{E}[f(\mu_1)]
\right)
\to 0
\qquad \text{almost surely},
\]
for every fixed
\(f\in \mathrm{BL}_1(\mathcal{P}_1(\mathbb{R}^d))\). We now prove that the convergence is uniform over
\(\mathrm{BL}_1(\mathcal{P}_1(\mathbb{R}^d))\).
Fix \(\varepsilon>0\). Since the law of \(\mu_1\) is tight, there exists
a compact set \(K\subset \mathcal{P}_1(\mathbb{R}^d)\) such that $\mathbb{P}(\mu_1\in K)\geq 1-\varepsilon.$  By the Arzelà--Ascoli theorem, \(\mathrm{BL}_1(K)\) is compact in the
uniform norm. Thus, there exist \(f_1,\ldots,f_M\in
\mathrm{BL}_1(\mathcal{P}_1(\mathbb{R}^d))\) such that, for every
\(f\in \mathrm{BL}_1(\mathcal{P}_1(\mathbb{R}^d))\), there is some
\(j\in\{1,\ldots,M\}\) satisfying
\[
\sup_{\mu\in K}|f(\mu)-f_j(\mu)|\leq \varepsilon.
\]
For such \(f\) and \(f_j\),
\begin{align*}
\left|
\frac{1}{n}\sum_{i=1}^n f(\widehat{\mu}_{i,N})
-\mathbb{E}[f(\mu_1)]
\right|
&\leq
\left|
\frac{1}{n}\sum_{i=1}^n f_j(\widehat{\mu}_{i,N})
-\mathbb{E}[f_j(\mu_1)]
\right| \\
&\quad+
\frac{1}{n}\sum_{i=1}^n
|f(\widehat{\mu}_{i,N})-f_j(\widehat{\mu}_{i,N})| \\
&\quad+
\mathbb{E}|f(\mu_1)-f_j(\mu_1)|.
\end{align*}
Since \(f\) and \(f_j\) are \(1\)-Lipschitz and bounded by $1$, 
\[
|f(\widehat{\mu}_{i,N})-f_j(\widehat{\mu}_{i,N})|
\leq
2\min(1, \mathcal{W}_1(\widehat{\mu}_{i,N},\mu_i))
+
|f(\mu_i)-f_j(\mu_i)|.
\]
Moreover,
\[
|f(\mu_i)-f_j(\mu_i)|
\leq
\varepsilon+2\mathbf{1}_{\{\mu_i\notin K\}},
\]
and
\[
\mathbb{E}|f(\mu_1)-f_j(\mu_1)|
\leq
\varepsilon+2\mathbb{P}(\mu_1\notin K)
\leq 3\varepsilon.
\]
Therefore,
\begin{multline}
    \label{eq:decomop-two-stage}
    \sup_{f\in \mathrm{BL}_1(\mathcal{P}_1(\mathbb{R}^d))}
\left|
\frac{1}{n}\sum_{i=1}^n f(\widehat{\mu}_{i,N})
-\mathbb{E}[f(\mu_1)]
\right| 
\\
\leq
\max_{1\leq j\leq M}
\left|
\frac{1}{n}\sum_{i=1}^n f_j(\widehat{\mu}_{i,N})
-\mathbb{E}[f_j(\mu_1)]
\right| +
4\varepsilon
\\+
\frac{2}{n}\sum_{i=1}^n
\mathbf{1}_{\{\mu_i\notin K\}}+
\frac{4}{n}\sum_{i=1}^n
\min(1, \mathcal{W}_1(\widehat{\mu}_{i,N},\mu_i)).
\end{multline}
By \eqref{eq:two-stage-centered},
\begin{equation}
    \label{eq:decomop-two-stage-2}
    \max_{1\leq j\leq M}
\left|
\frac{1}{n}\sum_{i=1}^n f_j(\widehat{\mu}_{i,N})
-\mathbb{E}[f_j(\mu_1)]
\right|
\to 0
\qquad \text{almost surely}.
\end{equation}
Also, by the strong law of large numbers,
\[
\frac{1}{n}\sum_{i=1}^n
\mathbf{1}_{\{\mu_i\notin K\}}
\to
\mathbb{P}(\mu_1\notin K)
\leq \varepsilon
\qquad \text{almost surely}.
\]
Since, for every \(i\), $\min\left( 2, \mathcal{W}_1\bigl(\widehat{\mu}_{i,N},\mu_i\bigr) \right)
\to 0$, a.s.,  
by \cite[Theorem~3]{Varadarajan} and the convergence of the first moments, the dominated convergence theorem gives
\begin{equation}
    \label{eq:two-stage-L1}
    \mathbb{E}\left[
\min\left( 2, \mathcal{W}_1\bigl(\widehat{\mu}_{i,N},\mu_i\bigr) \right)
\right]
\to 0.
\end{equation}
For every \(n\), the random variables $\{\min\left( 2, \mathcal{W}_1\bigl(\widehat{\mu}_{i,N},\mu_i\bigr) \right)\}_{i=1}^n$ 
are independent, identically distributed, and bounded between \(0\)
and \(2\). Hence, Hoeffding's inequality yields, for every
\(\varepsilon>0\),
\[
\begin{aligned}
\mathbb{P}\Bigg(
\Bigg|
\frac{1}{n}\sum_{i=1}^n
\min\left( 2, \mathcal{W}_1\bigl(\widehat{\mu}_{i,N},\mu_i\bigr) \right)
&-
\mathbb{E}\left[
\min\left( 2, \mathcal{W}_1\bigl(\widehat{\mu}_{1,N},\mu_1\bigr) \right)
\right]
\Bigg|
>\varepsilon
\Bigg) \leq
2\exp\left(-\frac{n\varepsilon^2}{2}\right).
\end{aligned}
\]
Therefore,
\[
\sum_{n=1}^{\infty}
\mathbb{P}\Bigg(
\Bigg|
\frac{1}{n}\sum_{i=1}^n
\min\left( 2, \mathcal{W}_1\bigl(\widehat{\mu}_{i,N},\mu_i\bigr) \right)
-
\mathbb{E}\left[
\min\left( 2, \mathcal{W}_1\bigl(\widehat{\mu}_{1,N},\mu_1\bigr) \right)
\right]
\Bigg|
>\varepsilon
\Bigg) 
<\infty.
\]
By the Borel--Cantelli lemma,
\[
\frac{1}{n}\sum_{i=1}^n
\min\left( 2, \mathcal{W}_1\bigl(\widehat{\mu}_{i,N},\mu_i\bigr) \right)
-
\mathbb{E}\left[
\min\left( 2, \mathcal{W}_1\bigl(\widehat{\mu}_{1,N},\mu_1\bigr) \right)
\right]
\to 0
\qquad\text{almost surely}.
\]
By \eqref{eq:two-stage-L1}, 
we conclude that
\[
\frac{1}{n}\sum_{i=1}^n
\min\left( 2, \mathcal{W}_1\bigl(\widehat{\mu}_{i,N},\mu_i\bigr) \right)
\to 0
\qquad\text{almost surely}.
\]
This  together with \eqref{eq:decomop-two-stage} and \eqref{eq:decomop-two-stage-2} yields
\[
\limsup_{n\to\infty}
\sup_{f\in \mathrm{BL}_1(\mathcal{P}_1(\mathbb{R}^d))}
\left|
\frac{1}{n}\sum_{i=1}^n f(\widehat{\mu}_{i,N})
-\mathbb{E}[f(\mu_1)]
\right|
\leq 6\varepsilon
\]
almost surely. Since \(\varepsilon>0\) is arbitrary,
\[
\sup_{f\in \mathrm{BL}_1(\mathcal{P}_1(\mathbb{R}^d))}
\left|
\frac{1}{n}\sum_{i=1}^n f(\widehat{\mu}_{i,N})
-\mathbb{E}[f(\mu_1)]
\right|
\to 0
\qquad \text{almost surely},
\]
and thus \eqref{eq:two-stage-empirical-convergence} holds. Applying
\Cref{theorem:Stability} with
${\mathbf P}_n=\widehat{\mathbf P}_{n,N(n)}$
shows that the empirical barycenters are relatively compact and that
all their limit points belong to \(m_c({\bf P})\). If
\(m_c({\bf P})=\{\nu_c\}\), the usual subsequence argument yields
\(\widehat\nu_{n,N(n),c}\to\nu_c\) almost surely.
\end{proof}
\begin{proof}[Proof of \cref{theorem:stability-in-c}]
\medskip
\noindent
We first show (i). Define
\[
\mathcal G_c(\nu)
:=
\frac1c\Gamma_{c,{\bf P}}(\nu)
=
\int_{\mathcal P_1(\mathbb R^d)}
\left[
\frac1c\mathcal T_{\rho_c}(\nu,\mu)-m_1(\mu)
\right]\,d{\bf P}(\mu).
\]
Clearly, $\argmin\mathcal G_c=m_c({\bf P}).$ For every \(t\geq0\),
\begin{equation}
\label{eq:rho-c-L1-uniform}
0
\leq
t-\frac{\rho_c(t)}{c}
\leq
\frac c2,
\end{equation}
as if \(t\leq c\), then
$t-\frac{\rho_c(t)}c
=
t-\frac{t^2}{2c},$
whose values lie in \([0,c/2]\), whereas for \(t>c\),
$t-\frac{\rho_c(t)}c=\frac c2.$
It follows that, for every
\(\nu,\mu\in\mathcal P_1(\mathbb R^d)\),
\begin{equation}
\label{eq:transport-Huber-W1}
0
\leq
\mathcal W_1(\nu,\mu)
-
\frac1c\mathcal T_{\rho_c}(\nu,\mu)
\leq
\frac c2.
\end{equation}
Indeed, lower bound inequality follows from
\(\rho_c(t)/c\leq t\). For the upper bound, integrate
\eqref{eq:rho-c-L1-uniform} against an arbitrary coupling and then take
infima. After subtracting \(m_1(\mu)\) and integrating with respect to
\({\bf P}\), we obtain the uniform estimate
\begin{equation}
\label{eq:F1-Gc}
0
\leq
\mathcal F_1(\nu)-\mathcal G_c(\nu)
\leq
\frac c2
\qquad
\text{for every }\nu\in\mathcal P_1(\mathbb R^d).
\end{equation}
Since \(\nu_n\in m_{c_n}({\bf P})\), for every
\(\eta\in\mathcal P_1(\mathbb R^d)\),
\begin{align}
\mathcal F_1(\nu_n)
&\leq
\mathcal G_{c_n}(\nu_n)+\frac{c_n}{2}
\nonumber\\
&\leq
\mathcal G_{c_n}(\eta)+\frac{c_n}{2}\leq
\mathcal F_1(\eta)+\frac{c_n}{2}.
\label{eq:approximate-F1-minimizer}
\end{align}
In particular, since $\mathcal F_1(\delta_0)=0,$ 
we have
\begin{equation}
\label{eq:F1-nun-upper}
\mathcal F_1(\nu_n)\leq\frac{c_n}{2}.
\end{equation}

We next establish weak relative compactness. Since \({\bf P}\) is a
probability measure on the Polish space
\(\mathcal P_1(\mathbb R^d)\), it is tight. Hence there exists a compact set
\(K\subset\mathcal P_1(\mathbb R^d)\) such that
$p:={\bf P}(K)>\frac12.$ 
Because \(K\) is compact in \(\mathcal W_1\),
\[
M_K:=\sup_{\mu\in K}m_1(\mu)<\infty.
\]
For \(\mu\in K\), the triangle inequality gives
$\mathcal W_1(\nu,\mu)
\geq
m_1(\nu)-m_1(\mu),$
and thus
\[
\mathcal W_1(\nu,\mu)-m_1(\mu)
\geq
m_1(\nu)-2M_K.
\]
Since  $\mathcal W_1(\nu,\mu)-m_1(\mu)
\geq
-m_1(\nu), $
we derive
\begin{align}
\mathcal F_1(\nu)
&\geq
p\bigl(m_1(\nu)-2M_K\bigr)
-(1-p)m_1(\nu)
\nonumber\\
&=
(2p-1)m_1(\nu)-2pM_K.
\label{eq:F1-weak-coercivity}
\end{align}
Combining \eqref{eq:F1-nun-upper} and
\eqref{eq:F1-weak-coercivity} yields
$\sup_n m_1(\nu_n)<\infty.$ 
By Markov's inequality, \(\{\nu_n\}_n\) is tight, and hence
relatively compact in the weak topology. Let \(\nu_{n_k}\to \nu_*\) in distribution. For a fixed \(\mu\),
lower semicontinuity of the transport cost gives
\[
\mathcal W_1(\nu_*,\mu)
\leq
\liminf_{k\to\infty}
\mathcal W_1(\nu_{n_k},\mu).
\]
Moreover,
\[
\mathcal W_1(\nu_{n_k},\mu)-m_1(\mu)
\geq
-m_1(\nu_{n_k})
\geq -C
\]
for some constant \(C<\infty\). Fatou's lemma therefore gives
\begin{equation}
\label{eq:F1-lsc-limit}
\mathcal F_1(\nu_*)
\leq
\liminf_{k\to\infty}\mathcal F_1(\nu_{n_k}).
\end{equation}
On the other hand, \eqref{eq:approximate-F1-minimizer} implies that, for
every \(\eta\in\mathcal P_1(\mathbb R^d)\),
\[
\limsup_{k\to\infty}\mathcal F_1(\nu_{n_k})
\leq
\mathcal F_1(\eta).
\]
Consequently, $\mathcal F_1(\nu_*)
\leq
\mathcal F_1(\eta)$ for every  $\eta\in\mathcal P_1(\mathbb R^d),$ 
and hence $\nu_*\in m_0({\bf P}).$ If \(m_0({\bf P})=\{\nu_0\}\), every subsequence of \(\{\nu_n\}\) admits a
further subsequence converging weakly to \(\nu_0\) and the result follows. 

\medskip
\noindent
We now show (ii). For \(a>0\), define
\[
\mathcal J_a(\nu)
:=
\int_{\mathcal P_2(\mathbb R^d)}
\mathcal T_{\rho_a}(\nu,\mu)\,d{\bf P}(\mu).
\]
The centering term in \(\Gamma_{a,{\bf P}}\) does not depend on \(\nu\);
therefore,
$m_a({\bf P})=\argmin_\nu\mathcal J_a(\nu).$
By optimality of \(\nu_n\) and comparison with \(\delta_0\),
\begin{align}
\mathcal J_{c_n}(\nu_n)
&\leq
\mathcal J_{c_n}(\delta_0)
\nonumber\\
&=
\int
\int \rho_{c_n}(\|x\|)\,d\mu(x)\,d{\bf P}(\mu)
\nonumber\\
&\leq
\frac12
\int
\int\|x\|^2\,d\mu(x)\,d{\bf P}(\mu)
=:C_2<\infty.
\label{eq:Jcn-uniform}
\end{align}
Fix \(a>0\). Since \(c\mapsto\rho_c(t)\) is nondecreasing, for all
sufficiently large \(n\),
\[
\mathcal J_a(\nu_n)
\leq
\mathcal J_{c_n}(\nu_n)
\leq C_2.
\]
Choose a compact set \(K\subset\mathcal P_2(\mathbb R^d)\) such that
$p:={\bf P}(K)>0.$
Since the embedding $\mathcal P_2(\mathbb R^d)\hookrightarrow\mathcal P_1(\mathbb R^d)$ 
is continuous,
\[
M_K:=\sup_{\mu\in K}m_1(\mu)<\infty.
\]
Using $\rho_a(t)\geq at-\frac{a^2}{2},$ 
we obtain
\begin{align*}
\mathcal T_{\rho_a}(\nu,\mu)
&\geq
a\mathcal W_1(\nu,\mu)-\frac{a^2}{2}\geq
a\bigl(m_1(\nu)-m_1(\mu)\bigr)-\frac{a^2}{2}.
\end{align*}
Since \(\mathcal T_{\rho_a}\geq0\),
\begin{equation}
\label{eq:Ja-coercivity}
\mathcal J_a(\nu)
\geq
\int_K\mathcal T_{\rho_a}(\nu,\mu)\,d{\bf P}(\mu)\geq
p\left[
a\bigl(m_1(\nu)-M_K\bigr)-\frac{a^2}{2}
\right].
\end{equation}
Applying this inequality to \(\nu_n\) shows that $\sup_n m_1(\nu_n)<\infty.$ 
Thus \(\{\nu_n\}_n\) is tight and relatively compact in the weak topology. Let \(\nu_{n_k}\to\nu_*\) in distribution. For every fixed \(a>0\), eventually
\(c_{n_k}\geq a\), and hence
$\mathcal J_{c_{n_k}}(\nu_{n_k})
\geq
\mathcal J_a(\nu_{n_k}).$
By lower semicontinuity of
\(\nu\mapsto\mathcal T_{\rho_a}(\nu,\mu)\) and Fatou's lemma,
\[
\mathcal J_a(\nu_*)
\leq
\liminf_{k\to\infty}
\mathcal J_a(\nu_{n_k})
\leq
\liminf_{k\to\infty}
\mathcal J_{c_{n_k}}(\nu_{n_k}).
\]
Taking the supremum over \(a>0\) gives
\begin{equation}
\label{eq:large-c-liminf}
\sup_{a>0}\mathcal J_a(\nu_*)
\leq
\liminf_{k\to\infty}
\mathcal J_{c_{n_k}}(\nu_{n_k}).
\end{equation}
We claim that
\begin{equation}
\label{eq:monotone-cost-limit}
\sup_{a>0}\mathcal T_{\rho_a}(\nu,\mu)
=
\frac12\mathcal W_2^2(\nu,\mu),
\end{equation}
with the convention that the right-hand side is \(+\infty\) when one of
the measures does not belong to \(\mathcal P_2\). Indeed,
\[
\rho_a(t)\uparrow\frac12t^2
\qquad\text{as }a\uparrow\infty,
\]
from which the ``\(\leq\)'' inequality in \eqref{eq:monotone-cost-limit} follows. For the reverse inequality, let \(a_j\uparrow\infty\) and let
\(\pi_j\in\Pi(\nu,\mu)\) be optimal for \(\rho_{a_j}\). The family
\(\{\pi_j\}\) is tight because its marginals are fixed. Passing to a weakly
convergent subsequence, say \(\pi_j\to\pi\in\Pi(\nu,\mu)\) in distribution,
for every fixed \(b>0\) and every \(j\) sufficiently large, $\rho_{a_j}\geq\rho_b.$ 
Hence
\[
\liminf_{j\to\infty}
\mathcal T_{\rho_{a_j}}(\nu,\mu)
\geq
\liminf_{j\to\infty}
\int\rho_b(\|x-y\|)\,d\pi_j
\geq
\int\rho_b(\|x-y\|)\,d\pi.
\]
Letting \(b\uparrow\infty\) and using monotone convergence yields
\[
\liminf_{j\to\infty}
\mathcal T_{\rho_{a_j}}(\nu,\mu)
\geq
\frac12\int\|x-y\|^2\,d\pi
\geq
\frac12\mathcal W_2^2(\nu,\mu).
\]
This proves \eqref{eq:monotone-cost-limit}.

Therefore, by monotone convergence,
\[
\sup_{a>0}\mathcal J_a(\nu_*)
=
\int
\sup_{a>0}\mathcal T_{\rho_a}(\nu_*,\mu)\,d{\bf P}(\mu)
=
\mathcal F_2(\nu_*).
\]
Thus \eqref{eq:large-c-liminf} becomes
\begin{equation}
\label{eq:F2-liminf-final}
\mathcal F_2(\nu_*)
\leq
\liminf_{k\to\infty}
\mathcal J_{c_{n_k}}(\nu_{n_k}).
\end{equation}
Let \(\eta\in\mathcal P_2(\mathbb R^d)\). By optimality of \(\nu_{n_k}\)
and the inequality \(\rho_c(t)\leq t^2/2\),
\[
\mathcal J_{c_{n_k}}(\nu_{n_k})
\leq
\mathcal J_{c_{n_k}}(\eta)
\leq
\mathcal F_2(\eta).
\]
Combining this with \eqref{eq:F2-liminf-final}, we obtain
\[
\mathcal F_2(\nu_*)
\leq
\mathcal F_2(\eta)
\qquad
\text{for every }\eta\in\mathcal P_2(\mathbb R^d).
\]
In particular, \(\mathcal F_2(\nu_*)<\infty\), so
\(\nu_*\in\mathcal P_2(\mathbb R^d)\), and $\nu_*\in m_\infty({\bf P}).$ 
If \(m_\infty({\bf P})=\{\nu_\infty\}\), every subsequence of
\(\{\nu_n\}\) has a further subsequence converging weakly to
\(\nu_\infty\), as required.
\end{proof}
\begin{proof}[Proof of \cref{theorem:charcterization}]
The fact that \textup{(i)}--\textup{(iii)} are sufficient follows from
\eqref{Subgradient-1-Gamma}. Indeed, define
\[
g(x):=\int f_{\nu,\mu}(x)\,d{\bf P}(\mu).
\]
Then, for every \(\gamma\in\mathcal P_1(\mathbb R^d)\),
\begin{align*}
\Gamma_{c,{\bf P}}(\gamma)
&\geq
\Gamma_{c,{\bf P}}(\nu)
+
\int g(x)\,d(\gamma-\nu)(x)\\
&=
\Gamma_{c,{\bf P}}(\nu)
+
\int g\,d\gamma
\geq
\Gamma_{c,{\bf P}}(\nu),
\end{align*}
where we used \textup{(ii)} and \textup{(iii)}.
We now prove necessity. Let \(\nu\) be a Huber-Wasserstein barycenter of
\({\bf P}\). Then, for every
\(\gamma\in\mathcal P_1(\mathbb R^d)\),
\begin{equation}
\label{eq:directional-nonnegative}
\int
\sup_{f\in{\rm Sol}^*(\nu,\mu)}
\int f(x)\,d(\gamma-\nu)(x)\,d{\bf P}(\mu)
\geq0,
\end{equation}
since otherwise \eqref{eq:Gateaux-Gamma} would contradict the minimality of
\(\nu\). Let \(\mathcal S(\nu)\) be the set of jointly measurable functions
$(\mu,x)\mapsto f_{\nu,\mu}(x)$
such that, for \({\bf P}\)-a.e.\ \(\mu\),
\(f_{\nu,\mu}\in{\rm Sol}^*(\nu,\mu)\), and impose the normalization
\begin{equation}
\label{eq:normalization-potentials}
\int f_{\nu,\mu}\,d\nu=0.
\end{equation}
We show the supremum in \eqref{eq:directional-nonnegative} can be moved outside the integrals. Let
\[
\mathscr X_{c,\nu}
:=
\left\{
h\in \mathcal C(\mathbb R^d):
\operatorname{Lip}(h)\le c,
\quad
\int h\,d\nu=0
\right\},
\]
endowed with the topology of uniform convergence on compact sets. We claim that
\(\mathscr X_{c,\nu}\) is a compact Polish space. Indeed, for every \(h\in\mathscr X_{c,\nu}\),
\[
|h(0)|
=
\left|
\int \bigl(h(0)-h(x)\bigr)\,d\nu(x)
\right|
\le
c\int \|x\|\,d\nu(x),
\]
and therefore
\begin{equation}
\label{eq:bound_on_h}
|h(x)|
\le
c\left(
\int \|y\|\,d\nu(y)+\|x\|
\right),
\qquad x\in\mathbb R^d.
\end{equation}
Hence \(\mathscr X_{c,\nu}\) is equicontinuous and uniformly bounded on
every compact subset of \(\mathbb R^d\). By the Arzelà--Ascoli theorem,
it is relatively compact for the topology of uniform convergence on
compact sets.

Moreover, \(\mathscr X_{c,\nu}\) is closed in this topology. Indeed, if
\(h_n\in\mathscr X_{c,\nu}\) and \(h_n\to h\) uniformly on compact sets,
then \(h\) is \(c\)-Lipschitz. In addition, \eqref{eq:bound_on_h} holds
uniformly in \(n\), and since \(\nu\in\mathcal P_1(\mathbb R^d)\),
dominated convergence yields
\[
\int h\,d\nu
=
\lim_{n\to\infty}\int h_n\,d\nu
=
0.
\]
Thus \(h\in\mathscr X_{c,\nu}\), so \(\mathscr X_{c,\nu}\) is compact. Finally, the topology of uniform convergence on compact sets on
\(\mathcal C(\mathbb R^d)\) is metrizable. Therefore
\(\mathscr X_{c,\nu}\), being a compact metrizable space, is complete
and separable, and hence Polish.

Let $S$ be the space \((\mathcal P_1(\mathbb R^d),\,\mathcal{W}_1) \) equipped with its Borel \(\sigma\)-algebra, and consider the
constant correspondence \(
\Phi_0:S\to \mathcal{B}(\mathscr X_{c,\nu}),
\Phi_0(\mu):=\mathscr X_{c,\nu}\), which is measurable and has nonempty compact values. Define
\[
D:S\times\mathscr X_{c,\nu}\longrightarrow\mathbb R,
\qquad
D(\mu,h)
:=
\int h\,d\nu+\int h^{\rho_c}\,d\mu.
\]
We claim that \(D\) is a Carathéodory function, see \cite[p.~156]{AliprantisBorder2006} for the definition. For each fixed
\(h\in\mathscr X_{c,\nu}\), the function \(h^{\rho_c}\) is \(c\)-Lipschitz.
Consequently, $\mu\mapsto D(\mu,h)$ is continuous with respect to \(\mathcal W_1\).

On the other hand, suppose that \(h_n\to h\) locally uniformly in
\(\mathscr X_{c,\nu}\). By the stability of \(\rho_c\)-conjugation, $h_n^{\rho_c}\to h^{\rho_c}$ uniformly on compact sets. Moreover,
\[
\sup_n|h_n(x)|\leq c(\mathbb{E}_\nu [\|X\|]+\|x\|),
\qquad
\sup_n|h_n^{\rho_c}(x)|\leq c\mathbb{E}_\nu[\|X\|]+c^2+c\|x\|.
\]
Since \(\mu,\nu\in\mathcal P_1(\mathbb R^d)\), these
bounds and uniform convergence on compact sets imply $D(\mu,h_n)\to D(\mu,h)$. Thus, \(h\mapsto D(\mu,h)\) is continuous for every fixed \(\mu\), showing
that \(D\) is Carathéodory.

Applying the measurable maximum theorem
\cite[Theorem~18.19]{AliprantisBorder2006} to \(\Phi_0\) and \(D\), we obtain
that the correspondence
\[
\Phi(\mu)
:=
\operatorname{argmax}_{h\in\mathscr X_{c,\nu}}D(\mu,h)
\]
is measurable, has nonempty compact values, and admits a measurable selection.
By Kantorovich duality,
\[
\Phi(\mu)
=
\left\{
h\in{\rm Sol}^*(\nu,\mu):\int h d\nu =0
\right\},
\]
as every dual potential can be normalized to have zero mean. 

Now fix \(\gamma\in\mathcal P_1(\mathbb R^d)\) and define
\[
L_\gamma:S\times\mathscr X_{c,\nu}\longrightarrow\mathbb R,
\qquad
L_\gamma(\mu,h)
:=
\int h\,d(\gamma-\nu).
\]
The function \(L_\gamma\) is Carathéodory. It is independent of \(\mu\), and
its continuity in \(h\) follows from uniform convergence on compact sets, the
bound \eqref{eq:bound_on_h}, and the assumption that \(\gamma,\nu\in\mathcal P_1(\mathbb R^d)\).

Applying the measurable maximum theorem a second time, now to the
correspondence \(\Phi\) and the function \(L_\gamma\), shows that
\[
\mu\longmapsto
\operatorname{argmax}_{h\in\Phi(\mu)}
\int h\,d(\gamma-\nu)
\]
is measurable, has nonempty compact values, and admits a measurable selector $
\mu\mapsto f_{\nu,\mu}^{\gamma}$.
Therefore,
\[
\int f_{\nu,\mu}^{\gamma}\,d(\gamma-\nu)
=
\max_{h\in\Phi(\mu)}
\int h\,d(\gamma-\nu)
=
\sup_{f\in{\rm Sol}^*(\nu,\mu)}
\int f\,d(\gamma-\nu),
\]
where the last equality holds because integration against \(\gamma-\nu\) is
invariant under the addition of constants. 

Finally, since the evaluation map $e:\mathscr X_{c,\nu}\times \mathbb{R}^d\to \mathbb{R},\, e(h,x)=h(x)$ is continuous in the product topology of the topology of uniform convergence on compact sets with the Euclidean topology on $\mathbb{R}^d$, and as $\mu\mapsto f_{\nu,\mu}^\gamma$ is measurable, the function $(\mu, x)\mapsto f_{\nu,\mu}^\gamma (x)$ is jointly measurable, and so $f_{\nu,(\cdot)}^\gamma \in \mathcal{S}(\nu)$. Recalling the definition of $\mathcal{S}(\nu)$, we obtain
\[
\sup_{f_{\nu,(\cdot)}\in\mathcal S(\nu)}
\int\!\!\int
f_{\nu,\mu}(x)\,d(\gamma-\nu)(x)\,d{\bf P}(\mu)
\geq0.
\]
Define
\[
\mathcal G
:=
\left\{
x\longmapsto
\int f_{\nu,\mu}(x)\,d{\bf P}(\mu):
f_{\nu,(\cdot)}\in\mathcal S(\nu)
\right\}.
\]
By Fubini's theorem,
\begin{equation}
\label{eq:direction-G}
\sup_{g\in\mathcal G}\int g\,d(\gamma-\nu)\geq0
\qquad
\text{for every }\gamma\in\mathcal P_1(\mathbb R^d).
\end{equation}
By \eqref{eq:normalization-potentials}, $\int g\,d\nu=0$ for every $g\in\mathcal G.$  
Hence, \eqref{eq:direction-G} becomes
\begin{equation}
\label{eq:direction-G-normalized}
\sup_{g\in\mathcal G}\int g\,d\gamma\geq0
\qquad
\text{for every }\gamma\in\mathcal P_1(\mathbb R^d).
\end{equation}
The set \(\mathcal G\) is convex, since the set of optimal potentials is
convex. Furthermore, every \(g\in\mathcal G\) is \(c\)-Lipschitz. We now show that there exist \(g\in\mathcal G\) and
\(a\in\mathbb R\) such that
\begin{enumerate}
\item[(a)] \(g(x)=a\) for \(\nu\)-a.e.\ \(x\);
\item[(b)] \(g(x)\geq a\) for every \(x\in\mathbb R^d\).
\end{enumerate}
We first assume that $\operatorname{supp}(\nu)\neq\mathbb R^d.$ 
Choose
\(\beta\in\mathcal P_1(\mathbb R^d)\) such that
\begin{equation}
\label{eq:properties-beta}
\beta(\operatorname{supp}\nu)=0,
\qquad
\operatorname{supp}(\nu+\beta)=\mathbb R^d.
\end{equation}
For example, one may take a countable convex combination of probability
measures supported on balls contained in
\(\mathbb R^d\setminus\operatorname{supp}\nu\), chosen so that their union is
dense in this open set. Define \(V:=L^1(\nu+\beta)/\langle1\rangle.\) 
Its dual is identified with
\[
V^*
=
\left\{
h\in L^\infty(\nu+\beta):
\int h\,d(\nu+\beta)=0
\right\}.
\]
Define
\[
M
:=
\left\{
[f]\in V:
\text{there exists }a\in\mathbb R
\text{ such that }
f=a\ \nu\text{-a.e. and }f\geq a\ \beta\text{-a.e.}
\right\}.
\]

\medskip

\noindent
\textit{The set \(M\) is convex and closed.}
Convexity is immediate. To prove closedness, choose in every equivalence class
the representative satisfying $\int f\,d\nu=0.$ 
For such representatives, the constant in the definition of \(M\) is zero.
Thus,
\[
M
=
\left\{
[f]\in V:
f=0\ \nu\text{-a.e. and }f\geq0\ \beta\text{-a.e.}
\right\}.
\]
If \(f_n\to f\) in \(L^1(\nu+\beta)\), then a subsequence converges
\((\nu+\beta)\)-a.e. Therefore,
\[
f=0\quad\nu\text{-a.e.},
\qquad
f\geq0\quad\beta\text{-a.e.},
\]
and hence \(f\in M\).

\medskip

\noindent
\textit{The set \(\mathcal G\) is relatively compact in \(V\).}
Let \(\{g_n\}_n\subset\mathcal G\), with
\[
g_n(z)
=
\int f_{\nu,\mu}^{(n)}(z)\,d{\bf P}(\mu),
\]
where $\int f_{\nu,\mu}^{(n)}\,d\nu=0.$ 
For every \(z\in\mathbb R^d\),
\begin{align*}
0
=
\int f_{\nu,\mu}^{(n)}(x)\,d\nu(x)
\leq
f_{\nu,\mu}^{(n)}(z)
+
c\int\|x-z\|\,d\nu(x)
\leq
f_{\nu,\mu}^{(n)}(z)
+c\|z\|
+c\int\|x\|\,d\nu(x),
\end{align*}
and similarly,
\[
0
\geq
f_{\nu,\mu}^{(n)}(z)
-c\|z\|
-c\int\|x\|\,d\nu(x).
\]
Setting $C:=c\int\|x\|\,d\nu(x),$ 
we obtain
\begin{equation}
\label{eq:bounded-f}
|f_{\nu,\mu}^{(n)}(z)|
\leq
C+c\|z\|,
\qquad
|g_n(z)|
\leq
C+c\|z\|.
\end{equation}
The sequence \(\{g_n\}_n\) is equi-Lipschitz and locally uniformly bounded.
By the Arzelà--Ascoli theorem, every subsequence admits a further subsequence
converging uniformly on compact sets to some \(c\)-Lipschitz function \(g\).
The bound \eqref{eq:bounded-f} passes to the limit:
$|g(z)|\leq C+c\|z\|.$
Since \(\nu+\beta\in\mathcal P_1(\mathbb R^d)\), for every
\(\varepsilon>0\), there exists a compact set \(K\subset\mathbb R^d\) such
that
\[
\int_{\mathbb R^d\setminus K}
(C+c\|x\|)\,d(\nu+\beta)(x)
<
\frac{\varepsilon}{4}.
\]
The uniform convergence on \(K\) and the linear growth bound give $\|g_n-g\|_{L^1(\nu+\beta)}\longrightarrow0.$ 
Thus \(\mathcal G\) is relatively compact in \(V\).

\medskip

\noindent
\textit{The set \(\mathcal G\) is closed in \(V\).}
Assume that
\begin{equation}
\label{eq:assumethatmean-to-mean}
g_n
:=
\int f_{\nu,\mu}^{(n)}\,d{\bf P}(\mu)
\longrightarrow F
\quad\text{in }V.
\end{equation}
The estimate \eqref{eq:bounded-f} implies that
\(\{f_{\nu,\mu}^{(n)}\}_n\) is uniformly integrable in $L^1\bigl({\bf P}\otimes(\nu+\beta)\bigr).$ 
By  \cite[Theorem~4.7.18]{bogachev2007measure}, after passing to a subsequence,
\[
f_{\nu,\mu}^{(n)}
\rightharpoonup
f_{\nu,\mu}
\quad\text{weakly in }
L^1\bigl({\bf P}\otimes(\nu+\beta)\bigr).
\]
Testing this convergence against functions depending only on \(x\) and using
\eqref{eq:assumethatmean-to-mean}, we obtain
\begin{equation}
\label{eq:mean-of-limit}
F(x)
=
\int f_{\nu,\mu}(x)\,d{\bf P}(\mu)
\quad
\text{for }(\nu+\beta)\text{-a.e. }x.
\end{equation}
By the weak Banach--Saks property of \(L^1\)  (see e.g.~\cite{foghem2023banachsakstheoreml1revisited}), there exists a subsequence
\(\{n_k\}_k\) such that
\[
\bar f_{\nu,\mu}^{(k)}
:=
\frac1k\sum_{j=1}^k f_{\nu,\mu}^{(n_j)}
\longrightarrow
f_{\nu,\mu}
\]
strongly in $L^1\bigl({\bf P}\otimes(\nu+\beta)\bigr).$ 
After passing to a further subsequence,
$\bar f_{\nu,\mu}^{(k)}(x)
\longrightarrow
f_{\nu,\mu}(x)$
for \({\bf P}\otimes(\nu+\beta)\)-a.e.\ \((\mu,x)\).

For \({\bf P}\)-a.e.\ \(\mu\), the functions
\(\bar f_{\nu,\mu}^{(k)}\) are \(c\)-Lipschitz. By
\eqref{eq:properties-beta}, the measure \(\nu+\beta\) has full support.
Therefore, the almost-everywhere convergence and equi-Lipschitz continuity
imply, after choosing the continuous representative of \(f_{\nu,\mu}\),
\[
\bar f_{\nu,\mu}^{(k)}
\longrightarrow
f_{\nu,\mu}
\quad\text{uniformly on compact sets}.
\]
Moreover, since the set of optimal potentials is convex, $\bar f_{\nu,\mu}^{(k)}
\in{\rm Sol}^*(\nu,\mu).$ 
Hence, by \cref{theorem:stsbilitypot},
\[
f_{\nu,\mu}\in{\rm Sol}^*(\nu,\mu)
\qquad
\text{for }{\bf P}\text{-a.e. }\mu.
\]
Together with \eqref{eq:mean-of-limit}, this proves that \(F\in\mathcal G\).
Thus \(\mathcal G\) is closed, and therefore compact in \(V\).

\medskip

\noindent
\textit{Conclusion of the proof: Hahn--Banach separation.}
Assume by contradiction that $\mathcal G\cap M=\varnothing.$ 
Since \(\mathcal G\) is compact and convex and \(M\) is closed and convex, the
geometric Hahn--Banach theorem
\citep[Theorem~1.7]{brezis2011functional} gives
\(\varepsilon>0\), \(\eta\in\mathbb R\), and
\[
h=(h_1,h_2)\in
L^\infty(\nu)\times L^\infty(\beta),
\qquad
\int h_1\,d\nu+\int h_2\,d\beta=0,
\]
such that
\begin{equation}
\label{eq:separation-G}
\int gh_1\,d\nu+\int gh_2\,d\beta
\leq
\eta-\varepsilon
\qquad
\text{for every }g\in\mathcal G,
\end{equation}
and
\begin{equation}
\label{eq:separation-M}
\int fh_1\,d\nu+\int fh_2\,d\beta
\geq
\eta
\qquad
\text{for every }f\in M.
\end{equation}
Since \(0\in M\), one has \(\eta\leq0\).

For every nonnegative \(u\in L^1(\beta)\), the function equal to zero
\(\nu\)-a.e.~and to \(u\) \(\beta\)-a.e.~belongs to \(M\). Therefore,
\[
t\int uh_2\,d\beta\geq\eta
\qquad
\text{for every }t>0.
\]
Rearranging and letting \(t\to\infty\), we obtain
\[
\int uh_2\,d\beta\geq0
\qquad
\text{for every }u\geq0,
\]
and consequently $h_2\geq0$, $\beta$-a.e. Define the finite signed measure \(\sigma\) by $d\sigma=h_1\,d\nu+h_2\,d\beta.$ 
Then $\sigma(\mathbb R^d)=0.$ 
For every sufficiently small \(t>0\), namely $0<t\leq
\frac{1}{\|h_1^-\|_{L^\infty(\nu)}},$ 
the measure $\gamma_t:=\nu+t\sigma$ 
is a probability measure in \(\mathcal P_1(\mathbb R^d)\), since
\[
d\gamma_t=(1+th_1)\,d\nu+th_2\,d\beta\geq0.
\]
By \eqref{eq:separation-G} and \(\eta\leq0\),
$\sup_{g\in\mathcal G}\int g\,d\sigma
\leq-\varepsilon.$ 
Therefore,
\[
\sup_{g\in\mathcal G}
\int g\,d(\gamma_t-\nu)
=
t\sup_{g\in\mathcal G}\int g\,d\sigma
\leq-t\varepsilon<0.
\]
This contradicts \eqref{eq:direction-G}, applied with \(\gamma=\gamma_t\).
Thus, $\mathcal G\cap M\neq\varnothing.$ Choose \(g\in\mathcal G\cap M\). Then there exists
\(f_{\nu,(\cdot)}\in\mathcal S(\nu)\) such that
\[
g(x)=\int f_{\nu,\mu}(x)\,d{\bf P}(\mu).
\]
Moreover, for some \(a\in\mathbb R\),
\[
g=a\quad\nu\text{-a.e.},
\qquad
g\geq a\quad\beta\text{-a.e.}
\]
Since \(g\) is continuous, \(\operatorname{supp}(\nu+\beta)=\mathbb R^d\),
and \(g=a\) on \(\operatorname{supp}\nu\) while
\(g\geq a\) on \(\operatorname{supp}\beta\), it follows that
\[
g(x)\geq a
\qquad
\text{for every }x\in\mathbb R^d.
\]
Finally, define $\widetilde f_{\nu,\mu}:=f_{\nu,\mu}-a.$ 
Since potentials are invariant under addition of constants,
\(\widetilde f_{\nu,\mu}\) is still a potential for \((\nu,\mu)\), and
\[
\int\widetilde f_{\nu,\mu}(x)\,d{\bf P}(\mu)=0
\quad
\text{for }\nu\text{-a.e. }x,
\]
while
\[
\int\widetilde f_{\nu,\mu}(x)\,d{\bf P}(\mu)\geq0
\quad
\text{for every }x\in\mathbb R^d.
\]
If \(\operatorname{supp}\nu=\mathbb R^d\), the same argument is carried out 
in $V=L^1(\nu)/\langle1\rangle$ 
with \(M=\{0\}\). Indeed, a continuous function which is constant
\(\nu\)-a.e. is constant everywhere. The separation argument then shows that
\(0\in\mathcal G\), which directly gives \textup{(ii)} and \textup{(iii)}. This concludes the proof.
\end{proof}
\begin{proof}[Proof of \cref{lemma:Gateaux}]
Fix \(f\in{\rm Sol}^*(\nu,\mu)\). Then
\[
\mathcal T_{\rho_c}(\mu,\nu)
=
\int f\,d\nu+\int f^{\rho_c}\,d\mu
\]
and
\[
\mathcal T_{\rho_c}(\mu,\gamma)
\geq
\int f\,d\gamma+\int f^{\rho_c}\,d\mu,
\]
so that \eqref{Subgradient-1} follows.

Set \( \nu_t:=\nu+t(\gamma-\nu)\). For \(f_0\in{\rm Sol}^*(\nu,\mu)\) and
\(f_t\in{\rm Sol}^*(\nu_t,\mu)\), we have
\[
\int f_0\,d(\gamma-\nu)
\leq
\frac{
\mathcal T_{\rho_c}(\mu,\nu_t)
-
\mathcal T_{\rho_c}(\mu,\nu)
}{t}
\leq
\int f_t\,d(\gamma-\nu).
\]
Therefore,
\begin{equation}
\label{eq:Gateaux-ineq}
\sup_{f\in{\rm Sol}^*(\nu,\mu)}
\int f\,d(\gamma-\nu)
\leq
\liminf_{t\downarrow0}
\frac{
\mathcal T_{\rho_c}(\mu,\nu_t)
-
\mathcal T_{\rho_c}(\mu,\nu)
}{t}.
\end{equation}
We show that equality holds. Let \(t_n\downarrow0\), and choose \( f_{t_n}\in{\rm Sol}^*(\nu_{t_n},\mu)\).
Since potentials are defined up to additive constants, we
normalize them by \(f_{t_n}(0)=0\). By
\cref{lemma:Lipschitz-extension}, the functions \(f_{t_n}\) are
\(c\)-Lipschitz and satisfy
$|f_{t_n}(x)|\leq c\|x\|.$
Hence, after passing to a subsequence,
\[
f_{t_n}\longrightarrow f_*
\quad\text{uniformly on compact sets}
\]
for some \(c\)-Lipschitz function \(f_*\) with \(f_*(0)=0\). Since \(\nu_{t_n}\to\nu\) in \(\mathcal W_1\),
\cref{theorem:stsbilitypot} implies that 
$f_*\in{\rm Sol}^*(\nu,\mu).$
Moreover, the local uniform convergence and the bound
\[
|f_{t_n}(x)|+|f_*(x)|\leq2c\|x\|
\]
give
\[
\int f_{t_n}\,d(\gamma-\nu)
\longrightarrow
\int f_*\,d(\gamma-\nu).
\]
Consequently,
\begin{align*}
\limsup_{n\to\infty}
\frac{
\mathcal T_{\rho_c}(\mu,\nu_{t_n})
-
\mathcal T_{\rho_c}(\mu,\nu)
}{t_n}
&\leq
\lim_{n\to\infty}\int f_{t_n}\,d(\gamma-\nu)\\
&=
\int f_*\,d(\gamma-\nu)\\
&\leq
\sup_{f\in{\rm Sol}^*(\nu,\mu)}
\int f\,d(\gamma-\nu).
\end{align*}
Since \(t_n\downarrow0\) was arbitrary, this inequality together with
\eqref{eq:Gateaux-ineq} proves the result.
\end{proof}
\begin{proof}[Proof of \cref{lemma:Gateaux-Ganmma}]
The inequality \eqref{Subgradient-1-Gamma} follows directly by integrating
\eqref{Subgradient-1} with respect to \({\bf P}\).
To prove \eqref{eq:Gateaux-Gamma}, set \(\nu_t:=\nu+t(\gamma-\nu).\) 
For normalized potentials
\[
f_0\in{\rm Sol}^*(\nu,\mu),
\qquad
f_t\in{\rm Sol}^*(\nu_t,\mu),
\]
the proof of \cref{lemma:Gateaux} gives
\[
\int f_0\,d(\gamma-\nu)
\leq
\frac{
\mathcal T_{\rho_c}(\mu,\nu_t)
-
\mathcal T_{\rho_c}(\mu,\nu)
}{t}
\leq
\int f_t\,d(\gamma-\nu).
\]
By \cref{lemma:Lipschitz-extension}, we may normalize the potentials at the
origin, in which case
\[
|f_0(x)|+|f_t(x)|\leq2c\|x\|.
\]
Thus,
\[
\left|
\int f_0\,d(\gamma-\nu)
\right|
+
\left|
\int f_t\,d(\gamma-\nu)
\right|
\leq
2c
\left(
\int\|x\|\,d\gamma(x)+\int\|x\|\,d\nu(x)
\right).
\]
The right-hand side is independent of \(\mu\) and \(t\).
Hence, \eqref{eq:Gateaux-Gamma} follows from
\cref{lemma:Gateaux} and the dominated convergence theorem.
\end{proof}
\begin{proof}[Proof of \Cref{Theorem:BP}]
    Since all measures involved in this proof are empirical, they have finite
first moments. Moreover, the term $-c\int\|x\|\,d\mu(x)$ 
does not depend on the barycenter variable. Hence, without loss of generality,
throughout this subsection we use the equivalent functional
\[
\Gamma_{c,\boldsymbol{\mu}}(\nu)
:=
\frac1n\sum_{i=1}^n\mathcal T_{\rho_c}(\mu_i,\nu),
\qquad
\boldsymbol{\mu}
=
\frac1n\sum_{i=1}^n\delta_{\mu_i}.
\]
This modification does not change the set \(m_c(\boldsymbol{\mu})\) of
minimizers.

\textit{Proof of}
\({\rm BP}(m_c(\boldsymbol{\mu}))
\geq n^{-1}\lceil n/2\rceil\)].
We recall the following straightforward inequalities:
\begin{align}
\rho_c(t)
&\geq
c|t|-\frac{c^2}{2},
\qquad t\in\mathbb R,
\label{eq:Hubber-ineq-1}\\
\rho_c(t)
&\leq
c|t|,
\qquad t\in\mathbb R.
\label{eq:Hubber-ineq-2}
\end{align}
From the monotonicity of
\([0,\infty)\ni t\mapsto\rho_c(t)\),
\eqref{eq:Hubber-ineq-1}, and \eqref{eq:Hubber-ineq-2}, we obtain
\begin{align}
\label{eq:bound-Hubber-x-y}
\rho_c(\|x-y\|)
&\geq
\rho_c\bigl(|\|x\|-\|y\||\bigr)
\nonumber\\
&\geq
-\frac{c^2}{2}
+
c|\|x\|-\|y\||
\nonumber\\
&\geq
-\frac{c^2}{2}
+
c\|x\|-c\|y\|\geq
-\frac{c^2}{2}
+
\rho_c(\|x\|)-c\|y\|.
\end{align}

Fix \(s<n/2\), let
\(\boldsymbol{\nu}\in\mathcal Q_s(\boldsymbol{\mu})\), and choose
\(\nu_*\in m_c(\boldsymbol{\nu})\). Since
\(m_c(\boldsymbol{\mu})\) is bounded by
\cref{theorem:Existence-General}, it is enough to show that $\int\|x\|\,d\nu_*(x)\leq C$ 
for a constant independent of
\(\boldsymbol{\nu}\in\mathcal Q_s(\boldsymbol{\mu})\).

Without loss of generality, after relabeling the atoms, we may write
\[
\boldsymbol{\nu}
=
\frac1n\sum_{i=1}^s\delta_{\nu_i}
+
\frac1n\sum_{i=s+1}^n\delta_{\mu_i},
\]
where \(\nu_1,\ldots,\nu_s\) are arbitrary. Define
\[
\boldsymbol{\gamma}
=
\frac1n\sum_{i=1}^s\delta_{\nu_i}
+
\frac{n-s}{n}\delta_{\delta_0}.
\]
By \cref{theorem:stability-cost}, for every
\(\nu\in\mathcal P_1(\mathbb R^d)\),
\begin{equation}
\label{eq:Bound-barycenter-by-trivial}
\left|
\Gamma_{c,\boldsymbol{\gamma}}(\nu)
-
\Gamma_{c,\boldsymbol{\nu}}(\nu)
\right|
\leq
\frac cn\sum_{i=s+1}^n
\mathcal W_1(\delta_0,\mu_i)
=:C_1<\infty.
\end{equation}

Since
\begin{align*}
\Gamma_{c,\boldsymbol{\gamma}}(\nu)
&=
\frac1n\sum_{i=1}^s
\mathcal T_{\rho_c}(\nu_i,\nu)
+
\frac{n-s}{n}
\mathcal T_{\rho_c}(\delta_0,\nu)=
\frac1n\sum_{i=1}^s
\mathcal T_{\rho_c}(\nu_i,\nu)
+
\frac{n-s}{n}
\int\rho_c(\|y\|)\,d\nu(y),
\end{align*}
the relation \eqref{eq:bound-Hubber-x-y} yields
\begin{align*}
\Gamma_{c,\boldsymbol{\gamma}}(\nu)
&\geq
\frac1n\sum_{i=1}^s
\int\rho_c(\|x\|)\,d\nu_i(x)
-\frac{sc^2}{2n}
-\frac{sc}{n}\int\|y\|\,d\nu(y)+
\frac{n-s}{n}
\int\rho_c(\|y\|)\,d\nu(y).
\end{align*}
Using \eqref{eq:Hubber-ineq-1}, we obtain
\begin{align}
\label{eq:Bound-in-Gamma-gamma}
\Gamma_{c,\boldsymbol{\gamma}}(\nu)
&\geq
\frac1n\sum_{i=1}^s
\int\rho_c(\|x\|)\,d\nu_i(x)
-\frac{c^2}{2}
+
\frac{(n-2s)c}{n}
\int\|y\|\,d\nu(y)
\nonumber\\
&=
\Gamma_{c,\boldsymbol{\gamma}}(\delta_0)
-\frac{c^2}{2}
+
\frac{(n-2s)c}{n}
\int\|y\|\,d\nu(y).
\end{align}

Therefore, by \eqref{eq:Bound-barycenter-by-trivial} and
\eqref{eq:Bound-in-Gamma-gamma},
\begin{align*}
\Gamma_{c,\boldsymbol{\nu}}(\nu_*)
&\geq
\Gamma_{c,\boldsymbol{\gamma}}(\nu_*)-C_1\\
&\geq
\Gamma_{c,\boldsymbol{\gamma}}(\delta_0)
-\frac{c^2}{2}
-C_1
+
\frac{(n-2s)c}{n}
\int\|y\|\,d\nu_*(y)\\
&\geq
\Gamma_{c,\boldsymbol{\nu}}(\delta_0)
-\frac{c^2}{2}
-2C_1
+
\frac{(n-2s)c}{n}
\int\|y\|\,d\nu_*(y)\\
&\geq
\Gamma_{c,\boldsymbol{\nu}}(\nu_*)
-\frac{c^2}{2}
-2C_1
+
\frac{(n-2s)c}{n}
\int\|y\|\,d\nu_*(y),
\end{align*}
where the last inequality follows from the optimality of \(\nu_*\).
Consequently,
\[
\frac{(n-2s)c}{n}
\int\|y\|\,d\nu_*(y)
\leq
\frac{c^2}{2}+2C_1.
\]
Since \(s<n/2\), this gives the required uniform first-moment bound. 
Moreover, for every
\(\alpha\in m_c(\boldsymbol{\nu})\) and
\(\beta\in m_c(\boldsymbol{\mu})\),
\[
\mathcal W_1(\alpha,\beta)
\leq
\mathcal W_1(\alpha,\delta_0)
+
\mathcal W_1(\delta_0,\beta)
=
\int\|x\|\,d\alpha(x)
+
\int\|x\|\,d\beta(x).
\]
Thus the uniform first-moment bound, together with the boundedness of
\(m_c(\boldsymbol{\mu})\), implies that
\[
\sup_{\boldsymbol{\nu}\in\mathcal Q_s(\boldsymbol{\mu})}
{\rm dist}
\bigl(m_c(\boldsymbol{\nu}),m_c(\boldsymbol{\mu})\bigr)
<\infty.
\] 
Hence, ${\rm BP}(m_c(\boldsymbol{\mu}))
\geq
\frac1n\left\lceil\frac n2\right\rceil.$ 

\textit{Proof of}
\({\rm BP}(m_c(\boldsymbol{\mu}))
\leq n^{-1}(\lfloor n/2\rfloor+1)\)]
Set $s:=\left\lfloor\frac n2\right\rfloor+1, $ 
so that \(2s>n\). For \(m>0\), define
\[
\boldsymbol{\nu}^m
=
\frac1n\sum_{i=1}^s\delta_{\delta_{me_1}}
+
\frac1n\sum_{i=s+1}^n\delta_{\mu_i}.
\] 
Then
\(\boldsymbol{\nu}^m\in\mathcal Q_s(\boldsymbol{\mu})\). Define  $\boldsymbol{\gamma}^m
=
\frac{s}{n}\delta_{\delta_{me_1}}
+
\frac{n-s}{n}\delta_{\delta_0}.$  
By \cref{theorem:stability-cost}, for every
\(\nu\in\mathcal P_1(\mathbb R^d)\),
\begin{equation}
\label{eq:BP-stability-comparison}
\left|
\Gamma_{c,\boldsymbol{\nu}^m}(\nu)
-
\Gamma_{c,\boldsymbol{\gamma}^m}(\nu)
\right|
\leq
\frac cn\sum_{i=s+1}^n
\mathcal W_1(\mu_i,\delta_0)
=:C_1<\infty,
\end{equation} 
where \(C_1\) does not depend on \(m\). Fix $\nu_m\in m_c(\boldsymbol{\nu}^m).$ 
Then
\begin{align}
\label{eq:BP-proof-1}
\Gamma_{c,\boldsymbol{\gamma}^m}(\nu_m)
&\leq
\Gamma_{c,\boldsymbol{\nu}^m}(\nu_m)+C_1
\nonumber\\
&\leq
\Gamma_{c,\boldsymbol{\nu}^m}(\delta_{me_1})+C_1
\leq
\Gamma_{c,\boldsymbol{\gamma}^m}(\delta_{me_1})+2C_1.
\end{align}
Moreover,
\begin{equation}
\label{eq:BP-proof-2}
\Gamma_{c,\boldsymbol{\gamma}^m}(\delta_{me_1})
=
\frac{n-s}{n}\rho_c(m)
\leq
\frac{c m(n-s)}{n}.
\end{equation}
On the other hand,
\begin{align*}
\Gamma_{c,\boldsymbol{\gamma}^m}(\nu_m)
&=
\frac{s}{n}
\int\rho_c(\|me_1-x\|)\,d\nu_m(x)
+
\frac{n-s}{n}
\int\rho_c(\|x\|)\,d\nu_m(x).
\end{align*}
Using \eqref{eq:Hubber-ineq-1}, we obtain
\begin{align*}
\Gamma_{c,\boldsymbol{\gamma}^m}(\nu_m)
&\geq
\frac{cs}{n}
\int\|me_1-x\|\,d\nu_m(x)
+
\frac{c(n-s)}{n}
\int\|x\|\,d\nu_m(x)
-\frac{c^2}{2}\\
&\geq
\frac{csm}{n}
+
\frac{c(n-2s)}{n}
\int\|x\|\,d\nu_m(x)
-\frac{c^2}{2},
\end{align*} 
where we used $\|me_1-x\|\geq m-\|x\|.$ Combining \eqref{eq:BP-proof-1} and \eqref{eq:BP-proof-2}, we obtain
\[
\frac{csm}{n}
+
\frac{c(n-2s)}{n}
\int\|x\|\,d\nu_m(x)
-\frac{c^2}{2}
\leq
\frac{cm(n-s)}{n}+2C_1.
\]
Since \(2s>n\), this is equivalent to
\[
\frac{c(2s-n)}{n}
\int\|x\|\,d\nu_m(x)
\geq
\frac{cm(2s-n)}{n}
-\frac{c^2}{2}
-2C_1.
\] 
Consequently,
\[
\int\|x\|\,d\nu_m(x)
\geq
m
-
\frac{n}{c(2s-n)}
\left(
\frac{c^2}{2}+2C_1
\right),
\]
and hence
\begin{equation}
\label{eq:uniform-divergence-barycenters}
\inf_{\nu_m\in m_c(\boldsymbol{\nu}^m)}
\int\|x\|\,d\nu_m(x)
\longrightarrow\infty
\qquad\text{as }m\to\infty.
\end{equation}
Finally, for every
\(\alpha\in m_c(\boldsymbol{\nu}^m)\) and
\(\beta\in m_c(\boldsymbol{\mu})\), the  triangle inequality for
\(\mathcal W_1\) gives
\[
\mathcal W_1(\alpha,\beta)
\geq
\left|
\mathcal W_1(\alpha,\delta_0)
-
\mathcal W_1(\beta,\delta_0)
\right|
=
\left|
\int\|x\|\,d\alpha(x)
-
\int\|x\|\,d\beta(x)
\right|.
\]
Since \(m_c(\boldsymbol{\mu})\) is bounded in \(\mathcal W_1\),
\eqref{eq:uniform-divergence-barycenters} implies that
${\rm dist}
\bigl(
m_c(\boldsymbol{\nu}^m),
m_c(\boldsymbol{\mu})
\bigr)
\longrightarrow\infty.$
Therefore,
\[
{\rm BP}(m_c(\boldsymbol{\mu}))
\leq
\frac1n
\left(
\left\lfloor\frac n2\right\rfloor+1
\right).
\]
\end{proof}
\begin{proof}[Proof of \Cref{prop:huber-barycenter-1d}]
For $u\in(0,1)$ and $z\in\mathbb R$, set
\[
S_{c,u}(z)
:=
\int_{\mathcal P_1(\mathbb R)}
\psi_c\bigl(z-Q_\mu(u)\bigr)\,d{\bf P}(\mu).
\]
For every $q\in\mathbb R$, the map $z\mapsto \rho_c(|z-q|)$ 
is convex and continuously differentiable, with derivative
$\psi_c(z-q)$. Since $|\psi_c|\leq c$, differentiation under the
integral sign gives
$\Phi_{c,u}'(z)=S_{c,u}(z).$ 
In particular, $\Phi_{c,u}$ is convex and continuously differentiable,
while $S_{c,u}$ is continuous and nondecreasing. For every fixed $\mu\in\mathcal P_1(\mathbb R)$ and $u\in(0,1)$,
the quantity $Q_\mu(u)$ is finite. Hence
$\lim_{z\to \pm \infty}\psi_c\bigl(z-Q_\mu(u)\bigr) =\pm  c$. 
Since $|\psi_c|\leq c$, dominated convergence yields
\[
\lim_{z\to-\infty}S_{c,u}(z)=-c,
\qquad
\lim_{z\to+\infty}S_{c,u}(z)=c.
\]
Therefore, the zero set of $S_{c,u}$ is nonempty and bounded. It is
closed by continuity and is an interval because $S_{c,u}$ is
nondecreasing. Finally, since $\Phi_{c,u}$ is convex and
differentiable,
\[
\operatorname{argmin}_{z\in\mathbb R}\Phi_{c,u}(z)
=
\{z\in\mathbb R:S_{c,u}(z)=0\}.
\]
This proves that $M_c(u)$ is a nonempty compact interval. We next establish a monotonicity property of its endpoints. Define $a_c(u):=\min M_c(u).$ 
If $0<u\leq v<1$, then
\[
Q_\mu(u)\leq Q_\mu(v)
\qquad\text{for every }\mu\in\mathcal P_1(\mathbb R).
\]
Since $\psi_c$ is nondecreasing, it follows that
\[
S_{c,u}(z)\geq S_{c,v}(z)
\qquad\text{for every }z\in\mathbb R.
\]
Moreover, $a_c(u)=\min\{z\in\mathbb R:S_{c,u}(z)=0\}.$ 
Consequently, $a_c(u)\leq a_c(v),$ 
so $a_c$ is nondecreasing and, in particular, measurable. Now we can show the final claim.  Suppose first that
\[
S_{c,u}\bigl(Q_\nu(u)\bigr)=0
\qquad\text{for a.e. }u\in(0,1).
\]
Then $Q_\nu(u)\in M_c(u)$ for almost every $u$. Thus, for every
$\gamma\in\mathcal P_1(\mathbb R)$,
\[
\Phi_{c,u}\bigl(Q_\nu(u)\bigr)
\leq
\Phi_{c,u}\bigl(Q_\gamma(u)\bigr)
\qquad\text{for a.e. }u.
\]
Integrating and using
\eqref{eq:objective-quantile-decomposition}, we obtain $\Gamma_{c,{\bf P}}(\nu)
\leq
\Gamma_{c,{\bf P}}(\gamma).$ 
Hence $\nu$ is a Huber--Wasserstein barycenter.  Conversely, suppose that $\nu$ is a Huber--Wasserstein barycenter. Since $a_c$ is monotone, there exists a probability measure $\nu'$ with quantile function $a_c$. Furthermore, by definition, for all $\gamma \in \mathcal{P}_1(\mathbb{R})$,
$$
\Phi_{c,u}(a_c(u))\leq \Phi_{c,u}(Q_\gamma(u))\qquad\text{for a.e. }u.
$$
In particular, $\nu'\in \mathcal{P}_1(\mathbb{R}^d)$. Using  \eqref{eq:objective-quantile-decomposition} and the optimality of $\nu$ gives
\[
0
\leq
\Gamma_{c,{\bf P}}(\nu')-\Gamma_{c,{\bf P}}(\nu)
=
\int_0^1
\left[
\Phi_{c,u}\bigl(a_c(u)\bigr)
-
\Phi_{c,u}\bigl(Q_\nu(u)\bigr)
\right]du.
\]
The integrand is nonpositive, and thus $\Phi_{c,u}\bigl(a_c(u)\bigr)
=
\Phi_{c,u}\bigl(Q_\nu(u)\bigr)$ for a.e.~$u\in (0,1).$ 
Therefore, $Q_\nu(u)\in M_c(u) $  for a.e.~$u\in (0,1).$
Equivalently,
\[
\int_{\mathcal P_1(\mathbb R)}
\psi_c\bigl(Q_\nu(u)-Q_\mu(u)\bigr)\,d{\bf P}(\mu)
=0
\qquad\text{for a.e. }u\in(0,1).
\]
\end{proof}
\begin{proof}[Proof of \Cref{prop:inside-influence-function}]
Fix $u$ and write
$q=Q_{\nu_c}(u)$ and
$q_\varepsilon=Q_{\nu_{c,\varepsilon,\eta}}(u)$.
By \Cref{prop:huber-barycenter-1d},
\[
(1-\varepsilon)
\int
\psi_c\bigl(q_\varepsilon-Q_\mu(u)\bigr)\,d{\bf P}(\mu)
+
\varepsilon
\psi_c\bigl(q_\varepsilon-Q_\eta(u)\bigr)
=0.
\]
The corresponding equation at $\varepsilon=0$ is
$\int\psi_c(q-Q_\mu(u))\,d{\bf P}(\mu)=0$.
The uniqueness assumption and continuity of the scalar Huber objective imply
$q_\varepsilon\to q$.

For $s\neq0$, let
\[
R_s
:=
\frac1s
\int
\left[
\psi_c(q+s-Q_\mu(u))
-
\psi_c(q-Q_\mu(u))
\right]
d{\bf P}(\mu).
\]
Since $\psi_c$ is $1$-Lipschitz, $|R_s|\leq1$. Moreover, for every $\mu$ such
that $|Q_\mu(u)-q|\neq c$, the integrand divided by $s$ converges to
$\mathbbm 1_{\{|Q_\mu(u)-q|<c\}}$. Hence dominated convergence gives
$R_s\to a_c(u)$.

Subtracting the score equation at $\varepsilon=0$ from the contaminated score
equation gives
\[
(1-\varepsilon)
R_{q_\varepsilon-q}
\frac{q_\varepsilon-q}{\varepsilon}
+
\psi_c\bigl(q_\varepsilon-Q_\eta(u)\bigr)
=0.
\]
Since $a_c(u)>0$, letting $\varepsilon\downarrow0$ yields
\[
\frac{q_\varepsilon-q}{\varepsilon}
\longrightarrow
-\frac{\psi_c(q-Q_\eta(u))}{a_c(u)}
=
\frac{\psi_c(Q_\eta(u)-q)}{a_c(u)},
\]
where we used that $\psi_c$ is odd. The bound follows from
$|\psi_c|\leq c$.
\end{proof}

\begin{proof}[Proof of \eqref{eq:median-IF-limit}]
Fix $u$ and let $m(u)$ be the unique median of $Q_\mu(u)$. By
\Cref{prop:huber-barycenter-1d}, the Huber center at level $u$ converges to
$m(u)$ as $c\downarrow0$. Continuity of $f_u$ at $m(u)$ gives
\[
{\bf P}\bigl(
|Q_\mu(u)-Q_{\nu_c}(u)|<c
\bigr)
=
2c f_u(m(u))+o(c).
\]
If $Q_\eta(u)\neq m(u)$, then, for all sufficiently small $c$,
$\psi_c(Q_\eta(u)-Q_{\nu_c}(u))
=
c\,\operatorname{sgn}(Q_\eta(u)-m(u))$.
The result follows from \eqref{eq:inside-IF-formula}.
\end{proof}

\begin{proof}[Proof of \Cref{prop:inside-pointwise-clt}]
Fix $u$, write $q=Q_{\nu_c}(u)$ and
$\widehat q_n=Q_{\widehat\nu_{n,c}}(u)$, and set
\[
S_n(z)
:=
\frac1n\sum_{i=1}^n
\psi_c\bigl(z-Q_{\mu_i}(u)\bigr).
\]
By construction, $S_n(\widehat q_n)=0$. The uniqueness of the population
minimizer and the law of large numbers imply
$\widehat q_n\to q$ in probability.

Since $\psi_c$ is absolutely continuous with derivative
$\mathbbm 1_{\{|t|<c\}}$ almost everywhere,
\[
S_n(\widehat q_n)-S_n(q)
=
(\widehat q_n-q)
\int_0^1
\frac1n\sum_{i=1}^n
\mathbbm 1_{\{
|q+t(\widehat q_n-q)-Q_{\mu_i}(u)|<c
\}}
\,dt.
\]
The factor multiplying $\widehat q_n-q$ converges in probability to $a_c(u)$
by the law of large numbers, the consistency of $\widehat q_n$, and the
assumption
${\bf P}(|Q_\mu(u)-q|=c)=0$. Therefore,
\[
\sqrt n(\widehat q_n-q)
=
\frac1{a_c(u)\sqrt n}
\sum_{i=1}^n
\psi_c\bigl(Q_{\mu_i}(u)-q\bigr)
+
o_{\mathbb P}(1).
\]
The summands are bounded, have mean zero by
\eqref{eq:population-huber-score}, and have variance
$\int\psi_c(Q_\mu(u)-q)^2\,d{\bf P}(\mu)$; therefore, applying the central limit theorem proves
\eqref{eq:inside-pointwise-clt}. The finite-dimensional statement follows by
the multivariate central limit theorem.
\end{proof}

\begin{proof}[Proof of \Cref{prop:inside-outside-translations}]
For every $u\in(0,1)$,
$Q_{\mu_Z}(u)=Q_{\mu_0}(u)+Z$. Hence
\Cref{prop:huber-barycenter-1d} implies that the proposed barycenter has
quantile function $Q_{\mu_0}(u)+\theta_c$, and is therefore
$(x\mapsto x+\theta_c)_\#\mu_0$.

For the distance-based Huber mean, all random quantile functions belong to
the affine line
$\{Q_{\mu_0}+\theta:\theta\in\mathbb R\}$ in $L^2(0,1)$.
Orthogonal projection of any candidate quantile function onto this affine
line cannot increase its $L^2$ distance to any $Q_{\mu_0}+Z$. Thus the
minimization in \eqref{eq:Wasserstein-bary-huber-BAD} can be restricted to
this line. For such candidates,
\[
\mathcal W_2\bigl(
(x\mapsto x+\theta)_\#\mu_0,\mu_Z
\bigr)
=
|\theta-Z|,
\]
so the outside objective is
$\mathbb E[\rho_c(|Z-\theta|)]$, whose unique minimizer is $\theta_c$.
The empirical statement follows identically.
\end{proof}

\begin{proof}[Proof of \Cref{prop:localized-tail-influence}]
The contaminating quantile is
$Q_{\eta_{M,\delta}}(u)
=
M\mathbbm 1_{(1-\alpha,1)}(u)$.
For the proposed estimator, \Cref{prop:inside-influence-function} applies at
${\bf P}_0$ with $a_c(u)=1$. Since $M>c$,
\[
\operatorname{IF}^{\rm in}_c(\eta_{M,\alpha};{\bf P}_0)(u)
=
c\mathbbm 1_{(1-\alpha,1)}(u).
\]

For the outside estimator, let $q_\varepsilon$ be the quantile function of
the contaminated center. At ${\bf P}_0$ the center is zero. For small
$\varepsilon$, the contribution of ${\bf P}_0$ is in the quadratic part of
the Huber loss, while
$\|Q_{\eta_{M,\alpha}}\|_{L^2}=M\sqrt\alpha>c$. The first-order condition is
\[
(1-\varepsilon)q_\varepsilon
+
\varepsilon c
\frac{
q_\varepsilon-Q_{\eta_{M,\alpha}}
}{
\|q_\varepsilon-Q_{\eta_{M,\alpha}}\|_{L^2}
}
=0.
\]
Since $q_\varepsilon\to0$, division by $\varepsilon$ gives
\[
\frac{q_\varepsilon}{\varepsilon}
\longrightarrow
c\frac{
Q_{\eta_{M,\alpha}}
}{
\|Q_{\eta_{M,\alpha}}\|_{L^2}
}
=
\frac{c}{\sqrt\alpha}
\mathbbm 1_{(1-\alpha,1)}.
\]
The stated $L^2$ magnitudes follow by direct integration.
\end{proof}

\begin{proof}[Proof of \Cref{prop:two-block-efficiency}]
The condition $L_\tau>\tau+b$ guarantees that the two atoms in
\eqref{eq:two-block-random-measure} are ordered. Hence every observed quantile
function is constant on $(0,1/2]$ and on $(1/2,1)$. Projection onto the
two-dimensional space of such functions cannot increase the $L^2$ distance
to any observation, so the outside-Huber center may be computed within this
two-block class.

For the inside estimator, the lower half depends only on $U$. Since
$|U|\leq b$ and $c>2b$, the Huber loss is quadratic throughout the relevant
range, and therefore
$\widehat\theta^{\,\rm in}_{n}=n^{-1}\sum_{i=1}^nU_i$. This gives the first
central limit theorem.

For the outside estimator, write a two-block candidate as $\theta_1$ on the
lower half and $L_\tau+\theta_2$ on the upper half. Up to an irrelevant
constant factor, the first component of the estimating equation is
\[
\frac1n\sum_{i=1}^n
w_i(\theta_1,\theta_2)(\theta_1-U_i)=0,
\]
where
\[
w_i(\theta_1,\theta_2)
=
\min\left\{
1,
\frac{\sqrt2c}{
\sqrt{(\theta_1-U_i)^2+(\theta_2-\tau B_iS_i)^2}
}
\right\}.
\]
By symmetry the population center is $(0,0)$. The cross derivative of the two
population score equations and the covariance between their score components
are both zero at the origin. Thus the asymptotic variance of the lower
coordinate is the variance of its score divided by the square of the
corresponding derivative.

When $B=0$, the weight at the origin is one. When $B=1$,
$\tau>\sqrt2c$ gives
\[
w_i(0,0)
=
\frac{\sqrt2c}{\sqrt{U_i^2+\tau^2}}.
\]
Differentiating the lower score and computing its variance give
\[
A_\tau
=
(1-\varepsilon)
+
\varepsilon
\mathbb E\left[
\frac{\sqrt2c\,\tau^2}{(U^2+\tau^2)^{3/2}}
\right],
\]
and
\[
B_\tau
=
(1-\varepsilon)\sigma_U^2
+
2\varepsilon c^2
\mathbb E\left[
\frac{U^2}{U^2+\tau^2}
\right].
\]
The standard multivariate $M$-estimator expansion yields
$V_\tau=B_\tau/A_\tau^2$, which is
\eqref{eq:two-block-outside-variance}. Finally, dominated convergence gives
$A_\tau\to1-\varepsilon$ and
$B_\tau\to(1-\varepsilon)\sigma_U^2$, proving the stated variance and
efficiency limits.
\end{proof}
\section{Sinkhorn iterations for Huber--Wasserstein barycenters}
\label{app:huber-sinkhorn-barycenter}

We describe the fixed-grid Sinkhorn procedure used to compute the image Huber--Wasserstein barycenters.  The iterations are the standard entropic barycenter iterations of \cite{Cuturi.Doucet.2014.ICML}, equivalently obtained from the iterative Bregman projection formulation of entropic barycenters in \cite{Benamou.Carlier.2015.SISC}, using the Huber cost $\rho_c$ instead of the quadratic cost.

Let \(x_1,\ldots,x_m\in\mathbb R^d\), and set  \(\Delta_m:=\{a\in\mathbb R_{\geq 0}^m:\sum_i a_i=1\}\). Consider  $a^{(1)},\ldots,a^{(N)}\in\Delta_m,$ 
which we identify with discrete probability measures $P^{(1)},\dots, P^{(N)}$ supported on $\{x_1,\dots,x_m\}$ via the relation $P^{(k)}(x_j)=a^{(k)}_j$ for all $1\leq k\leq N$ and $1\leq j\leq m$.  For a fixed Huber cutoff \(c>0\), we define the discrete Huber cost matrix
\[
(C^{(c)}_{ij})_{i,j=1}^m
=
(\rho_c(\|x_i-x_j\|))_{i,j=1}^m\in \mathbb{R}_{\geq 0}^{m\times m}.
\]
\begin{definition}
For \(\varepsilon>0\), the discrete entropically regularized Huber transport cost is
\begin{equation}
    \label{eq:def_of_discrete_huber_primal_bary}
T_{c,\varepsilon}(q,a)
:=
\min_{P\in\Pi(q,a)}
\left\{
\sum_{i,j} C^{(c)}_{ij}P_{ij}
+
\varepsilon\sum_{i,j}P_{ij}(\log P_{ij}-1)
\right\}.
\end{equation}
A regularized Huber barycenter is then
\begin{equation}
\label{eq:regularised_huber_barycenter_problem}
\widehat q_{c,\varepsilon}
\in 
\argmin_{q\in\Delta_m}
\sum_{\ell=1}^N
\lambda_\ell
T_{c,\varepsilon}(q,a^{(\ell)}),
\end{equation}
where $\sum_{\ell=1}^N\lambda_\ell=1$.
\end{definition}
\begin{remark}
For the empirical computation, the additional term $-c\int \|x\|\,d\mu(x)$ 
is omitted, because it is independent of the candidate barycenter \(q\) and hence does not affect the minimizer.
\end{remark}

\begin{theorem}[Existence and uniqueness of the regularized Huber barycenter]
\label{thm:huber-sinkhorn-existence-uniqueness}
Fix \(c>0\) and \(\varepsilon>0\). Set
\[
\lambda_\ell>0,
\qquad
\sum_{\ell=1}^N\lambda_\ell=1,\quad 
a^{(1)},\ldots,a^{(N)}\in\Delta_m.
\]Then the
regularized Huber barycenter problem \eqref{eq:regularised_huber_barycenter_problem} admits a unique minimizer \(\widehat q_{c,\varepsilon}\in\mathbb{R}_{>0}^m\).
\end{theorem}

\begin{proof}
We use the convention \(0\log 0=0\). For matrices \(P\) and \(R\) with nonnegative and strictly positive entries respectively, define the generalized
Kullback--Leibler divergence
\[
\operatorname{KL}(P\mid R)
:=
\sum_{i,j}
\left[
P_{ij}\log\left(\frac{P_{ij}}{R_{ij}}\right)
-P_{ij}+R_{ij}
\right].
\]
Define the matrix $K^{(c,\varepsilon)}$ entrywise as $K_{ij}^{(c,\varepsilon)} := \exp\left(-{C_{ij}^{(c)}}/{\varepsilon}\right)$. As \(C_{ij}^{(c)}<\infty\), it follows that $K^{(c,\varepsilon)}_{ij}>0$, for $1\leq i,j\leq m$.
Furthermore,
\begin{align*}
\varepsilon\operatorname{KL}(P\mid K^{(c,\varepsilon)})
&=
\sum_{i,j}C_{ij}^{(c)}P_{ij}
+
\varepsilon\sum_{i,j}
P_{ij}\left(\log P_{ij}-1\right)
+
\varepsilon\sum_{i,j}K^{(c,\varepsilon)}_{ij}.
\end{align*}
Note that the last term is independent of \(P\); consequently, the regularized barycenter problem is equivalent, up to an additive constant, to
\begin{equation}
\label{eq:joint-kl-huber-barycenter}
\min_{\boldsymbol P\in\mathcal C_1\cap\mathcal C_2}
\sum_{\ell=1}^N
\lambda_\ell
\operatorname{KL}\left(P^{(\ell)}\mid K^{(c,\varepsilon)}\right),
\end{equation}
where
\(
\boldsymbol P
=
\left(P^{(1)},\ldots,P^{(N)}\right),
\) and
\[
\mathcal C_1
=
\left\{
\boldsymbol P:
\left(P^{(\ell)}\right)^\top\mathbf 1_m
=
a^{(\ell)}
\text{ for every }\ell
\right\}, \quad
\mathcal C_2
=
\left\{
\boldsymbol P:
P^{(1)}\mathbf 1_m
=
\cdots
=
P^{(N)}\mathbf 1_m
\right\}.
\]
Indeed, the common row marginal of any
\(\boldsymbol P\in\mathcal C_1\cap\mathcal C_2\) belongs to
\(\Delta_m\), because each \(P^{(\ell)}\) has total mass one, and the equality of both problems then follows by the definition of $\mathcal{C}_1$.

We show that the feasible set is nonempty: for any \(q_0\in\Delta_m\), we have
\[
P_0^{(\ell)}
=
q_0\left(a^{(\ell)}\right)^\top\in\mathcal C_1\cap\mathcal C_2,
\]
showing the feasible set is nonempty. The feasible set is also
compact and convex, as it is the intersection of two compact and convex sets.
 Note that the function
\[
\boldsymbol P
\longmapsto
\sum_{\ell=1}^N
\lambda_\ell
\operatorname{KL}\left(P^{(\ell)}\mid K\right)
\]
is strictly convex, since $s\longmapsto s\log s-s$ is strictly convex on \([0,\infty)\) and every \(\lambda_\ell\) is
positive. Therefore, \eqref{eq:joint-kl-huber-barycenter} has a unique
minimizer
$\boldsymbol P_\star
=
\left(P_\star^{(1)},\ldots,P_\star^{(N)}\right).$
Let \(q_\star\) denote its common row marginal.

For every fixed pair of marginals \(q,a\in\Delta_m\), the regularized
transport objective is strictly convex in \(P\). Hence the minimizer
in the definition of \(T_{c,\varepsilon}(q,a)\) is unique, even when
\(q\) or \(a\) has zero entries.

Suppose that \(q\) minimizes the barycenter objective, and let
\(P^{(\ell)}\) be the unique regularized optimal coupling between
\(q\) and \(a^{(\ell)}\). Then $\left(P^{(1)},\ldots,P^{(N)}\right)$ 
minimizes \eqref{eq:joint-kl-huber-barycenter}. By uniqueness of the
joint minimizer, \(
P^{(\ell)}
=
P_\star^{(\ell)}\) for every $\ell$, and hence \(q=q_\star\). Therefore the regularized barycenter is unique
and $q_\star=\widehat q_{c,\varepsilon}$.

It remains to show that \(q_\star\) is strictly positive. Suppose that $(q_\star)_i=0$ for some \(i\), and choose \(k\) such that \((q_\star)_k>0\). For
\(0<\theta<1\), define \(P_\theta^{(\ell)}\) by
\[
(P_\theta^{(\ell)})_{ij}
=
\theta(P_\star^{(\ell)})_{kj},
\qquad
(P_\theta^{(\ell)})_{kj}
=
(1-\theta)(P_\star^{(\ell)})_{kj},
\]
leaving all other entries unchanged. This preserves every column
marginal and gives the same perturbed row marginal for all \(\ell\);
hence the collection \(\boldsymbol P_\theta\) remains feasible.

Let
\[
\mathcal J(\boldsymbol P)
:=
\sum_{\ell=1}^N\lambda_\ell
\left[
\sum_{r,j=1}^mC_{rj}^{(c)}P_{rj}^{(\ell)}
+
\varepsilon\sum_{r,j=1}^m\phi(P_{rj}^{(\ell)})
\right],
\qquad
\phi(s)=s(\log s-1).
\]
Since only rows \(i\) and \(k\) change, the variation of the transport
cost is
\begin{align*}
&\sum_{\ell=1}^N\lambda_\ell
\sum_{j=1}^m
\left[
C_{ij}^{(c)}(P_\theta^{(\ell)})_{ij}
+
C_{kj}^{(c)}(P_\theta^{(\ell)})_{kj}
-
C_{kj}^{(c)}(P_\star^{(\ell)})_{kj}
\right] =
\theta
\sum_{\ell=1}^N\lambda_\ell
\sum_{j=1}^m
\left(C_{ij}^{(c)}-C_{kj}^{(c)}\right)
(P_\star^{(\ell)})_{kj}.
\end{align*}
For the entropy term, applying the identities
\[
\phi(\theta x)+\phi((1-\theta)x)-\phi(x)
=
x\left[
\theta\log\theta+(1-\theta)\log(1-\theta)
\right],\quad
\sum_j(P_\star^{(\ell)})_{kj}=(q_\star)_k,
\]
we obtain
\begin{align*}
\mathcal J(\boldsymbol P_\theta)
-
\mathcal J(\boldsymbol P_\star)
&=
\theta
\sum_{\ell=1}^N\lambda_\ell
\sum_j
\left(C_{ij}^{(c)}-C_{kj}^{(c)}\right)
(P_\star^{(\ell)})_{kj}
\\
&\quad+
\varepsilon(q_\star)_k
\left[
\theta\log\theta+(1-\theta)\log(1-\theta)
\right].
\end{align*}
Dividing by \(\theta\) yields
\begin{align*}
\frac{
\mathcal J(\boldsymbol P_\theta)
-
\mathcal J(\boldsymbol P_\star)
}{\theta}
&=
\sum_{\ell=1}^N\lambda_\ell
\sum_j
\left(C_{ij}^{(c)}-C_{kj}^{(c)}\right)
(P_\star^{(\ell)})_{kj}
\\
&\quad+
\varepsilon(q_\star)_k
\left[
\log\theta
+
\frac{1-\theta}{\theta}\log(1-\theta)
\right].
\end{align*}
The first term is finite, while the second tends to \(-\infty\) as
\(\theta\downarrow0\). Hence
\[
\mathcal J(\boldsymbol P_\theta)
<
\mathcal J(\boldsymbol P_\star)
\]
for sufficiently small \(\theta>0\), contradicting optimality.
Therefore \(q_\star\) is strictly positive.
\end{proof}

For fixed \(c>0\) and \(\varepsilon>0\), in the proof above we defined the matrix $K^{(c,\varepsilon)}$ entrywise as
\[
\left(K^{(c,\varepsilon)}\right)_{ij}
:=
\exp\left(-\frac{C^{(c)}_{ij}}{\varepsilon}\right)
=
\begin{cases}
\displaystyle
\exp\left(
-\frac{\|x_i-x_j\|^2}{2\varepsilon}
\right),
&
\|x_i-x_j\|\leq c,
\\[3mm]
\displaystyle
\exp\left(
-\frac{c\|x_i-x_j\|-c^2/2}{\varepsilon}
\right),
&
\|x_i-x_j\|>c.
\end{cases}
\]
This matrix $K^{(c,\varepsilon)}$ is the optimal coupling, up to row and column scalings:
\begin{lemma}[Form of the optimal entropic Huber coupling]
\label{lemma:sinkhorn-form-huber}
Fix \(q\in\Delta_m\cap \mathbb{R}^m_{>0}\), and let \(a\in\Delta_m\), possibly with zero entries. Define
\[
S(a)
:=
\left\{
j\in\{1,\ldots,m\}:a_j>0
\right\}.
\]
The unique minimizer \(P_\star\) in
\eqref{eq:def_of_discrete_huber_primal_bary} satisfies
\[
(P_\star)_{ij}=0
\qquad
\text{for every }j\notin S(a).
\]
Moreover, there exist vectors $u\in\mathbb R_{+}^m,\,
v\in\mathbb R_{\geq 0}^m$ such that
\[
v_j>0
\quad\Longleftrightarrow\quad
j\in S(a);\quad 
P_\star
=
\operatorname{diag}(u)
K^{(c,\varepsilon)}
\operatorname{diag}(v).
\]
The marginal conditions are, equivalently,
\[
u\odot K^{(c,\varepsilon)}v=q,
\qquad
v\odot
\left(K^{(c,\varepsilon)}\right)^\top u=a.
\]
\end{lemma}

\begin{proof}  We use the convention \(0\log 0=0\).
If \(a_j=0\), then every \(P\in\Pi(q,a)\) satisfies $\sum_{i=1}^mP_{ij}=0$, which implies
$P_{ij}=0$ for every $1\leq i\leq m$.
Thus the columns outside \(S(a)\) are effectively zeros which can be
ignored. 

The set of all possible couplings on $S(a)$ contains the  matrix $qa^\top$, which satisfies 
$$ \begin{cases}
    (qa^\top)_{ij}>0 & {\rm if }\  j\in S(a),\\
    (qa^\top)_{ij}=0 & {\rm if }\  j\notin S(a).
\end{cases}
$$
Now we show that the minimizer is unique and satisfies the conclusion of the result. Uniqueness follows easily from the strict convexity of the functional and the convexity of $\Pi(q,a)$. 

To show that the minimizer has strictly positive entries in each column in $S(a)$, assume {\it ad absurdum} that the minimizer $P'$ satisfies $P'_{ik}=0$ for some $k\in S(a)$ (note that we already know that $P'_{i,j}=0$ for $j\notin S(a)$ in order that $P'\in \Pi(q,a)$.) Define $P_\theta = (1-\theta) P' + \theta qa^\top\in \Pi(q,a)$. Let
\[
\mathcal E(P)
=
\langle C^{(c)},P\rangle
+
\varepsilon\sum_{i,j}\phi(P_{ij}),
\qquad
\phi(s):=s(\log s-1).
\]
Writing \(\Delta:=qa^\top-P'\), we have
\begin{align*}
\frac{\mathcal E(P_\theta)-\mathcal E(P')}{\theta}
&=
\langle C^{(c)},\Delta\rangle+
\varepsilon\sum_{P'_{ij}>0}
\frac{\phi(P'_{ij}+\theta\Delta_{ij})-\phi(P'_{ij})}{\theta}+
\varepsilon\sum_{P'_{ij}=0}
\frac{\phi(P'_{ij}+\theta\Delta_{ij})}{\theta}\\
&= \langle C^{(c)},\Delta\rangle+
\varepsilon\sum_{P'_{ij}>0}
\frac{\phi(P'_{ij}+\theta\Delta_{ij})-\phi(P'_{ij})}{\theta}+
\varepsilon\sum_{P'_{ij}=0}
[(qa^\top)_{ij}\log \theta + \phi((qa^\top)_{ij})].
\end{align*}
The first term is finite, and the second converges to
\[
\varepsilon\sum_{P'_{ij}>0}\Delta_{ij}\log P'_{ij}
\]
as $\theta \downarrow 0$. Since \((qa^\top)_{ij}>0\) and \(\{i,j: P'_{ij}=0\}\neq\varnothing\), the last term tends to
\(-\infty\) as \(\theta\downarrow0\). Therefore
\[
\lim_{\theta\downarrow0}
\frac{\mathcal E(P_\theta)-\mathcal E(P')}{\theta}
=
-\infty,
\]
contradicting the optimality of \(P'\). 

We may therefore introduce finite Lagrange multipliers
\(\alpha\in\mathbb R^m\) and
\(\beta\in\mathbb R^{|S(a)|}\) for the two reduced marginal
constraints. The first-order condition is
\[
C_{ij}^{(c)}
+
\varepsilon\log(P_\star)_{ij}
-
\alpha_i
-
\beta_j
=
0,
\qquad
j\in S(a).
\]
Hence
\[
(P_\star)_{ij}
=
\exp\left(\frac{\alpha_i}{\varepsilon}\right)
\exp\left(-\frac{C_{ij}^{(c)}}{\varepsilon}\right)
\exp\left(\frac{\beta_j}{\varepsilon}\right),
\qquad
j\in S(a).
\]
Set
\[
u_i
=
\exp\left(\frac{\alpha_i}{\varepsilon}\right),\quad 
v_j
=
\begin{cases}
\displaystyle
\exp\left(\frac{\beta_j}{\varepsilon}\right),
&
j\in S(a),
\\[2mm]
0,
&
j\notin S(a),
\end{cases}
\]
which combined with the definition of $K^{(c,\varepsilon)}$ gives
\[
P_\star
=
\operatorname{diag}(u)
K^{(c,\varepsilon)}
\operatorname{diag}(v).
\]
This implies
\[
P_\star\mathbf 1_m
=
u\odot K^{(c,\varepsilon)}v,\quad 
P_\star^\top\mathbf 1_m
=
v\odot
\left(K^{(c,\varepsilon)}\right)^\top u,
\]
which proves the marginal identities.
\end{proof}

\begin{theorem}[Convergence of the Huber--Sinkhorn iterates]
\label{thm:huber-sinkhorn-convergence}
Fix \(c>0\) and \(\varepsilon>0\), and for every $\ell \in \{ 1,\dots, N\}$, let
\[
\lambda_\ell>0,
\qquad
\sum_{\ell=1}^N\lambda_\ell=1,\quad 
a^{(1)},\ldots,a^{(N)}\in\Delta_m,\quad
S_\ell
:=
\left\{
j\in\{1,\ldots,m\}:a_j^{(\ell)}>0
\right\}.
\]
Initialize
\[
v_j^{(\ell,0)}
=
\begin{cases}
1,
&
j\in S_\ell,
\\
0,
&
j\notin S_\ell.
\end{cases}
\]
For \(t\geq0\), define
\[
g_i^{(t)}
:=
\prod_{\ell=1}^N
\left([K^{(c,\varepsilon)}v^{(\ell,t)}]_i\right)^{\lambda_\ell},
\quad
q^{(t)}
:=
\frac{g^{(t)}}
{\left\langle g^{(t)},\mathbf 1_m\right\rangle},
\]
and
\[
u^{(\ell,t)}
:=
q^{(t)}
\oslash
\left(K^{(c,\varepsilon)}v^{(\ell,t)}\right),\quad
v^{(\ell,t+1)}
:=
a^{(\ell)}
\oslash
\left(\left[K^{(c,\varepsilon)}\right]^\top u^{(\ell,t)}\right).
\]
Then all the iterates are well defined,
$q^{(t)}\in \mathbb{R}^m_{+}$ for all $t>0$,
and
\[
q^{(t)}
\longrightarrow
\widehat q_{c,\varepsilon}
\qquad
\text{as }t\to\infty.
\]
Furthermore, if
\[
P^{(\ell,t)}
:=
\operatorname{diag}\left(u^{(\ell,t)}\right)
K^{(c,\varepsilon)}
\operatorname{diag}\left(v^{(\ell,t+1)}\right),
\]
then
\[
P^{(\ell,t)}
\longrightarrow
P_\star^{(\ell)},
\qquad
\ell=1,\ldots,N,
\]
where \(P_\star^{(\ell)}\) is the unique regularized optimal coupling
between \(\widehat q_{c,\varepsilon}\) and \(a^{(\ell)}\).
\end{theorem}

\begin{proof}
For ease of notation, we set $K:= K^{(c,\varepsilon)}$ throughout, noting that both $c,\varepsilon$ are fixed. We show by induction that, for every fixed $t$, $v_j^{(\ell,t)}>0$ if and only if  $j\in S_\ell.$ 
 By construction, this holds for $t = 0$. Since \(S_\ell\neq\varnothing\) and \(K_{ij}>0\), one has
\[
\left[Kv^{(\ell,t)}\right]_i
=
\sum_{j\in S_\ell}
K_{ij}v_j^{(\ell,t)}
>0
\]
for every \(i\). Therefore, $g_i^{(t)}, q_i^{(t)} , u_i^{(\ell,t)}>0$. It follows that $ \left[K^\top u^{(\ell,t)}\right]_j>0$ for every \(j\), and hence
\[
v_j^{(\ell,t+1)}
=
\begin{cases}
\displaystyle
\frac{a_j^{(\ell)}}
{\left(K^\top u^{(\ell,t)}\right)_j}
>0,
&
j\in S_\ell,
\\[4mm]
0,
&
j\notin S_\ell.
\end{cases}
\]
By induction, all the iterations are well defined and every
\(q^{(t)}\) has strictly positive entries.

We now restrict the \(\ell\)-th iteration's transport problem to its active columns. Set
\[
K^{(\ell)}
:=
(K_{i,j})_{i\leq m, j\in S_\ell}
\in
\mathbb R_{>0}^{m\times |S_\ell|},\quad
\bar a^{(\ell)}
:=
\left(a_j^{(\ell)}\right)_{j\in S_\ell}
\in
\mathbb R_{>0}^{|S_\ell|}.
\]
Let \(\bar v^{(\ell,t)}\) denote the restriction of
\(v^{(\ell,t)}\) to \(S_\ell\). For collections $\boldsymbol P
=
\left(P^{(1)},\ldots,P^{(N)}\right)$ 
of matrices $P^{(\ell)}
\in
\mathbb R_+^{m\times |S_\ell|}$, define the weighted KL divergence
\[
\operatorname{KL}_{\boldsymbol\lambda}
\left(
\boldsymbol P
\mid
\boldsymbol R
\right)
:=
\sum_{\ell=1}^N
\lambda_\ell
\operatorname{KL}
\left(
P^{(\ell)}
\mid
R^{(\ell)}
\right).
\]
For a nonempty constraint set \(\mathcal C\), its weighted KL
projection is
\[
\operatorname{Proj}_{\mathcal C}^{
\operatorname{KL}_{\boldsymbol\lambda}}
\left(\boldsymbol R\right)
:=
\argmin_{\boldsymbol P\in\mathcal C}
\operatorname{KL}_{\boldsymbol\lambda}
\left(
\boldsymbol P
\mid
\boldsymbol R
\right).
\]
Consider the affine constraint sets
\[
\mathcal C_1
=
\left\{
\boldsymbol P:
\left(P^{(\ell)}\right)^\top\mathbf 1_m
=
\bar a^{(\ell)}
\text{ for every }\ell
\right\},\quad
\mathcal C_2
=
\left\{
\boldsymbol P:
P^{(1)}\mathbf 1
=
\cdots
=
P^{(N)}\mathbf 1
\right\}.
\]
We first compute the weighted KL projection onto \(\mathcal C_1\).
By definition, it solves
\[
\argmin_{P^{(1)},\dots, P^{(N)}} \sum_{\ell=1}^N
\lambda_\ell
\operatorname{KL}\left(P^{(\ell)}\mid R^{(\ell)}\right)
 \  \text{s.t.}\ 
\left(P^{(\ell)}\right)^\top\mathbf 1_m
=
\bar a^{(\ell)},
\
\ell=1,\ldots,N.
\]
The problem separates over \(\ell\), and the positive factor
\(\lambda_\ell\) does not affect the corresponding minimizer. For
fixed \(\ell\), introduce multipliers \(\beta_j\) for the column
constraints. The first-order condition is
\[
\log\left(
\frac{P_{ij}^{(\ell)}}{R_{ij}^{(\ell)}}
\right)
+
\beta_j
=
0,
\]
and therefore
\begin{equation}
    \label{eq:rearranged_exp_beta}
P_{ij}^{(\ell)}
=
R_{ij}^{(\ell)}e^{-\beta_j}.
\end{equation}
Imposing the \(j\)-th column constraint gives
\begin{equation}
    \label{eq:lagrange_exp_is_relation}
e^{-\beta_j}
=
{\bar a_j^{(\ell)}}/\left(
{\sum_i R_{ij}^{(\ell)}}\right),
\end{equation}
and hence, plugging \eqref{eq:lagrange_exp_is_relation} into \eqref{eq:rearranged_exp_beta} and expressing the result in matrix form gives
\[
\left(
\operatorname{Proj}_{\mathcal C_1}^{
\operatorname{KL}_{\boldsymbol\lambda}}
\left(\boldsymbol R\right)
\right)^{(\ell)}
=
R^{(\ell)}
\operatorname{diag}\left(
\bar a^{(\ell)}
\oslash
\left(
\left(R^{(\ell)}\right)^\top\mathbf 1_m
\right)
\right).
\]
Strict convexity of the KL divergence shows that this projection is
unique. We next compute the weighted KL projection onto \(\mathcal C_2\).
By definition, it solves
\[
\argmin_{P^{(1)},\dots, P^{(N)}} \sum_{\ell=1}^N
\lambda_\ell
\operatorname{KL}\left(P^{(\ell)}\mid R^{(\ell)}\right)
 \  \text{s.t.}\ 
P^{(\ell)}\mathbf 1
=
P^{(k)}\mathbf 1,
\
\ell\neq k.
\]
Denote the common marginal by \(p\), so that
$P^{(\ell)}\mathbf 1=p,$ for  $\ell=1,\ldots,N.$ 
For fixed \(p\), the minimization separates over \(\ell\) and over the
rows. Let $r^{(\ell)}
:=
R^{(\ell)}\mathbf 1.$ 
Introducing a multiplier for the \(i\)-th row constraint gives
\[
\log\left(
{P_{ij}^{(\ell)}}/{R_{ij}^{(\ell)}}
\right)
+\alpha_i^{(\ell)}
=0.
\]
Therefore, $P_{ij}^{(\ell)}
=
R_{ij}^{(\ell)}e^{-\alpha_i^{(\ell)}}.$ 
Imposing the common marginal constraint yields $e^{-\alpha_i^{(\ell)}}
=
\frac{p_i}{r_i^{(\ell)}},$ 
and hence $P^{(\ell)}
=
\operatorname{diag}\left(
p\oslash r^{(\ell)}
\right)R^{(\ell)}.$ Substituting this expression into the objective, the terms depending
on \(p_i\) are
\[
\sum_{\ell=1}^N
\lambda_\ell
\left[
p_i\log\left(\frac{p_i}{r_i^{(\ell)}}\right)
-p_i+r_i^{(\ell)}
\right].
\]
Differentiating with respect to \(p_i\) gives
\[
\sum_{\ell=1}^N
\lambda_\ell
\log\left(\frac{p_i}{r_i^{(\ell)}}\right)
=0.
\]
Since $\sum_{\ell=1}^N\lambda_\ell=1$, we obtain
\[
p_i
=
\exp\left(
\sum_{\ell=1}^N
\lambda_\ell\log r_i^{(\ell)}
\right)
=
\prod_{\ell=1}^N
\left(r_i^{(\ell)}\right)^{\lambda_\ell}=\prod_{\ell=1}^N
\left(r_i^{(\ell)}\right)^{\lambda_\ell}.
\]
Thus, defining
\[
g_i
:=
\prod_{\ell=1}^N
\left(r_i^{(\ell)}\right)^{\lambda_\ell}=\prod_{\ell=1}^N
\left([R^{(\ell)}\mathbf 1]_i\right)^{\lambda_\ell},
\]
the weighted KL projection onto \(\mathcal C_2\) is
\[
\left(
\operatorname{Proj}_{\mathcal C_2}^{
\operatorname{KL}_{\boldsymbol\lambda}}
\left(\boldsymbol R\right)
\right)^{(\ell)}
=
\operatorname{diag}\left(
g\oslash r^{(\ell)}
\right)R^{(\ell)}=\operatorname{diag}\left(
g\oslash (R^{(\ell)}\mathbf 1)
\right)R^{(\ell)}.
\]
Strict convexity of the KL objective shows that this projection is
unique. Notice that \(g\) need not have total mass one, since
\(\mathcal C_2\) imposes equality of the row marginals, but does not their normalization.

Starting from $R^{(\ell,0)}
=
K^{(\ell)},$ 
alternating the projections onto \(\mathcal C_2\) and
\(\mathcal C_1\) gives the unnormalized iterations
\[
\widetilde u^{(\ell,t)}
=
g^{(t)}
\oslash
\left(K^{(\ell)}\bar v^{(\ell,t)}\right),\quad
\widetilde v^{(\ell,t+1)}
=
\bar a^{(\ell)}
\oslash
\left(
\left(K^{(\ell)}\right)^\top
\widetilde u^{(\ell,t)}
\right).
\]
The corresponding full-cycle coupling is
\[
\widetilde P^{(\ell,t)}
=
\operatorname{diag}\left(\widetilde u^{(\ell,t)}\right)
K^{(\ell)}
\operatorname{diag}\left(\widetilde v^{(\ell,t+1)}\right).
\]

The cyclic Bregman projection theorem for affine constraint sets
implies that $(
\widetilde P^{(1,t)},\ldots,
\widetilde P^{(N,t)}
)$ converges to the unique weighted KL projection of
$\left(
K^{(1)},\ldots,K^{(N)}
\right)$ onto \(\mathcal C_1\cap\mathcal C_2\); see
\cite[Section 2.1]{Benamou.Carlier.2015.SISC}. Equivalently, the limit is the
unique solution of
\[
\min_{\boldsymbol P\in\mathcal C_1\cap\mathcal C_2}
\sum_{\ell=1}^N
\lambda_\ell
\operatorname{KL}
\left(
P^{(\ell)}
\mid
K^{(\ell)}
\right).
\]
By
\cref{thm:huber-sinkhorn-existence-uniqueness}, this solution exists, is unique and has common row marginal
\(\widehat q_{c,\varepsilon}\).

It remains to compare the unnormalized projection iteration with the
normalized iteration in the statement. Let
$Z_t
:=
\left\langle g^{(t)},\mathbf 1_m\right\rangle.$ 
Since $q^{(t)}
=
\frac{g^{(t)}}{Z_t},$ 
one has $u^{(\ell,t)}
=
\frac{1}{Z_t}
\widetilde u^{(\ell,t)}.$ 
Consequently, $v^{(\ell,t+1)}
=
Z_t\widetilde v^{(\ell,t+1)}.$ 
The scalar factors ($Z_t$) cancel in the coupling:
\begin{align*}
&
\operatorname{diag}\left(u^{(\ell,t)}\right)
K^{(\ell)}
\operatorname{diag}\left(\bar v^{(\ell,t+1)}\right) =
\operatorname{diag}\left(\widetilde u^{(\ell,t)}\right)
K^{(\ell)}
\operatorname{diag}\left(\widetilde v^{(\ell,t+1)}\right).
\end{align*}
Therefore, the normalized and unnormalized algorithms generate the
same sequence of coupling matrices. It follows that $P^{(\ell,t)}
\longrightarrow
P_\star^{(\ell)}$ for every $\ell$.  

The intermediate projection onto \(\mathcal C_2\) has common row
marginal \(g^{(t)}\). Convergence of the alternating Bregman
projections therefore gives $g^{(t)}
\longrightarrow
\widehat q_{c,\varepsilon}.$ 
Since $\left\langle
\widehat q_{c,\varepsilon},
\mathbf 1_m
\right\rangle
=
1,$ 
we also have
\[
Z_t
=
\left\langle g^{(t)},\mathbf 1_m\right\rangle
\longrightarrow
1.
\]
Thus
\[
q^{(t)}
=
\frac{g^{(t)}}{Z_t}
\longrightarrow
\widehat q_{c,\varepsilon}.
\]
Finally, extending the reduced coupling matrices by zero in the
columns outside \(S_\ell\) gives convergence in the original
\(m\times m\) representation.
\end{proof}

\begin{remark}[Zero columns]
If \(a_j^{(\ell)}=0\), then the \(j\)-th column of every feasible
coupling between \(q\) and \(a^{(\ell)}\) is identically zero.
Accordingly, the associated Sinkhorn scaling satisfies $v_j^{(\ell,t)}=0$ 
at every iteration $t$. These zero columns should not be replaced by
a positive numerical floor unless one intentionally wishes to modify
the input histogram. By contrast, for fixed \(\varepsilon>0\), every barycenter iterate
\(q^{(t)}\) and the limiting barycenter
\(\widehat q_{c,\varepsilon}\) have strictly positive entries.
\end{remark}

\begin{remark}
The theorem does not assert convergence of the individual scaling
vectors \(u^{(\ell,t)}\) and \(v^{(\ell,t)}\). These vectors have the reciprocal scalar indeterminacy:
\[
\left(
u^{(\ell,t)},v^{(\ell,t)}
\right)
\longmapsto
\left(
\alpha^{-1}u^{(\ell,t)},
\alpha v^{(\ell,t)}
\right),
\]
which leaves the associated coupling unchanged. The intrinsic
quantities are the barycenter vectors and coupling matrices, whose
convergence is established above.
\end{remark}

We can use \cref{thm:huber-sinkhorn-convergence} to implement an algorithm similar to the ones in \cite{Cuturi.Doucet.2014.ICML, Benamou.Carlier.2015.SISC}.

\begin{algorithm}[H]
\caption{Huber--Sinkhorn barycenter iterations}
\label{alg:huber-sinkhorn-barycenter}
\begin{algorithmic}[1]
\State \textbf{Input:} vectors \(a^{(1)},\ldots,a^{(N)}\in\Delta_m\), weights \(\lambda_\ell\), points \(\{x_1,\ldots,x_m\}\subset \mathbb{R}^d\), Huber parameter \(c>0\), entropic regularization parameter \(\varepsilon>0\).
\State Compute the Huber cost and kernel
\[
C^{(c)}_{ij}=\rho_c(\|x_i-x_j\|),\quad 
K^{(c,\varepsilon)}=\exp(-C^{(c)}/\varepsilon).
\]
\State Initialize \(v^{(\ell)}=\mathbf 1_{S^{(\ell)}}\), \(\ell=1,\ldots,N\), where $S^{(\ell)}=\{i: a^{(\ell)}_i>0\}$.
\For{\(t=0,\ldots,T-1\)}
    \State For each \(\ell\), compute
    \[
    r^{(\ell)}=K^{(c,\varepsilon)}v^{(\ell)}.
    \]
    \State Update the barycenter:
    \[
    \widetilde q_i
    =
    \exp\left(
    \sum_{\ell=1}^N
    \lambda_\ell\log r_i^{(\ell)}
    \right),
    \qquad
    q_i
    =
    \frac{\widetilde q_i}{\sum_{k=1}^m \widetilde q_k}.
    \]
    \State For each \(\ell\), update the Sinkhorn scalings:
    \[
    u^{(\ell)}
    =
    \frac{q}{K^{(c,\varepsilon)}v^{(\ell)}},
    \qquad
    v^{(\ell)}
    =
    \frac{a^{(\ell)}}{(K^{(c,\varepsilon)})^\top u^{(\ell)}}.
    \]
\EndFor
\State \textbf{Output:} \(q\).
\end{algorithmic}
\end{algorithm}

In the algorithm above, all divisions are componentwise. In the implementation, the updates are evaluated with small numerical floors in the nonzero entries, and, for smaller values of \(\varepsilon\), in the log domain using log-sum-exp stabilization.  The stopping criterion monitors both the change in \(q\) and the marginal residuals of
\[
P^{(\ell)}
=
\operatorname{diag}(u^{(\ell)})
K^{(c,\varepsilon)}
\operatorname{diag}(v^{(\ell)}).
\]
Furthermore, as our goal is to compute the Huber-Wasserstein barycenter of barycenter of the image collection, we take $\lambda_\ell = 1/N$ for all $1\leq \ell \leq N$.

\end{document}